\documentclass[acmsmall,nonacm]{acmart}

\RequirePackage{./sty}
\fulltrue

\title{Pure Borrow: Linear Haskell Meets Rust-Style Borrowing}
\titlenote{This is an extended version of the PLDI 2026 paper \cite{MatsushitaIshii26-Pure-Borrow} with an appendix on the formal calculus and metatheory.}

\setcopyright{cc}
\setcctype{by}
\acmJournal{PACMPL}

\author{Yusuke Matsushita}
\orcid{0000-0002-5208-3106}
\affiliation{%
  \institution{Kyoto University}
  \city{Kyoto}
  \country{Japan}
}
\email{ymat@fos.kuis.kyoto-u.ac.jp}

\author{Hiromi Ishii}
\orcid{0000-0002-7752-1782}
\affiliation{%
  \institution{JIJ Inc.}
  \city{Tokyo}
  \country{Japan}
}
\email{konn.jinro@gmail.com}

\setcitestyle{nosort}



\begin{document}

\begin{abstract}
A promising approach to unifying functional and imperative programming paradigms is to localize mutation using linear or affine types.
Haskell, a purely functional language, was recently extended with linear types by Bernardy et al., in the name of Linear Haskell.
However, it remained unknown whether such a pure language could safely support non-local \emph{borrowing} in the style of Rust, where each borrower can be freely split and dropped without direct communication of ownership back to the lender.

We answer this question affirmatively with \emph{Pure Borrow}, a novel framework that realizes Rust-style borrowing in Linear Haskell with purity.
Notably, it features parallel state mutation with affine mutable references inside pure computation, unlike the IO and ST monads and existing Linear Haskell APIs.
It also enjoys purity, lazy evaluation, first-class polymorphism and leak freedom, unlike Rust.
We implement Pure Borrow simply as a library in Linear Haskell and demonstrate its power with a case study in parallel computing.
We formalize the core of Pure Borrow and build a metatheory that works toward establishing safety, leak freedom and confluence, with a new, history-based model of borrowing.
\end{abstract}

\maketitle

\section{Introduction}
\label{sect:intro}

Software is the art of embodying human ideals on finite physical machines.
This tension naturally leads to two major programming paradigms.
One is \emph{functional programming}, more on the side of humans.
It offers high-level abstractions and safety for smooth human reasoning.
The other is \emph{imperative programming}, more geared toward machines.
It offers low-level stateful operations, enabling great performance.
Unifying these two programming paradigms has been a long-standing, central topic in programming language research.

\newcommand*{\Fstar}{F\({}^\star\)}
Functional programming favors \emph{purity}, or the property that the result of the computation is always the same regardless of the timing and context.
Purity allows one to safely forget the physical machine state and reason about behaviors just as in mathematics.
It enables safe lazy evaluation, safe concurrency, and deep program optimizations such as fusion.
Moreover, a number of languages---e.g., Haskell \cite{HudakPWBFFGHHJKNPP92-Haskell}, \Fstar \cite{SwamyHKRDFBFSKZ16-Fstar}, and Lean \cite{MouraKADR15-Lean}---rigorously ensure that \emph{every computation is pure}, unless marked otherwise.
Especially in the presence of concurrency, purity is really helpful because it eliminates tricky, subtle concurrency bugs that occur only with low probability depending on thread scheduling.

Purity is thus desired, but naïvely it causes significant runtime overhead to keep all data persistently.
This calls for unifying pure functional programming with performative, stateful imperative programming.
A long-known baseline approach to this is the ST monad~\cite{LaunchburyJ94-ST}.
The monad \hsinline{ST^s a} allows efficient stateful computation, but its result can still be taken out into pure computation, unlike the IO monad \cite{PeytonJonesW93-IO}.

\begin{table}
  \caption{
    High-level comparison of our Pure Borrow with existing approaches.
    Yes \yes{} / No \no.
    \emph{Mutate}: Supports destructive mutation.
    \emph{Pure}: The result can be safely taken out in pure computation.
    \emph{Parallel}: Supports safe parallel execution.
    \emph{Aff.\,mut.}: Supports affine mutable references (\yesno{} for droppable but copyable mutable references).
    \emph{Leak-free}: Guarantees the absence of resource leaks.
  }
  \label{table:high-level-compare}
  \begin{center}
  \setlength{\tabcolsep}{3pt}
  \begin{tabular}{r@{\hskip5pt}|@{\hskip5pt}ccccc}
    & \emph{Mutate} & \emph{Pure} & \emph{Parallel} & \emph{Aff.\,mut.} & \emph{Leak-free} \\ \hline
    Pure computation only & \no & \yes & \yes & - & - \\
    IO monad \cite{PeytonJonesW93-IO} & \yes & \no & \yes & \yesno & \no \\
    ST monad \cite{LaunchburyJ94-ST} & \yes & \yes & \no & \yesno & \no \\
    Rust \cite{MatsakisK14-Rust} & \yes & \no & \yes & \yes & \no \\
    Pure Linear Haskell \cite{BernardyBNJS18-Linear-Haskell} & \yes & \yes & \yes & \no & \yes \\
    RIO monad \cite{BernardyBNJS18-Linear-Haskell} & \yes & \no & \yes & \yesno & \yes \\
    \textbf{Our Pure Borrow: BO monad} & \yes & \yes & \yes & \yes & \yes
  \end{tabular}
  \end{center}
\end{table}

Another promising approach to this goal is to localize mutation using \emph{linear} or \emph{affine types} \cite{Wadler90-linear-type, TurnerWM95-once-upon, WansbroughJ99-once-upon-poly}, requiring resources to be used exactly once (linear) or up to once (affine), inspired by \citet{Girard87-linear-logic}'s linear logic.
Linear or affine types can localize the effect of state mutation, especially for safety.
Also, linearity can prevent resource leaks.
Haskell, a purely functional language, was recently extended with linear types by \citet{BernardyBNJS18-Linear-Haskell}, in the name of Linear Haskell.
It achieves safe state mutation inside pure computation, while supporting concurrency and ensuring leak freedom unlike the ST monad.
It also introduced a leak-free variant of the IO monad, called the RIO monad.

However, it remained unknown whether such a pure language as Linear Haskell could safely support \emph{borrowing} in the sophisticated style of Rust~\cite{MatsakisK14-Rust}, a highly successful imperative language with affine types.
While Rust-style borrowing originates from the long line of work on region types \cite{TofteT94-region-type, FahndrichD02-adoption, GrossmanMJHWC02-Cyclone-region, FluetMA06-region, HallerO10-borrow},
it features a distinct form of \emph{non-locality}:
each borrower can be freely split and dropped (i.e., affine), and its ownership is non-locally returned to the lender without direct communication.
This flexibility offers an advantage over the existing Linear Haskell APIs, which require direct threading of ownership.

We answer the question above affirmatively with this work, \emph{Pure Borrow}.
Our contributions are summarized as follows:
\begin{itemize}
\item
  We develop \emph{Pure Borrow}, a novel framework that realizes Rust-style borrowing in Linear Haskell with purity. (\cref{sect:overview}).
  \Cref{table:high-level-compare} shows a high-level comparison with existing approaches.
  Notably, Pure Borrow features parallel state mutation with affine mutable references inside pure computation, unlike the IO and ST monads and existing Linear Haskell APIs.
  It also enjoys purity, lazy evaluation, first-class polymorphism and leak freedom, unlike Rust.
\item
  We implement Pure Borrow simply as a library in Linear Haskell and demonstrate its power by a case study on parallel quicksort (\cref{sect:casestudy}).
  Notably, our implementation compiles with the current Haskell compiler GHC 9.10+, without any modification to the compiler.
\item
  We formalize the core of Pure Borrow and build a metatheory that works toward establishing safety, leak freedom and confluence, with a new, history-based model of borrowing (\cref{sect:metatheory}).
\end{itemize}
We also explain the background (\cref{sect:background}), compare Pure Borrow with existing approaches (\cref{sect:compare}), and discuss related and future work (\cref{sect:related}).
Our Haskell implementation of Pure Borrow and our benchmark suite are archived at \cite{MatsushitaIshii26-Pure-Borrow-artifact}, and the latest version is maintained at \url{https://github.com/SoftwareFoundationGroupAtKyotoU/pure-borrow}.
\iffull\else
The extended version of this paper with an appendix on our formal calculus and metatheory is available as \cite{MatsushitaIshii26-Pure-Borrow-arXiv}.
\fi

\section{Background}
\label{sect:background}

Before diving deeply into our work, Pure Borrow, we briefly review prior approaches to local state mutation inside pure computation in Haskell, namely the ST monad and Linear Haskell.

\subsection{ST Monad}
\label{sect:background-st}

The classical approach to achieve local mutation in Haskell is the \hsinline{ST} monad \cite{LaunchburyJ94-ST}, as mentioned in \cref{sect:intro}.
\vspace{-.7em}
\begin{figure}[H]
  \raggedright
  \begin{minipage}{.45\textwidth}
\begin{hscode}
data ST^s a
instance Monad ST^s
runST :: (∀s. ST^s a) -> a

\end{hscode}\vspace{-.3em}
\begin{hscode}
data Vector_ST^s a
newVector_ST :: [a] -> ST^s (Vector_ST^s a)
elems_ST    @\base@:: Vector_ST^s a -> ST^s [a]
writeAt_ST  @\fit@:: Int -> a -> Vector_ST^s a -> ST^s ()
readAt_ST   @\fit@:: Int -> Vector_ST^s a -> ST^s a
\end{hscode}
  \end{minipage}\hspace{1.0em}
  \begin{minipage}{.5125\textwidth}
\begin{hscode}
example :: (Int, [Int])
example = runST do
    vec <- newVector_ST [0, 1, 2]
    _modifyAt_ST 0 (+ 3) vec;  _modifyAt_ST 2 (+ 5) vec
    _modifyAt_ST 0 (* 4) vec;  n <- readAt_ST 0 vec
    _pure (n, elems_ST vec)
\end{hscode}\vspace{-.3em}
\begin{hscode}
_modifyAt_ST i f vec = do
    a <- readAt_ST i vec;  writeAt_ST i (f a) vec
\end{hscode}
  \end{minipage}
  \caption{ST monad APIs and an example usage.}
  \label{fig:st-monad}
\end{figure}

\Cref{fig:st-monad} summarizes APIs and an example usage of the ST monad.
The monad \hsinline{ST^s a} is parameterized with a \emph{region} parameter \hsinline{s}.
In the ST monad, one can freely use mutable state such as a vector \hsinline{Vector_ST^s a}\footnote{
  The term `vector' is commonly used in Haskell for a slice to an array.
  It amounts to the slice type \rsinline{[T]} in Rust.
}, bound to some region \hsinline{s}.
In other words, a region \hsinline{s} represents memory locations that are allocated and kept during \hsinline{ST^s}.
Most notable is the function \hsinline{runST}, which can strip off the ST monad into a pure computation.
Its key trick is to use \emph{rank-2 polymorphism} over the region variable \hsinline{s} (technically working as a \emph{skolem}).
This aptly encapsulates a fresh region \hsinline{s} for mutable state inside \hsinline{ST^s} and prohibits it from leaking outside.

Although the ST monad is widely used in Haskell, it has a major limitation: it does not support parallelism.
The ST monad allows shared mutable state like \hsinline{Vector_ST^s a}.
Parallelism would allow multiple threads to mutate such state.
The result of computation could nondeterministically depend on thread scheduling, making \hsinline{runST}, the embedding to pure computation, unsound.

\subsection{Linear Haskell Before Pure Borrow}
\label{sect:background-linear}

A recent groundbreaking take on purity and mutation is \emph{Linear Haskell}~\cite{BernardyBNJS18-Linear-Haskell}, which retrofits a linear type system~\cite{Wadler90-linear-type} onto Haskell.
Notably, it is implemented as a practical, stable extension called \hsinline{LinearTypes} in the Glasgow Haskell Compiler (GHC).

Linear Haskell simply adds the \emph{linear} arrow type \hsinline{a -o b}, along with the original unrestricted arrow type \hsinline{a -> b}.
Roughly speaking, the linear arrow \hsinline{a -o b} enforces that the argument \hsinline{a} is used exactly once in \hsinline{b}.\footnote{
  Due to lazy evaluation, this is actually more subtle.
  \citet{BernardyBNJS18-Linear-Haskell} describe the condition more precisely as follows:
  whenever the output \hsinline{b} is consumed exactly once, the input \hsinline{a} is consumed exactly once.
}
This linearity is the key to safely introducing mutable state.

\begin{figure}
  \raggedright
  \begin{minipage}{.625\textwidth}
\begin{hscode}
data Ur a where Ur :: a -> Ur a
class Linearly;  linearly :: (Linearly =o Ur a) -o Ur a
class Movable a where move :: a -o Ur a
\end{hscode}\vspace{-.3em}
\begin{hscode}
data Vector a
newVector @\base@:: Linearly =o [a] -o Vector a
freeVector @\fit@:: Vector a -o [a]
modifyAt_  @\fit@:: Int -> (a -o a) -o Vector a -o Vector a
readAt_ :: Movable a => Int -> Vector a -o (Ur a, Vector a)
\end{hscode}
  \end{minipage}\hspace{.5em}
  \begin{minipage}{.325\textwidth}
\begin{hscode}
example :: Ur (Int, [Int])
example = linearly _do
    vec <- newVector [0, 1, 2]
    vec <- modifyAt_ 0 (+ 3) vec
    vec <- modifyAt_ 2 (+ 5) vec
    vec <- modifyAt_ 0 (* 4) vec @\label{line:modify-lh}@
    (Ur n, vec) <- readAt_ 0 vec @\label{line:read-lh}@
    move (n, freeVector vec)
\end{hscode}
  \end{minipage}
\begin{hscode}
class Consumable a where consume :: a -o ()
class LinearOnly a;    withLinearly :: LinearOnly a => a -o (Linearly =o a -o r) -o r
par :: a -o b -o (a, b)
\end{hscode}
  \caption{Core Linear Haskell APIs and an example usage before Pure Borrow.}
  \label{fig:linear-haskell}
\end{figure}

\Cref{fig:linear-haskell} shows core APIs and an example usage of Linear Haskell.
We feature here the \emph{linear} mutable vector type \hsinline{Vector a}.
Most notable is the \hsinline{modifyAt_} operator.
A call \hsinline{modifyAt_ i f vec} performs an \emph{in-place mutation} of the vector \hsinline{vec}, applying \hsinline{f} to the \hsinline{i}-th element, and outputs the updated vector.
Notably, unlike the traditional purely functional style, the old value of the vector can be safely lost.
This is because there is no other alias to the original \hsinline{vec}, thanks to the \emph{linearity} of \hsinline{Vector a}.
Technically, this is guaranteed by the \emph{linear} constraint \hsinline{Linearly =o} in the type of the \hsinline{newVector} operator creating a new vector, as we explain later.

\paragraph{Pure do-notation}

Throughout the paper, we use \hsinline{_do}-notation for pure computations, analogous to the usual \hsinline{do}-notation for monads.
Roughly speaking, the monadic binding \hsinline{... <- ...} in \hsinline{_do} works like the \hsinline{let}-binding \hsinline{let ... = ... in}, but the former is often a better fit in Linear Haskell, because it supports \emph{shadowing} and adopts \emph{strict} pattern matching, unlike the latter.\footnote{
  In Linear Haskell, shadowing is useful because it is common to thread through the same resource linearly, and strict pattern matching is often required to satisfy linearity.
  The \hsinline{let}-binding in Haskell disallows shadowing to allow recursion such as \hsinline{let ones = 1 :: ones}, just like \hsinline{let rec} in ML, and adopts lazy pattern matching.
}
Practically, we can use this kind of custom \hsinline{do}-notation by the \hsinline{QualifiedDo} extension.

\paragraph{Opting out linearity}

We introduce a special data type \hsinline{Ur a} (`Ur' stands for `Unrestricted').
It means that its content \hsinline{a} can be used any number of times (i.e., be freely shared and dropped) even under linear binding, analogous to the \(!\)-modality in linear logic.
The type class \hsinline{Movable a} abstracts data types that can be moved to an unrestricted context, i.e., put under \hsinline{Ur}.
Some primitive types such as \hsinline{()}, \hsinline{Bool} and \hsinline{Int} are \hsinline{Movable}.
Also, containers such as the list \hsinline{[a]} and the pair \hsinline{(a, b)} are \hsinline{Movable} if the component types \hsinline{a}, \hsinline{b} are \hsinline{Movable}.
The type class \hsinline{Consumable a} abstracts data types that can be consumed (but not necessarily be shared) under linear binding.

\paragraph{Linearly}

In this paper, we use a special type class called \hsinline{Linearly}, proposed by \citet{SpiwackKBWE22-linear-constraint}.
It is introduced as a \emph{linear constraint} \hsinline{Linearly =o}, which is a linear version of a usual type class constraint such as \hsinline{Ord a =>}.
Essentially, \hsinline{=o} is just a variant of \hsinline{-o} where the argument is managed implicitly (just like \hsinline{=>} is a variant of \hsinline{->} with an implicit argument).

The type class \hsinline{Linearly} serves as a \emph{witness of linearity}, or a guarantee that the current computation is inside a \emph{linear} thunk introduced by a special operator \hsinline{linearly}.
Functions that create new mutable state, such as \hsinline{newVector} for \hsinline{Vector a}, typically take \hsinline{Linearly}, because the output mutable state captures \hsinline{Linearly} and thus is used linearly.
At a high level, the operator \hsinline{linearly} is analogous to the ST monad's \hsinline{runST}:
it embeds computation with mutable state into a pure context.
It also introduces \hsinline{Linearly}, somewhat like how \hsinline{runST} introduces \hsinline{s}.

To see how \hsinline{Linearly} serves for safety, first imagine a fake API \hsinline{newVector :: [a] -o Vector a} without \hsinline{Linearly} constraint.
We can weaken the linear arrow of \hsinline{newVector} to a usual arrow to get \hsinline{newVector :: [a] -> Vector a}.\footnote{
  Recall that the linear arrow \hsinline{-o} is a promise by the \emph{callee} to use the argument once, not a request to the \emph{caller}.
}
Combined with \hsinline{freeVector :: Vector a -> [a]}, it easily causes a dangerous \emph{double free}, violating memory safety.
The linearity witness \hsinline{Linearly} solves this.
The only way to (newly) introduce \hsinline{Linearly} is by calling \hsinline{linearly :: (Linearly =o Ur a) -o Ur a}.
It requests a \emph{linear} function \hsinline{Linearly =o Ur a} that linearly takes \hsinline{Linearly}.
By letting \hsinline{newVector} take \hsinline{Linearly =o}, the resulting \hsinline{Vector a} captures a linear resource \hsinline{Linearly} and thus will not be shared unrestrictedly.
See \cite{SpiwackKBWE22-linear-constraint, SpiwackKBWE26-linear-constraint-arXiv} for more details.

For utility, we also have the \hsinline{LinearOnly} type class, which abstracts over data types that can be introduced only inside linear contexts, such as \hsinline{Vector a}.
In the presence of such data types, we can safely get \hsinline{Linearly}, using the \hsinline{withLinearly} operator.

The type class \hsinline{Linearly} is treated very specially as a linear constraint \hsinline{=o}: it can be freely finitely duplicated or discarded (but not unrestrictedly shared), automatically by the type system.\footnote{
  As of now, linear constraints and \hsinline{Linearly} have not yet landed in GHC.
  Still, we can already get the same functionality by introducing \hsinline{Linearly} as a linear \emph{type} (using \hsinline{Linearly -o} instead of \hsinline{Linearly =o}), manually handling instances of \hsinline{Linearly}.
  Our Haskell implementation of Pure Borrow does this, and thus it compiles with GHC 9.10+.
  Our metatheory (\cref{sect:metatheory}) also follows this style for simplicity.
  In this paper, we use linear constraints and \hsinline{Linearly} in (informal) Haskell code just as a useful automation just for the brevity of presentation.
}

\paragraph{Parallelism}

Pure Linear Haskell supports parallel execution, like pure non-linear Haskell.
For example, \hsinline{par a b} evaluates \hsinline{a} and \hsinline{b} in parallel and combines their outputs.
The point is that the linear arrow \hsinline{-o} of \hsinline{par} ensures that the mutable resources owned by \hsinline{a} and \hsinline{b} are disjoint.
This is unlike the \hsinline{ST} monad, which does not support parallelism lacking such knowledge of disjointness.

\paragraph{On the example}

The \hsinline{example} in \cref{fig:linear-haskell} corresponds to that of \cref{fig:st-monad}.
As expected, it type-checks and evaluates to \hsinline{Ur (12, [12, 1, 7])}.
Still, threading of \hsinline{vec} is quite in the way.
Although the order of updates to the 0th element matters and they cannot be swapped, updates to the 0th element and the 2nd element can be swapped and should be able to be \emph{parallelized}.
Reads can also be swapped.
Threading of the whole vector \hsinline{vec} interrupts such natural swapping.
This motivates our work, Pure Borrow, and we come back to this point later in \cref{sect:overview-taste}.

\section{Overview of Pure Borrow}
\label{sect:overview}

We now provide an overview of our Pure Borrow framework.
We first give a taste of it (\cref{sect:overview-taste}), showing how it can parallelize the previous example \cref{fig:linear-haskell}, and then present its details (\cref{sect:overview-details}).

\subsection{Taste of Pure Borrow}
\label{sect:overview-taste}

\paragraph{Rust-style borrowing}

Pure Borrow features Rust-style \emph{borrowing} \cite{FahndrichD02-adoption, GrossmanMJHWC02-Cyclone-region, FluetMA06-region, HallerO10-borrow, MatsakisK14-Rust} safely and purely achieved in Linear Haskell.
But what makes Rust-style borrowing so special?

\begin{figure}
  \centering
  \pgfdeclarelayer{background}
  \pgfdeclarelayer{foreground}
  \pgfsetlayers{background,main,foreground}
  \begin{tikzpicture}[
    node distance=2.5cm]
    \begin{pgfonlayer}{foreground}
      \node[resource, circle, inner sep=7pt] (res-original) {};
      \node[draw=none, left=10pt of res-original] (padding-original) {};
      \node[resource, circle, inner sep=7pt, below right=1.5cm and 1.25cm of res-original] (res-borrow) {};
      \matrix[matrix of nodes, right=5pt of res-borrow, outer sep=0pt] (nodes-middle){
        |(res-middle-subbor1)[resource, triangle, rotate=90, inner sep=2pt]| \\
        |(res-middle-subbor2)[resource, triangle, rotate=-90, inner sep=2.25pt]| \\
      };
      \matrix (nodes-reclaimed) [matrix of nodes, right=7.7cm of res-original]{
        |[resource, star, inner sep=3pt] (res-subbor1-final)| &[4pt] \\[1pt]
        |[resource, diamond, inner sep=3.25pt] (res-subbor2-final)| \\
      };
      \node[draw=none, right=10pt of nodes-reclaimed] (padding-reclaimed) {};
    \end{pgfonlayer}

    \begin{pgfonlayer}{main}
      \node[resource group, minimum width=22pt, minimum height=32pt, fit={(nodes-middle)}] (group-initial) {};
      \node[resource group, minimum width=22pt, minimum height=32pt, fit={(res-subbor1-final) (res-subbor2-final)}] (group-final) {};
    \end{pgfonlayer}

    \begin{pgfonlayer}{foreground}
      \node[above right=1.5mm and 7.5mm of group-initial.east,
        label={[yshift=-2.2mm]above:{\hsinline{Mut^α U_1}}}]
        (subbor1) {\includegraphics[width=1cm]{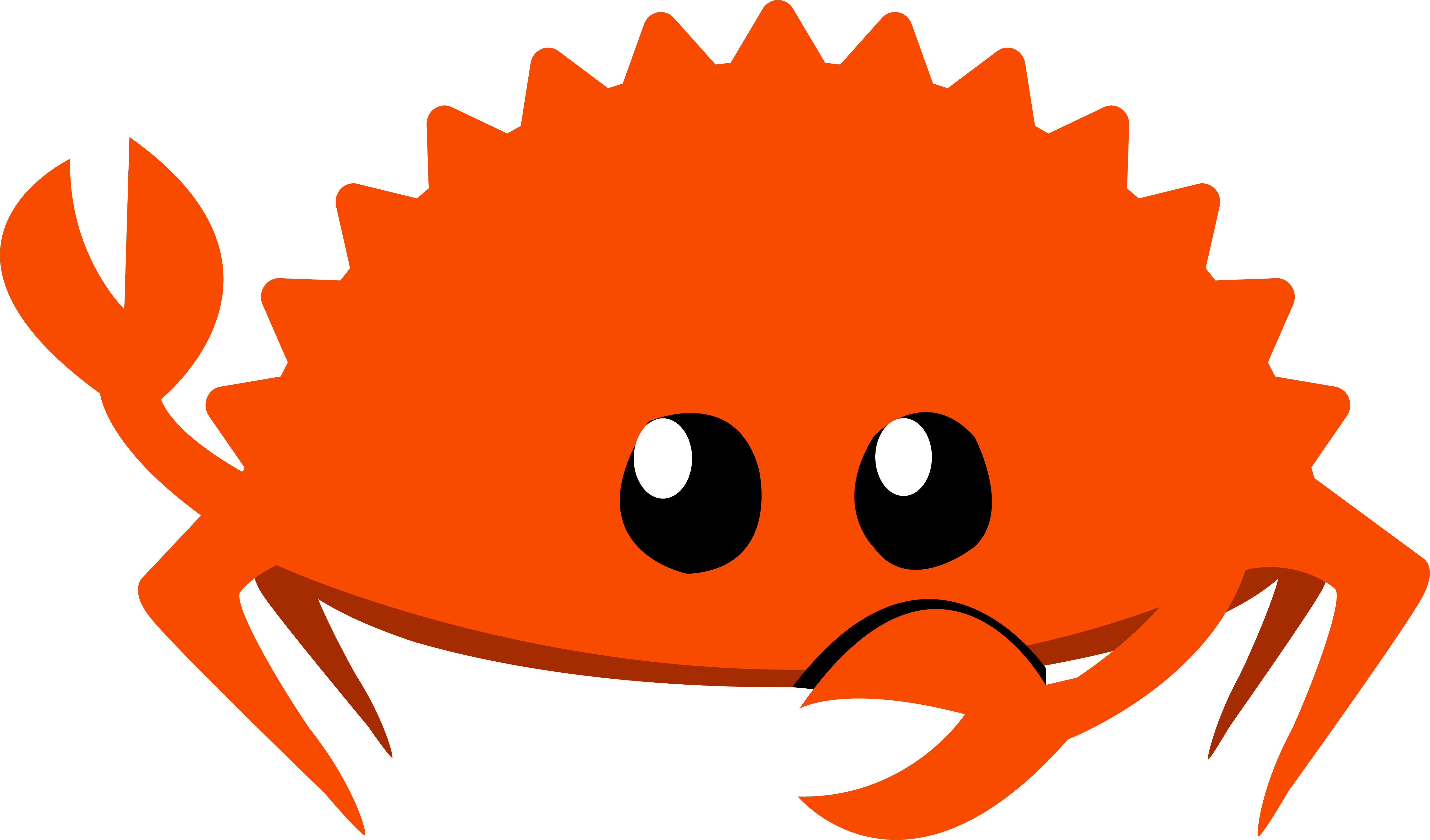}};
      \matrix[matrix of nodes, column sep=1.5mm, right=2mm of subbor1] {
        |(res-subbor1-start) [resource, triangle, rotate=90, inner sep=2pt]| &[.3mm]
        |(res-subbor1-mid) [resource, circle, inner sep=4pt]| &
        |(res-subbor1-end) [resource, star, inner sep=3pt]| \\
      };

      \node[below right=0mm and 7.5mm of group-initial.east,
        label={[yshift=1.7mm]below:{\hsinline{Mut^α U_2}}}]
        (subbor2) {\includegraphics[width=1cm]{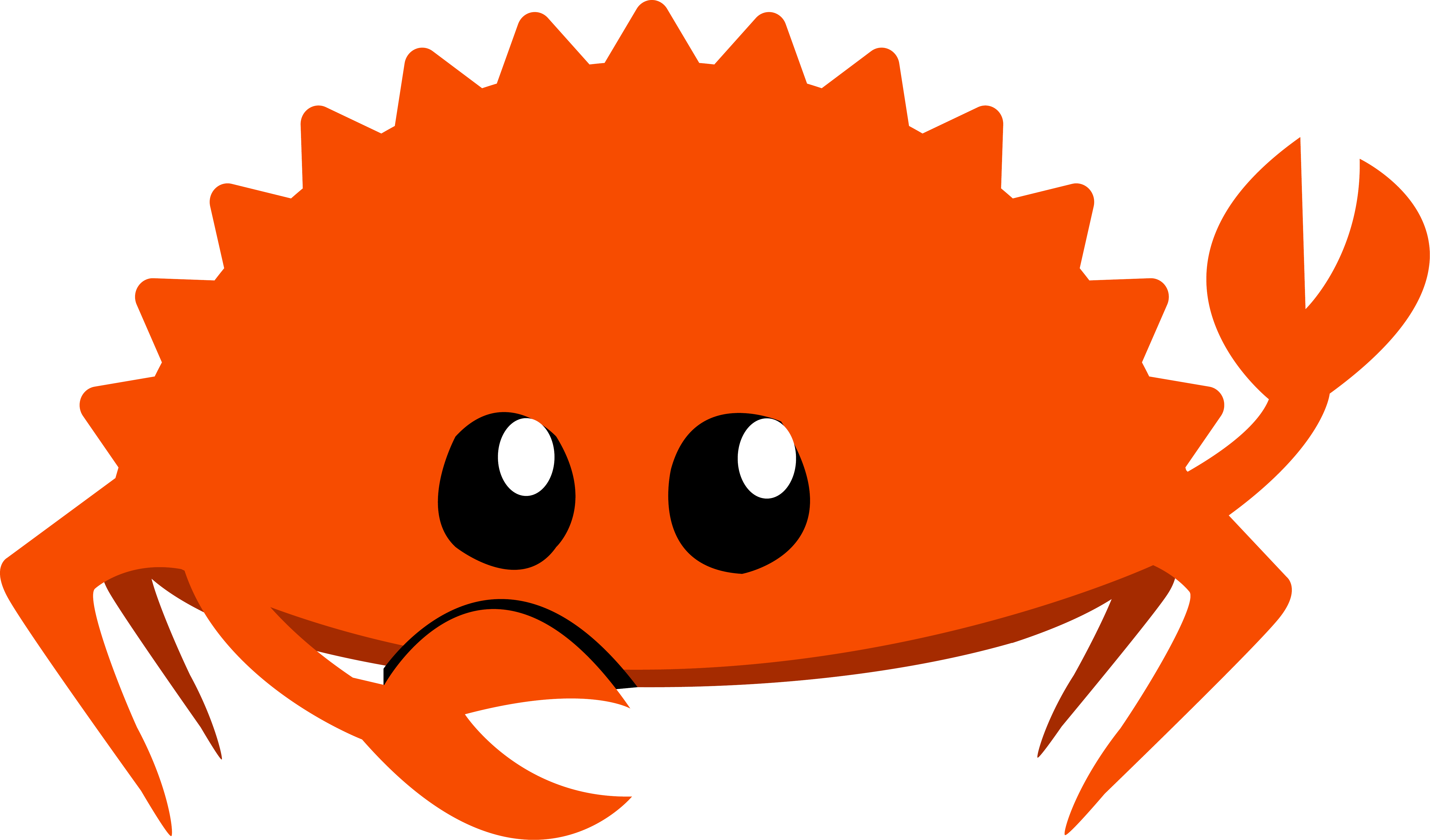}};
      \matrix[matrix of nodes, column sep=1.5mm, right=2mm of subbor2] {
        |(res-subbor2-start) [resource, triangle, rotate=-90, inner sep=2.25pt]| &
        |(res-subbor2-end) [resource, diamond, inner sep=3.25pt]| \\
      };
    \end{pgfonlayer}

    \begin{pgfonlayer}{background}
      \node[mutable region, rectangle corner radius north west=0pt, rectangle corner radius south west=0pt,
        minimum height=40pt, outer sep=0pt, fit={(padding-original) (res-original)}] (region-original) {};
      \node[mutable region, outer sep=0pt, fit={(res-borrow) (group-initial)}, minimum height=40pt] (region-borrow) {};
      \node[mutable region, outer sep=0pt,
        fit={(res-subbor1-start) (res-subbor1-mid) (res-subbor1-end)}, xshift=1pt] (region-subbor1) {};
      \node[mutable region, fit={(res-subbor2-start) (res-subbor2-end)}, xshift=-.5pt] (region-subbor2) {};
      \node[mutable region, rectangle corner radius north east=0pt,rectangle corner radius south east=0pt,
        minimum height=40pt, fit={(nodes-reclaimed) (padding-reclaimed)},
        right=7.35cm of region-original] (region-reclaimed) {};
    \end{pgfonlayer}

    \node (lender) [left=6pt of region-original, node distance=5mm, anchor=east] {Owner \hsinline{T}};
    \node (borrower) [left=6pt of region-borrow, node distance=5mm,
      label={left:{\sffamily Mutable borrower}},
      label={[yshift=6pt]below:{\hsinline{Mut^α T}}}]
      {\includegraphics[width=1.25cm]{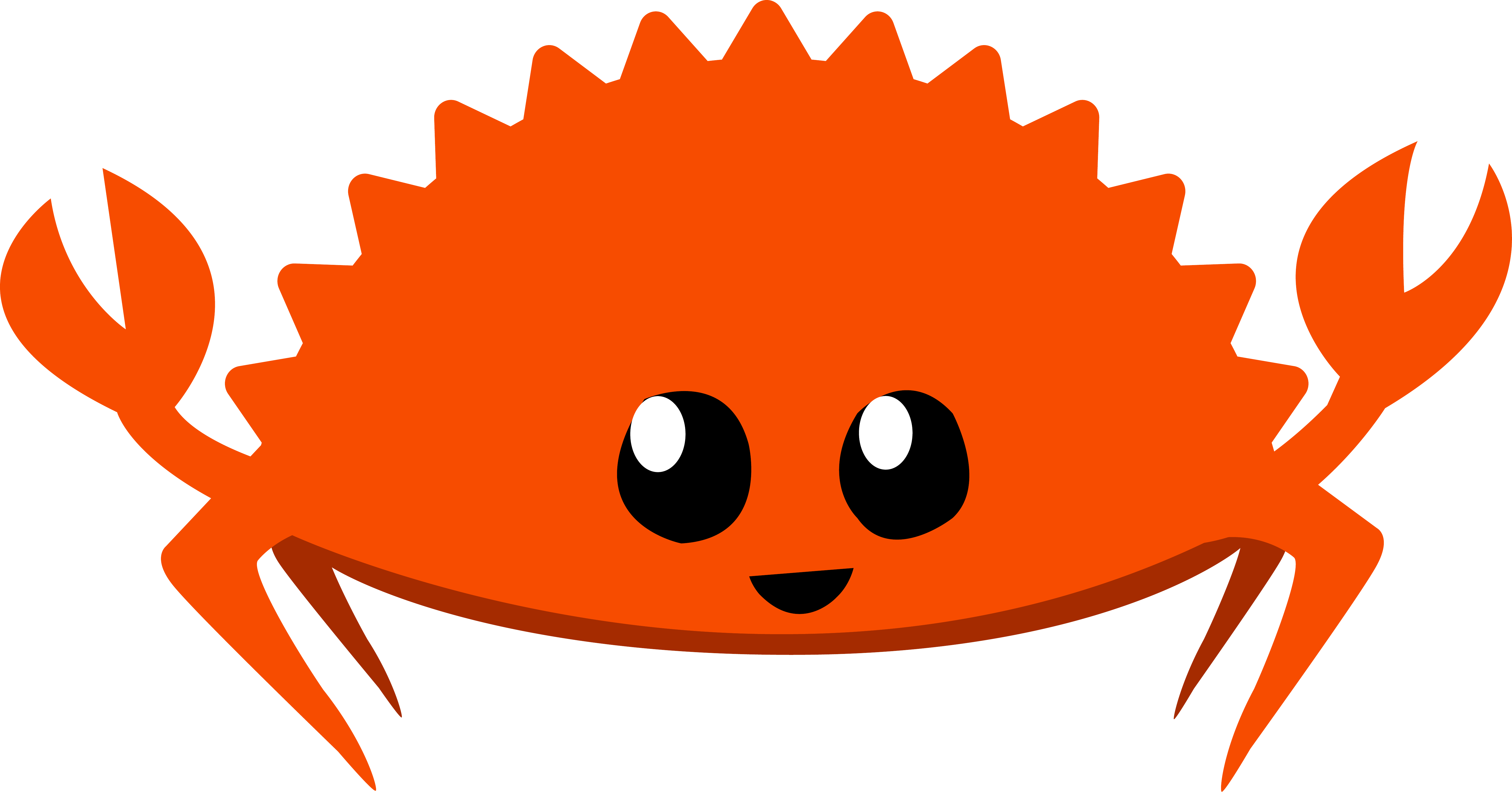}};

    \node[above right=-18pt and 2pt of region-original] {\sffamily Lender \hsinline{Lend^α T}};

    \coordinate (lifetime bound mid) at ($(region-reclaimed.west)!.45!(region-subbor1.east)$);
    \coordinate (lifetime bound top) at ($(region-reclaimed.north west)!(lifetime bound mid)!(region-reclaimed.north east)$);
    \coordinate (lifetime bound bottom) at ($(region-subbor2.south west)!(lifetime bound mid)!(region-subbor2.south east)-(0,4pt)$);
    \draw[thick, dashed, name path=lifetime-vert] (lifetime bound top) to (lifetime bound bottom);

    \draw[->, line width=3pt, draw=BorrowColor, name path=lifetime-horiz]
      (region-original.south east) to node[above right] {\sffamily Borrow} (region-borrow.north west);
    \draw[->, line width=3pt, draw=BorrowColor]
      (nodes-middle.north east) + (-2pt,-5pt) to (subbor1.west);
    \draw[->, line width=3pt, draw=BorrowColor]
      (nodes-middle.south east) + (-2pt,5pt) to (subbor2.west);
    \draw[line width=3pt, draw=MutationColor]
      (region-original) to (region-reclaimed);

    \node[above left=-12pt and 2pt of lifetime bound top] {\sffamily Lifetime \(\lft\)};

    \node[below left=2pt and 8pt of region-reclaimed, anchor=north west, align=left] {\sffamily \textuparrow\ Reclaim \hsinline{T}};
  \end{tikzpicture}
  \caption{Rust-style borrowing illustrated. Each Ferris corresponds to a borrower. Ferris illustrations by Karen Rustad Tölva, CC0, via \url{https://rustacean.net}.}
  \label{fig:borrow-illust}
\end{figure}

\Cref{fig:borrow-illust} illustrates the concept of Rust-style borrowing.
An original owner \hsinline{T} can be \emph{mutably borrowed} to create a \emph{mutable borrower} \hsinline{Mut^α T} and a \emph{lender} \hsinline{Lend^α T} for a time period called the \emph{lifetime} \hsinline{α}.\footnote{
  In Rust, a mutable borrower (our \hsinline{Mut^α T}) is typed \rsinline{&'a mut T}, but lending is only implicitly managed, lenders not being first-class citizens.
  Please see \cref{sect:compare} for more detailed discussion.
}
A mutable borrower can freely mutate the data during the lifetime \hsinline{α}.
Notably, borrowers are \emph{affine}: they can be freely split (e.g., \hsinline{Mut^α T} into borrowers \hsinline{Mut^α U_1}, \hsinline{Mut^α U_2} of subregions) and dropped at any time before the end of \hsinline{α}.
After the lifetime \hsinline{α} ends, the lender can reclaim full ownership of \hsinline{T}.
One might wonder how this happens \emph{non-locally}, without direct communication of ownership from borrowers to the lender.
Intuitively, ownership thrown away by borrowers is collected behind the scenes and finally returned to the lender once the lifetime \hsinline{α} ends.
This non-local handling of ownership and data is the distinct feature of Rust-style borrowing.

\paragraph{Core of Pure Borrow}

Now let's see how Rust-style borrowing works in Pure Borrow.
\vspace{-.7em}
\begin{figure}[H]
  \raggedright
  \begin{minipage}{.5\textwidth}
\begin{hscode}
type BO^α a;  type End^α
type Mut^α a;  type Share^α a;  type Lend^α a
\end{hscode}\vspace{-.3em}
\begin{hscode}
runBO :: Linearly =o
    (∀α. BO^α (End^α -> a)) -o a
borrow :: Linearly =o a -o (Mut^α a, Lend^α a)
share :: Mut^α a -o Ur (Share^α a)
reclaim :: Lend^α a -o End^α -> a
\end{hscode}\vspace{-.3em}
\begin{hscode}
modifyAt :: Int -> (a -o a) -o
    Mut^α (Vector a) -o BO^α (Mut^α (Vector a))
copyAt :: Copyable a => Int ->
    Share^α (Vector a) -> BO^α (Ur a)
\end{hscode}
  \end{minipage}\hspace{1em}
  \begin{minipage}{.455\textwidth}
\begin{hscode}[style=linenos]
example :: Ur (Int, [Int])
example = linearly _do
    vec <- newVector [0, 1, 2]
    runBO do @\label{line:example-pb-runBO}@
        let !(mvec, lend) = borrow vec @\label{line:example-pb-borrow}@
        mvec <- modifyAt 0 (+ 3) mvec @\label{line:example-pb-mutation-1}@
        mvec <- modifyAt 2 (+ 5) mvec @\label{line:example-pb-mutation-2}@
        mvec <- modifyAt 0 (* 4) mvec @\label{line:example-pb-mutation-3}@
        let !(Ur svec) = share mvec @\label{line:example-pb-share}@
        Ur n <- copyAt 0 svec @\label{line:example-pb-copy}@
        pure \end -> move
            (n, freeVector $ reclaim lend end) @\label{line:example-pb-reclaim}@
\end{hscode}
  \end{minipage}
  \caption{Core APIs of Pure Borrow and an example usage.}
  \label{fig:taste-core}
\end{figure}

\Cref{fig:taste-core} shows the core APIs and an example usage, implementing the same logic as \cref{fig:st-monad,fig:linear-haskell}.\footnote{
  The bang \hsinline{!} makes a pattern strict, not lazy, and here serves to pass type checking.
  The reader can safely ignore it.
}
At first, it proceeds similarly to \cref{fig:linear-haskell} by allocating a linear mutable vector \hsinline{vec} with \hsinline{newVector}.
But then we run a computation with borrowing using \hsinline{runBO} (\cref{line:example-pb-runBO}).
Crucially, \hsinline{runBO} can run \hsinline{BO^α} inside a \emph{pure} computation, with a guard by rank-2 polymorphism over the lifetime \hsinline{α}.
The \emph{BO monad} \hsinline{BO^α} is a novel monad we introduce in Pure Borrow, and it denotes a computation that can be run during the lifetime \hsinline{α}.
The mechanism of \hsinline{runBO} and \hsinline{BO^α} is akin to \hsinline{runST} and the \hsinline{ST^s} monad (recall \cref{fig:st-monad}, \cref{sect:background-st}).
Now we \hsinline{borrow} the vector \hsinline{vec} under \hsinline{α}, obtaining a mutable borrower \hsinline{mvec :: Mut^α (Vector Int)} and a lender \hsinline{lend :: Lend^α (Vector Int)} (\cref{line:example-pb-borrow}).
We can do mutations over the vector using \hsinline{mvec} (\cref{line:example-pb-mutation-1,line:example-pb-mutation-2,line:example-pb-mutation-3}).
Then we \hsinline{share} the affine mutable borrower \hsinline{mvec} to get a shared borrower \hsinline{svec :: Share^α (Vector Int)} (\cref{line:example-pb-share}) and read the value at the index \hsinline{0} with \hsinline{copyAt} (\cref{line:example-pb-copy}).
Notably, \hsinline{svec} is wrapped in \hsinline{Ur}, so we can use it arbitrarily many times, especially for reading with \hsinline{copyAt}.
We are reading only once in this example, but we could add more reads without any rebinding of \hsinline{svec}.
Finally, we can \hsinline{reclaim} the \hsinline{lend}er (\cref{line:example-pb-reclaim}) to retrieve \hsinline{Vector Int}, once given by \hsinline{runBO} a \emph{dead lifetime token} \hsinline{end :: End^α}, asserting that the lifetime \hsinline{α} has ended.

\paragraph{Parallelism}

With Pure Borrow, we can go further: mutations on different elements can be performed \emph{in parallel}.
Such safe parallel mutation in pure computation is a central benefit of Pure Borrow, practically demonstrated later in \cref{sect:casestudy}.
\vspace{-.7em}
\begin{figure}[H]
  \raggedright
  \begin{minipage}{.45\textwidth}
\begin{hscode}
parBO :: BO^α a -o BO^α b -o BO^α (a, b)
\end{hscode}\vspace{.2em}
\begin{hscode}
splitAt :: Int -> Borrow_k^α (Vector a) -o
    (Borrow_k^α (Vector a), Borrow_k^α (Vector a))
\end{hscode}
  \end{minipage}\hspace{2.5em}
  \begin{minipage}{.45\textwidth}
\begin{hscode}
let !(mvec_1, mvec_2) = splitAt 1 mvec
(mvec, ()) <- parBO
    (do @\base@mvec_1 <- modifyAt 0 (+ 3) mvec_1
         @\fit@modifyAt 0 (* 4) mvec_1)
    (consume <$> modifyAt 1 (+ 5) mvec_2)
\end{hscode}
  \end{minipage}
  \caption{Pure Borrow APIs for parallelism and an example usage, modifying \cref{line:example-pb-mutation-1,line:example-pb-mutation-2,line:example-pb-mutation-3} of \cref{fig:taste-core}.}
  \label{fig:taste-parallel}
\end{figure}

\Cref{fig:taste-parallel} shows a parallel variant of the mutation part (\cref{line:example-pb-mutation-1,line:example-pb-mutation-2,line:example-pb-mutation-3}) of \cref{fig:taste-core}.
First, we split with \hsinline{splitAt} the mutable borrower \hsinline{mvec} at the index \hsinline{1} into two disjoint mutable borrowers \hsinline{mvec_1} and \hsinline{mvec_2}.
Then we use our key API, \hsinline{parBO}, for \emph{parallel} composition of BO monads.
Here, the first thread mutates the 0th element twice through \hsinline{mvec_1} and the second thread mutates the 2nd element (of \hsinline{mvec}) once through \hsinline{mvec_2}.
Linearity \hsinline{-o} plays a vital role in the safety of \hsinline{parBO}:
thanks to linearity, the ownership held by the two BO monads is mutually \emph{disjoint}, which ensures that their mutations can safely be performed \emph{in parallel}, while retaining \emph{purity}.

\paragraph{Reborrowing}

There was one thing we glossed over in the previous example \cref{fig:taste-parallel}.
The new \hsinline{mvec} obtained from \hsinline{parBO} is actually \hsinline{mvec_1} in disguise.
Can we do better?

\emph{Reborrowing} solves this kind of situation.
It creates a borrower that borrows ownership from an existing mutable borrower for a shorter lifetime.
\vspace{-.9em}
\begin{figure}[H]
  \raggedright
  \begin{minipage}{.475\textwidth}
\begin{hscode}
reborrowing ::
    Mut^α a -o (∀β. Mut^β∧α a -o BO^β∧α' r) -o
    BO^α' (r, Mut^α a)
\end{hscode}
  \end{minipage}\hspace{1.2em}
  \begin{minipage}{.45\textwidth}
\begin{hscode}
((), mvec) <- reborrowing mvec \mvec ->
    let !(mvec_1, mvec_2) = splitAt 1 mvec in
    consume <$> parBO
        (do @\base@mvec_1 <- modifyAt 0 (+ 3) mvec_1
             @\fit@modifyAt 0 (* 4) mvec_1)
        (modifyAt 1 (+ 5) mvec_2)
\end{hscode}
  \end{minipage}
  \caption{Reborrowing API and an example usage, modifying \cref{line:example-pb-mutation-1,line:example-pb-mutation-2,line:example-pb-mutation-3} of \cref{fig:taste-core}.\label{fig:taste-reborrow}}
\end{figure}

\Cref{fig:taste-reborrow} shows an API, \hsinline{reborrowing}, and an example usage of it.
The function \hsinline{reborrowing} creates a mutable reference \hsinline{mvec :: Mut^β∧α (Vector a)} that only lives during the execution of the function.
Using this new \hsinline{mvec}, we can perform parallel mutation just like in \cref{fig:taste-parallel}.
Finally, the function restores the original mutable reference \hsinline{mvec :: Mut^α}, which has access to the \emph{whole} vector.
To obtain a shorter lifetime than \hsinline{α}, we use the intersection \hsinline{β ∧ α} of \hsinline{α} with a fresh lifetime \hsinline{β}.

\subsection{Details of Pure Borrow}
\label{sect:overview-details}

Now that we know the taste of Pure Borrow, let's take a closer look at it.
\Cref{fig:apis-lifetime,fig:apis-bo,fig:apis-borrow,fig:apis-vec,fig:apis-copy,fig:apis-split,fig:subty,fig:reborrow} present the central APIs of Pure Borrow.
One beautiful aspect is that high-level APIs can be derived from simpler primitives, thanks to Haskell's expressivity.

\paragraph{Lifetimes}

\begin{figure}
\begin{hscode}
data Lifetime = Al ι | (α :: Lifetime) ∧ (β :: Lifetime) | Static
class (α :: Lifetime) <= (β :: Lifetime)
\end{hscode}\vspace{-.3em}
\begin{hscode}
data Now@\(\ka^{\lft \ccol \Lifetime}\)@;  data End@\(\ka^{\lft \ccol \Lifetime}\)@;  instance LinearOnly Now^α;  instance Movable End^α
newLifetime :: Linearly =o (∀ι, Now^Alι -o a) -o a;    endLifetime :: Now^Alι -o Ur End^Alι
\end{hscode}
    \caption{Core APIs for lifetimes.}
    \label{fig:apis-lifetime}
\end{figure}

\Cref{fig:apis-lifetime} summarizes APIs for lifetimes.
A lifetime \hsinline{α :: Lifetime} is an atomic lifetime \hsinline{Al ι} for some id \hsinline{ι}, the static lifetime \hsinline{Static} that lives forever (\rsinline{'static} in Rust), or an intersection \hsinline{∧} of lifetimes.
As we saw above, lifetime intersection plays a key role in reborrowing.
Lifetime inclusion \hsinline{α <= β} is the partial order naturally induced, regarding \hsinline{∧} as the greatest lower bound and \hsinline{Static} as the maximum element.
We introduce two basic tokens for lifetimes, the \emph{live lifetime token} \hsinline{Now^α} and the \emph{dead lifetime token} \hsinline{End^α}.\footnote{
  To enhance automation, our Haskell implementation actually introduces a type class for \hsinline{End^α}.
}
The former exclusively asserts that \hsinline{α} is ongoing, and the latter persistently asserts that \hsinline{α} has ended.
With \hsinline{newLifetime}, we can get \hsinline{Now^Alι} for a fresh id \hsinline{ι}.
Then with \hsinline{endLifetime}, we can consume that to get \hsinline{End^Alι}.
The idea of these tokens comes from RustBelt's lifetime logic \cite{JungKD18-RustBelt}, as we discuss later.

\paragraph{BO monad}

\Cref{fig:apis-bo} shows the APIs for our BO monad.
At runtime, \hsinline{BO^α a} is just a thunk for \hsinline{a}, just like \hsinline{ST^s a}, but the APIs for it are quite new.
The key primitive is \hsinline{execBO}, which takes a pure value out of \hsinline{BO^α} using the live lifetime token \hsinline{Now^α}.
Actually, a high-level combinator \hsinline{runBO} we used above (\cref{fig:taste-core}) is derived from \hsinline{execBO} together with \hsinline{newLifetime} and \hsinline{endLifetime}, using rank-2 polymorphism.
Also, we have a primitive \hsinline{sexecBO} and a derived combinator \hsinline{srunBO} for \emph{scoped} execution of the BO monad for sub-lifetimes introduced through intersection \hsinline{∧}.
It is used to derive the \hsinline{reborrowing} combinator we used above (\cref{fig:taste-reborrow}), as shown soon (\cref{fig:reborrow}).

\paragraph{Borrowing}

\begin{figure}
\begin{hscode}
newtype BO@\(\ka^{\lft \ccol \Lifetime}\)@ a;    pure :: a -o BO^α a
(>>=) :: BO^α a -o (a -o BO^α b) -o BO^α b;    parBO :: BO^α a -o BO^α b -o BO^α (a, b)
\end{hscode}\vspace{-.3em}
    \begin{minipage}{.425\textwidth}
\begin{hscode}
execBO :: Now^α -o BO^α a -o (Now^α, a)
\end{hscode}\vspace{-.7em}
\begin{hscode}
runBO :: Linearly =o
    (∀α. BO^α (End^α -> a)) -o a
runBO bo = newLifetime \now -> _do
    (now, f) <- execBO now bo
    Ur end <- endLifetime now;  f end
\end{hscode}
    \end{minipage}\hspace{1.2em}
    \begin{minipage}{.525\textwidth}
\begin{hscode}
sexecBO :: Now^β -o BO^β∧α a -o BO^α (Now^β, a)
\end{hscode}\vspace{-.7em}
\begin{hscode}
srunBO :: Linearly =o
    (∀β. BO^β∧α (End^β -> a)) -o BO^α a
srunBO bo = newLifetime \now -> do
    (now, f) <- sexecBO now bo
    Ur end <- pure (endLifetime now);  pure (f end)
\end{hscode}
    \end{minipage}
    \caption{Core APIs for the BO monad.}
    \label{fig:apis-bo}
\end{figure}

\Cref{fig:apis-borrow} lists the APIs for borrowing.
We feature \hsinline{Mut^α a} for a mutable borrower, \hsinline{Share^α a} for a shared borrower, and \hsinline{Lend^α a} for a lender.
Notably, they are just \hsinline{newtype} wrappers over the body type \hsinline{a}, i.e., they have the same runtime representation as \hsinline{a}.\footnote{
  We exploit this fact to implement various borrowing operations in the subsequent APIs (e.g., \hsinline{borrow}, \hsinline{share} and \hsinline{reclaim}) by internally using unsafe zero-cost type coercion (i.e., a representationally identity function).
}
Also, any Haskell type \hsinline{a} can be a target of borrowing.
For utility, \hsinline{Mut^α} and \hsinline{Share^α} are unified into a general type \hsinline{Borrow_k^α}.
The following are the core primitives:
we can \hsinline{borrow} an object to create a mutable borrower and a lender,
then we can \hsinline{share} a mutable borrower to get an unrestricted shared borrower,
and finally we can \hsinline{reclaim} the ownership from a lender once the lifetime has ended.

\begin{figure}
    \begin{minipage}{.35\textwidth}
\begin{hscode}
newtype Borrow@\(\ka_{\namesty{k} \ccol \BorrowKind}^{\lft \ccol \Lifetime}\)@ a
data BorrowKind = Mut | Share
\end{hscode}\vspace{-.5em}
\begin{hscode}
type Mut^α = Borrow@\({}_{\Mut}^\lft\)@
type Share^α = Borrow@\({}_{\Share}^\lft\)@
\end{hscode}\vspace{-.5em}
\begin{hscode}
newtype Lend@\(\ka^{\lft \ccol \Lifetime}\)@ a
\end{hscode}
    \end{minipage}\hspace{1.2em}
    \begin{minipage}{.525\textwidth}
\begin{hscode}
instance LinearOnly (Mut^α a)
instance Consumable (Mut^α a)
instance Movable (Share^α a)
\end{hscode}\vspace{-.5em}
\begin{hscode}
borrow @\base@:: Linearly =o a -o (Mut^α a, Lend^α a)
reclaim @\fit@:: Lend^α a -o End^α -> a
share   @\fit@:: Mut^α a -o Ur (Share^α a)
\end{hscode}
    \end{minipage}
    \caption{Core APIs for borrowing.}
    \label{fig:apis-borrow}
\end{figure}

\paragraph{Vectors}

\begin{figure}
    \begin{minipage}{.475\textwidth}
\begin{hscode}
getAt :: β <= α => Int ->
    Borrow_k^α (Vector a) -o BO^β (Borrow_k^α a)
swapAt :: β <= α => Int -> Int ->
    Mut^α (Vector a) -o BO^β (Mut^α (Vector a))
updateAt :: β <= α => Int ->
    (a -o BO^β (b, a)) -o Mut^α (Vector a) -o
    BO^β (b, Mut^α (Vector a))
\end{hscode}
    \end{minipage}\hspace{.7em}
    \begin{minipage}{.475\textwidth}
\begin{hscode}
modifyAt :: β <= α => Int -> (a -o a) -o
    Mut^α (Vector a) -o BO^β (Mut^α (Vector a))
modifyAt i f mvec = do
    ((), mvec) <- updateAt i
        (\ a -> pure ((), f a)) mvec
    pure mvec
\end{hscode}
    \end{minipage}\vspace{-.3em}
    \caption{Core APIs for accessing vectors through borrowers.}
    \label{fig:apis-vec}
\end{figure}

\Cref{fig:apis-vec} shows the functions for accessing vectors through borrowers.
The primitive \hsinline{getAt i vec} takes a borrower to the \hsinline{i}-th element of \hsinline{vec}.
The primitive \hsinline{swapAt i j vec} swaps the \hsinline{i}-th and \hsinline{j}-th elements of \hsinline{vec}.
The primitive \hsinline{updateAt i k vec} updates the \hsinline{i}-th element of \hsinline{vec} using \hsinline{k};
it is quite expressive, allowing access to the \hsinline{BO} monad inside \hsinline{k}.

\paragraph{Copying}

\begin{figure}
\begin{hscode}
class Copyable a;    copy :: Copyable a => Share^α a -> a
instance Copyable Int;    instance Copyable End^α;    instance Copyable (Share^α a)
\end{hscode}\vspace{-.5em}
\begin{hscode}
copyAt :: (Copyable a, β <= α) => Int -> Share^α (Vector a) -> BO^β (Ur a)
copyAt i svec = do Ur s <- move <$> getAt i svec; pure $ Ur (copy s)
\end{hscode}
    \caption{Core APIs for copying.}
    \label{fig:apis-copy}
\end{figure}

We provide a mechanism called copying, as shown in \cref{fig:apis-copy}.
We can \hsinline{copy} the body \hsinline{a} out of a shared reference \hsinline{Share^α a} if \hsinline{Copyable a} holds.
It holds roughly when \hsinline{a} is a persistent data type, such as \hsinline{Int}, under a condition similar to but different from \hsinline{Movable a}; \hsinline{Copyable} does not take ownership, while \hsinline{Movable} does.
More specifically, any persistent data types (i.e., those without any mutable states or linear resources) can be \hsinline{Copyable}.\footnote{
  The class \hsinline{Copyable} roughly corresponds to Rust's \rsinline{Copy}, but further allows immutable boxing (like Rust's \rsinline{Box<T>} but immutable) unlike it.
  Notably, this enables recursive data types to be \hsinline{Copyable} (e.g., \hsinline{[a]} is \hsinline{Copyable} under \hsinline{Copyable a}).
}

\paragraph{Splitting}

\begin{figure}
\begin{hscode}
splitPair    @\base@:: Borrow_k^α (a, b) -o (Borrow_k^α a, Borrow_k^α b)
splitEither   @\fit@:: Borrow_k^α (Either a b) -o Either (Borrow_k^α a) (Borrow_k^α b)
splitAt       @\fit@:: Int -> Borrow_k^α (Vector a) -o (Borrow_k^α (Vector a), Borrow_k^α (Vector a))
\end{hscode}
    \caption{APIs for splitting.}
    \label{fig:apis-split}
\end{figure}

\Cref{fig:apis-split} shows primitives for splitting borrowers.
Our Haskell implementation includes a generic mechanism for deriving a \hsinline{split^*} primitive for various data types, not only the ones shown here.
We also have a splitting primitive for a vector.
Notably, these splits can be performed even outside the BO monad, because the information required for them is all persistent.\footnote{
  Note that the length of \hsinline{Vector a} is immutable, as the type represents a fixed-size slice of an array.
}

\paragraph{Subtyping}

\begin{figure}
  \raggedright
  \begin{minipage}{.425\textwidth}
\begin{hscode}
class a <: b;    upcast :: a <: b => a -o b
\end{hscode}
\begin{hscode}
instance a <: b => Vector a <: Vector b
instance β <= α => End^α <: End^β
\end{hscode}
    \end{minipage}\hspace{1.0em}
    \begin{minipage}{.5375\textwidth}
\begin{hscode}
instance (β <= α, a <: b) => BO^α a <: BO^β b
instance (β <= α, a <: b, b <: a) => Mut^α a <: Mut^β b
instance (β <= α, a <: b) => Share^α a <: Share^β b
instance (α <= β, a <: b) => Lend^α a <: Lend^β b
\end{hscode}
  \end{minipage}
  \caption{Core APIs for subtyping.}
  \label{fig:subty}
\end{figure}

We introduce \emph{subtyping} \hsinline{a <: b} as a type class, as shown in \cref{fig:subty}.
Subtyping is crucial, especially for modifying the lifetime information.
Note that the mutable vector type \hsinline{Vector a} is covariant, because it is used linearly.
On the other hand, the mutable borrower type \hsinline{Mut^α a} is \emph{invariant} over the body type \hsinline{a}, because the body is mutably shared between the borrower and the lender.
This aspect is identical to Rust's subtyping.

\paragraph{Reborrowing}

\begin{figure}
  \centering
\begin{hscode}[xleftmargin=.25\textwidth]
joinMut :: Borrow_k^β (Mut^α a) -o Borrow_k^β∧α a
\end{hscode}\vspace{-.3em}
  \raggedright
  \begin{minipage}{.4\textwidth}
\begin{hscode}
reborrow :: Mut^α a -o
    (Mut^β∧α a, Lend^β (Mut^α a))
\end{hscode}\vspace{-.5em}
\begin{hscode}
reborrowing :: Mut^α a -o
    (∀β. Mut^β∧α a -o BO^β∧α' r) -o
    BO^α' (r, Mut^α a)
\end{hscode}\vspace{-.5em}
\begin{hscode}
sharing :: Mut^α a -o
    (∀β. Share^β∧α a -> BO^β∧α' r) -o
    BO^α' (r, Mut^α a)
\end{hscode}\vspace{-.5em}
\begin{hscode}
copyAtMut :: (Copyable a, β <= α) =>
    Int -> Mut^α (Vector a) ->
    BO^β (Ur a, Mut^α (Vector a))
\end{hscode}
  \end{minipage}\hspace{.8em}
  \begin{minipage}{.55\textwidth}
\begin{hscode}
reborrow mut = withLinearly mut \mut ->
    let !(mut', lend) = borrow mut in
    (joinMut mut', lend)
\end{hscode}\vspace{-.3em}
\begin{hscode}
reborrowing mut ko = withLinearly mut \mut ->
    srunBO do
        let !(mut', lend) = reborrow mut;  r <- ko mut'
        pure \end -> (r, reclaim lend (upcast end))
\end{hscode}\vspace{-.3em}
\begin{hscode}
sharing mut ko = reborrowing mut \mut ->
    let !(Ur shr) = share mut in ko shr
\end{hscode}\vspace{-.3em}
\begin{hscode}
copyAtMut i mvec = sharing mvec $ copyAt i
\end{hscode}
  \end{minipage}
  \caption{Primitive and derived operators for reborrowing.}
  \label{fig:reborrow}
\end{figure}

\Cref{fig:reborrow} shows the primitive and derived combinators for reborrowing.
Interestingly, the core primitive for reborrowing is \hsinline{joinMut}, which flattens a borrower over a mutable borrower, taking the intersection lifetime.
The key case is when it is instantiated for \hsinline{Mut}, with type \hsinline{Mut^β (Mut^α a) -o Mut^β∧α a}.
A simple operator \hsinline{reborrow} can be derived from that, simply by borrowing a mutable borrower and joining that.
Now we can derive the useful combinator \hsinline{reborrowing} used above (\cref{fig:taste-reborrow}) for reborrowing a mutable borrower, using a delimited scope of a sub-lifetime introduced by \hsinline{srunBO}.
We can then derive \hsinline{sharing} from \hsinline{reborrowing}, which can temporarily share a mutable borrower within a scope.
For example, from it we can immediately derive a useful function \hsinline{copyAtMut} for reading an element from a mutable borrower.

\paragraph{RustBelt and Pure Borrow}

\citet{JungKD18-RustBelt} developed RustBelt, a semantic soundness proof of Rust's ownership type system, by building the \emph{lifetime logic}, a library for modeling Rust-style borrowing in the \emph{separation logic} Iris \cite{JungSSSTBD15-Iris}.
Some of Pure Borrow's design ideas originate from RustBelt's lifetime logic \cite{JungKD18-RustBelt}.
These include lifetime tokens \hsinline{Now^α} and \hsinline{End^α} and the use of \hsinline{joinMut} for reborrowing.
Nevertheless, designing Pure Borrow safely required careful thought because the lifetime logic (or Iris in general) does not guarantee purity or leak freedom.

\paragraph{Expressivity compared to Rust}

Notably, Pure Borrow supports non-lexical lifetimes \cite{Rust17-NLL} to some extent.
Although our high-level combinators \hsinline{runBO} and \hsinline{srunBO} offer lexically scoped lifetimes, we also provide more primitive APIs, \hsinline{execBO}, \hsinline{sexecBO}, \hsinline{newLifetime} and \hsinline{endLifetime}, for managing lifetimes manually with lifetime tokens \hsinline{Now^α} and \hsinline{End^α}.
This is similar to what RustBelt \cite{JungKD18-RustBelt} does to model Rust's lifetime-based type system using lifetime tokens.
Still, Pure Borrow lacks a highly automated borrow checker like Rust has, and such smooth automation in Pure Borrow is left for future work.

Some advanced borrowing patterns available in Rust are not supported in Pure Borrow, on the other hand.
For example, let us consider the following Rust code:
\begin{rscode}
enum Sum<A, B> { Left(A), Right(B) }
fn test<A, B>(mut s: Sum<A, B>, newb: B) -> Option<Sum<A, B>> {
    match s  {  @\base@Sum::Left(a) => { drop(a); None }
                 @\fit@Sum::Right(ref mut b) => { *b = newb; Some(s) } } }
\end{rscode}
This is a mixture of moving and borrowing, enabled by Rust's \rsinline{ref} pattern.\footnote{
  In the \rsinline{Left} case, the body \rsinline{a: A} of \rsinline{s} is moved out and dropped,
  whereas in the \rsinline{Right} case, the body of \rsinline{s} is mutably borrowed as \rsinline{b: &mut B}.
  Therefore, \rsinline{s} is partially moved out only in the first case and mutably borrowed only in the second case.
}
Pure Borrow does not support this kind of advanced borrowing pattern because it requires pattern matching to be tightly coupled with borrowing.

\section{Case Study and Evaluation}
\label{sect:casestudy}

We have implemented the Pure Borrow APIs presented in \cref{sect:overview} as a \texttt{pure-borrow} package in Haskell,
which compiles with the current Haskell compiler GHC 9.10+.
In this section, we demonstrate its power with a case study: \emph{parallel} quicksort \cite{Hoare61b-quicksort} implemented in our library.\footnote{
  Our implementation and benchmark suite are archived at \cite{MatsushitaIshii26-Pure-Borrow-artifact}, and the latest version is available at \url{https://github.com/SoftwareFoundationGroupAtKyotoU/pure-borrow}.
}

\subsection{Parallel Quicksort Implementations}
\label{sect:casestudy-qsort}

We have two implementations of parallel quicksort: a naïve version and a work-stealing version.
In what follows, we present these implementations and discuss the expressiveness of our APIs.

\begin{figure}
  \raggedright
  \begin{minipage}{.8\textwidth}
\begin{hscode}[style=linenos]
qsort :: (Ord a, Copyable a, β <= α) => Int -> Mut^α (Vector a) -o BO^β ()
qsort budget mvec = case size mvec of
    (Ur n, mvec) @\base@| n <= 1 -> pure $ consume mvec
                  @\fit@| otherwise -> do
        let i = quot n 2
        (Ur pivot, mvec) <- copyAtMut i mvec @\label{line:copy1}@
        (mlo, mhi) <- divide pivot mvec 0 n
        let budget' = quot budget 2 @\label{line:divide-naive}@
        consume <$> parIf (budget' > 0)
            (qsort budget' mlo) (qsort budget' mhi) @\label{line:naive-recurse}@
    where
        parIf b = if b then parBO else \l r -> do l <- l; r <- r; pure (l, r)
\end{hscode}
  \end{minipage}\hspace{-.25\textwidth}
  \begin{minipage}{.4375\textwidth}
    \centering
    \vspace{-.8em}
    \newcommand*{\rectheight}{.7}
    \newcommand*{\pivotwidth}{.75}
    \newcommand*{\drawrectangle}[3]{
      \draw[elt]
        (#1, {#3 + \rectheight / 2}) rectangle
        (#2, {#3 - \rectheight / 2});
      \node[above] at ({(#1 + #2) / 2}, {#3 + \rectheight / 2 - .225})
        {\includegraphics[width=1.25em]{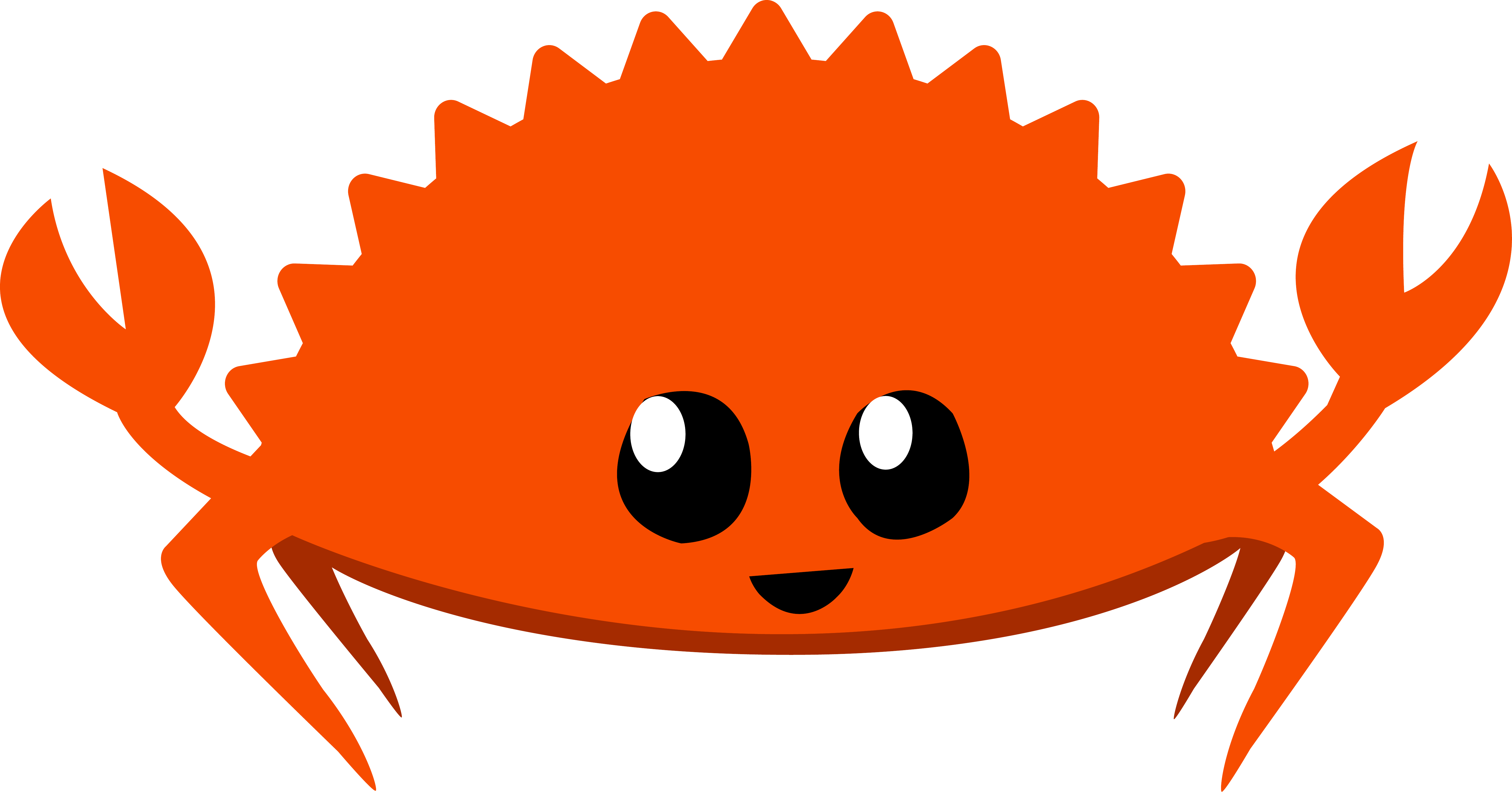}};
    }
    \newcommand*{\drawpivotrectangle}[3]{
      \drawrectangle{#1}{#2}{#3}
      \draw[pivot]
        ({(#1 + #2) / 2 - \pivotwidth / 2}, {#3 + \rectheight / 2}) rectangle
        ({(#1 + #2) / 2 + \pivotwidth / 2}, {#3 - \rectheight / 2});
    }
    \newcommand*{\drawbudget}[2]{
      \node[anchor=west, yshift=-.05em] at (10.3, #1) {\small #2};
    }
    \begin{tikzpicture}[
      x=1.15em, y=1.15em,
      elt/.style={fill=WorkFillColor,draw=WorkBorderColor},
      pivot/.style={fill=PivotColor,draw=WorkBorderColor}
    ]
      \drawpivotrectangle{0}{10}{4}
      \drawbudget{4}{\hsinline{budget = 4}}
      \drawpivotrectangle{0}{5.25}{2}
      \drawpivotrectangle{5.75}{10}{2}
      \drawbudget{2}{\hsinline{budget = 2}}
      \drawrectangle{0}{0.75}{0}
      \drawpivotrectangle{1}{5.25}{0}
      \drawrectangle{5.75}{6.5}{0}
      \drawpivotrectangle{6.75}{10}{0}
      \drawbudget{0}{\hsinline{budget = 1}}
    \end{tikzpicture}\vspace{.3em}
    \sffamily\small
    Each \raisebox{-.125em}{\includegraphics[width=1.5em]{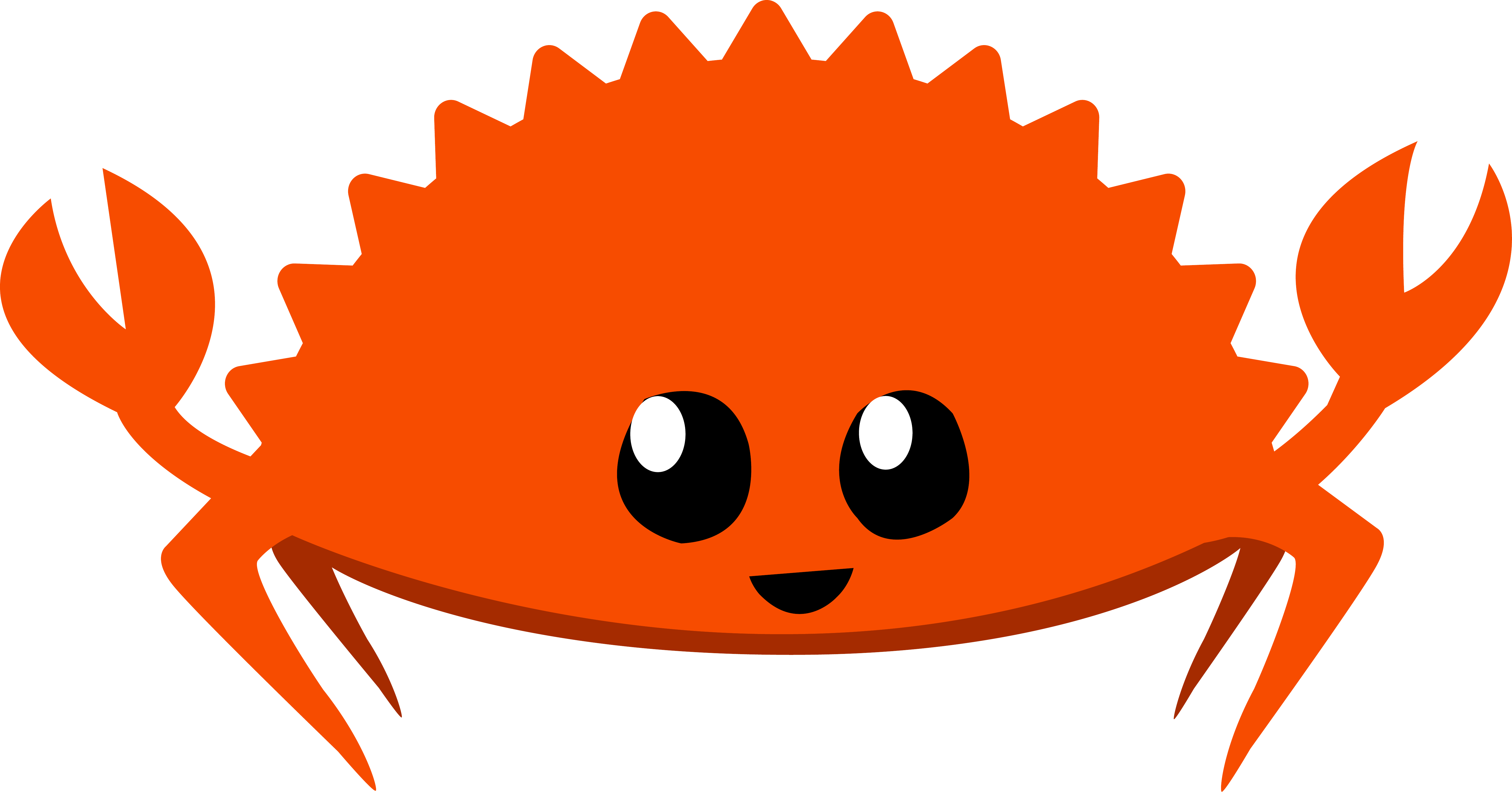}} depicts a lightweight thread\\
    spawned by \hsinline{parBO}.
    Ferris: Tölva, CC0.
  \end{minipage}
\begin{hscode}[style=nolinenos]
qsort' :: (Ord a, Copyable a, Movable a) => Int -> [a] -> Ur [a]
qsort' budget xs = linearly _do
    vec <- newVector xs
    runBO do @\base@(mvec, lend) <- borrow vec;  qsort budget mvec
              @\fit@pure $ move . freeVector . reclaim lend
\end{hscode}
  \caption{
    A naïve parallel implementation of quicksort and an illustration.
  }
  \label{fig:qsort-naive}
\end{figure}

\begin{figure}
\begin{hscode}[style=linenos]
divide :: (Ord a, Copyable a, β <= α) =>
    Int -> Mut^α (Vector a) -o Int -> Int -> BO^β (Mut^α (Vector a), Mut^α (Vector a))
divide pivot = partUp
    where
        partUp mvec l u
            | l < u = do @\base@(Ur a, mvec) <- copyAtMut l mvec
                          @\fit@if a < pivot then partUp mvec (l + 1) u else partDown mvec l (u - 1)
            | otherwise = pure $ splitAt l mvec @\label{line:splitAt1}@
        partDown mvec l u
            | l < u = do @\base@(Ur a, mvec) <- copyAtMut u mvec
                          @\fit@if pivot < a @\baseB@then partDown mvec l (u - 1)
                                              @\fitB@else do mvec <- swapAt mvec l u; partUp mvec (l + 1) u @\label{line:swapAt}@
            | otherwise = pure $ splitAt l mvec @\label{line:splitAt2}@
\end{hscode}
  \caption{The implementation of \hsinline{divide} function.}
  \label{fig:divide}
\end{figure}

\paragraph{Naïve parallel quicksort}

\Cref{fig:qsort-naive} shows our \emph{naïve} implementation of parallel quicksort.
The function \hsinline{qsort} is the main part.
It creates a spark (lightweight thread) for each partitioned sub-vector, as illustrated in the figure.
Its \hsinline{budget} upper-bounds the number of sub-vectors, and it falls back to sequential execution if the \hsinline{budget} is exhausted.
The function proceeds as follows.
If the size of the vector is more than one, it selects a pivot value from the middle of the vector (\cref{line:copy1}) with \hsinline{copyAtMut}---recall that it is implemented with \emph{reborrowing}, as shown in \cref{fig:reborrow}.
Then it calls \hsinline{divide}, which we explain soon, to partition the mutable borrower \hsinline{mvec} into low and high parts (\cref{line:divide-naive}).
Finally, it recurses on the sub-vectors with halved budgets (\cref{line:naive-recurse}).
The calls are combined using the \hsinline{parIf} combinator, which invokes \hsinline{parBO} (\cref{fig:apis-bo}) to evaluate both arguments in parallel if the budget remains, or otherwise executes them sequentially.

\paragraph{Division}

\Cref{fig:divide} shows the \hsinline{divide} function that partitions the given vector into low and high parts according to the pivot value.
This is an almost literal translation of the \hsinline{partitionBy} function from \texttt{vector-algorithms} package \cite{vector-algorithms} with our API.
The function works in two modes: \hsinline{partUp} and \hsinline{partDown}, which scan and grow the lower and upper ends of the active interval, respectively.
When an out-of-place element is found, it swaps the elements at the two indices and switches mode with \hsinline{swapAt} (\cref{line:swapAt}).
When the two indices meet, it splits the vector into two parts with \hsinline{splitAt} (\cref{line:splitAt1,line:splitAt2}) and returns them.

\paragraph{Work-stealing parallel quicksort}

\begin{figure}
  \raggedright
  \begin{subcaptionblock}{\textwidth}
\begin{hscode}
data DivideConquer^α a = DivideConquer { divide :: ∀β. β <= α => Mut^β a -o BO^β (Result^β a) }
data Result^β a = Done | Continue [Mut^β a]
divideAndConquer :: β <= α => Int -> DivideConquer^α a -> Mut^α a -o BO^β (Mut^α a)
\end{hscode}\vspace{-.45em}
    \caption{General API for work-stealing parallelism.}
    \label{fig:work-stealing-api}
  \end{subcaptionblock}
  \begin{subcaptionblock}{.6\textwidth}
\begin{hscode}
qsortDC :: (Ord a, Copyable a, β <= α) =>
    Int -> Int -> Mut^α (Vector a) -o BO^β (Mut^α (Vector a))
qsortDC nwork thresh = divideAndConquer nwork $
    DivideConquer \mvec -> case size mvec of
        (Ur n, mvec)
            | n <= 1 -> let () = consume mvec in pure Done
            | n <= thresh -> Done <$> qsort 0 mvec
            | otherwise -> let i = quot n 2 in do
                  (Ur pivot, mvec) <- copyAtMut i mvec
                  (mlo, mhi) <- divide pivot mvec 0 n
                  pure $ Continue [mlo; mhi]
\end{hscode}\vspace{-.35em}
    \caption{Quicksort implementation.}
    \label{fig:work-stealing-qsort}
  \end{subcaptionblock}\hspace{.3em}
  \begin{subcaptionblock}{.375\textwidth}
    \centering
    \pgfdeclarelayer{background}
    \pgfsetlayers{background,main}
    \newcommand*{\rectheight}{.45em}
    \begin{tikzpicture}[baseline,
      work/.style={rectangle,fill=WorkFillColor,draw=WorkBorderColor,minimum width=2.4em,minimum height=\rectheight},
      deque/.style={rectangle,draw=WorkBorderColor,fill=gray!8,inner sep=2pt},
      lbl/.style={font=\sffamily\tiny},
      ]
      \node (w1s3) [work] at (0,0) {};
      \node (w1s2) [work] at (0,-0.35) {};
      \node (w1s1) [work] at (0,-0.7) {};
      \begin{pgfonlayer}{background}
        \node (w1deque) [deque, fit={(w1s1) (w1s2) (w1s3)}] {};
      \end{pgfonlayer}

      \node (w1) at (0,-1.45) {\includegraphics[width=0.5cm]{images/ferris-happy.pdf}};
      \draw[->,thick,draw=BorrowColor] (w1deque.south) -- node[right,lbl,text=BorrowColor]{pop} (w1.north);

      \node (w1cur) [work,minimum width=3em] at (0,-2.1) {};
      \node [rectangle,fill=PivotColor,draw=WorkBorderColor,minimum width=.7em,minimum height=\rectheight] at (w1cur.center) {};
      \node (w1da) [work,minimum width=1.4em] at (-0.5,-2.6) {};
      \node (w1db) [work,minimum width=1.2em] at (0.5,-2.6) {};
      \draw[->,thick] (w1cur.south west) ++(2pt,0) to (w1da);
      \draw[->,thick] (w1cur.south east) ++(-2pt,0) to (w1db);
      \begin{pgfonlayer}{background}
        \node (divs) [ellipse,draw,fill=gray!25, fit={(w1da) (w1db)},inner sep=1pt] {};
      \end{pgfonlayer}

      \draw[->,thick,draw=BorrowColor] (w1.south) to (w1cur.north);
      \draw[->,thick,draw=BorrowColor] (divs.west) to[out=180,in=210] node[left,lbl,text=BorrowColor]{push} (w1deque.south west);

      \node (w2empty) [lbl,text=gray] at (1.5,-0.35) {(empty)};
      \begin{pgfonlayer}{background}
        \node (w2deque) [deque, fit={(w2empty)},minimum width=2.8em,minimum height=2.6em] {};
      \end{pgfonlayer}

      \node (w2) at (1.5,0.35) {\includegraphics[width=0.5cm]{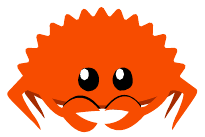}};

      \draw[->,ultra thick,draw=BorrowColor,densely dashed]
        (w2.north) to[bend right=25]
        node[left,lbl,xshift=-10,text=BorrowColor] {steal}
        (w1deque.north);

      \node (mAw) at (-1.4,-0.9) {\includegraphics[width=3mm]{images/ferris-happy.pdf}};
      \node (mAs1) [work,minimum width=1.4em,minimum height=.3em] at (-1.4,-0.55) {};
      \node (mAs2) [work,minimum width=1.4em,minimum height=.3em] at (-1.4,-0.35) {};
      \begin{pgfonlayer}{background}
        \node [deque,fit={(mAs1)(mAs2)},inner sep=1pt] {};
      \end{pgfonlayer}

      \node (mBw) at (1.1,-1.5) {\includegraphics[width=3mm]{images/ferris-happy.pdf}};
      \node (mBs1) [work,minimum width=1.4em,minimum height=.3em] at (1.1,-1.2) {};
      \begin{pgfonlayer}{background}
        \node [deque,fit={(mBs1)},inner sep=1pt,minimum height=.8em] {};
      \end{pgfonlayer}

      \node (mCw) at (1.9,-1.5) {\includegraphics[width=3mm]{images/ferris-usual.pdf}};
      \node (mCe) [lbl,text=gray] at (1.9,-1.2) {};
      \begin{pgfonlayer}{background}
        \node [deque,fit={(mCe)},inner sep=1pt,minimum width=1.6em,minimum height=.8em] {};
      \end{pgfonlayer}
    \end{tikzpicture}\vspace{.5em}
    \caption{
      Each worker \raisebox{-.125em}{\includegraphics[width=1.5em]{images/ferris-happy.pdf}} owns a deque.
      It pushes/pops at the top (near-side).
      An idle worker \raisebox{-.125em}{\includegraphics[width=1.5em]{images/ferris-usual.pdf}} steals from another's bottom.
      Ferris: Tölva, CC0.
    }
    \label{fig:work-stealing-illust}
  \end{subcaptionblock}
  \caption{Work-stealing parallelism.}
  \label{fig:work-stealing}
\end{figure}

While the previous implementation is concise, it is rather \emph{naïve} because \hsinline{parBO} is a structural parallel composition that waits for both arguments to finish, which can lead to extra waiting time and degraded performance.

To address this issue, we develop a general, modular, safe, and efficient API for \emph{work-stealing} parallelization of divide-and-conquer algorithms, as shown in \cref{fig:work-stealing-api}.
Although the work-stealing API is designed to be used with our Pure Borrow APIs and is mostly implemented using them, it also uses some unsafe data structures, such as concurrent queues mutated by multiple threads under the hood.
Hence, this API should be interpreted as a proof-of-concept demonstrating that Pure Borrow can provide good basic building blocks for more sophisticated APIs, and we do not claim that it is the best API for work-stealing parallelism.
\hsinline{DivideConquer} wraps a function for each step of divide-and-conquer, which takes a mutable borrower, executes in the \hsinline{BO} monad, and returns \hsinline{Result}, which is either \hsinline{Done} or \hsinline{Continue}s with divided parallelized pieces of work.
Given \hsinline{DivideConquer}, \hsinline{divideAndConquer} parallelizes the whole divide-and-conquer algorithm with a work-stealing scheduler.
As illustrated in \cref{fig:work-stealing-illust}, each worker owns a local double-ended queue of jobs, and whenever a thread finishes a job, its sub-jobs are pushed to the top. Then an idle worker first tries to pop another piece of work from the top of its own deque, and tries to \emph{steal} from the bottom of another's queue when its own queue is empty.
We also allocate a condition variable for each piece of work for synchronization.
We carefully design \hsinline{divideAndConquer} so that concurrent mutation is safely encapsulated, keeping the observed behavior pure.

\Cref{fig:work-stealing-qsort} gives a more sophisticated implementation of parallel quicksort with this work-stealing scheduler.
It retains almost the same logic as the naïve implementation \cref{fig:qsort-naive}, but returns ``divided'' pieces of work instead of recursing on the sub-vectors.
If the size of the vector is less than a given \hsinline{thresh}old, it falls back to sequential quicksort.

Notably, this is efficiently parallelized thanks to the absence of post-processing, which benefits greatly from \emph{affine} mutable borrowers in Pure Borrow.
Traditional move-based APIs would require post-processing to join vectors (see \nameref{sect:compare-revert} in \cref{sect:compare} for a detailed discussion).

\subsection{Benchmark Results}
\label{sect:casestudy-benchmark}

To evaluate performance, we measured the runtime and memory allocation for randomly generated integer vectors of varying sizes.
We compared our two implementations with the introsort\footnote{
  Introsort \cite{Musser97-introsort} is a variant of quicksort that cleverly switches to heap or insertion sort to improve performance.
} implementation in the \texttt{vector-algorithms} package \cite{vector-algorithms}, widely used in industry.
More concretely, we compare the following implementations:
\textbf{Introsort}, industrial-level sequential introsort from the \texttt{vector-algorithms} package;
\textbf{Sequential}, sequential quicksort, i.e., our naïve quicksort in \cref{fig:qsort-naive} with \hsinline{budget = 0}, serving as the baseline;
\textbf{Naïve \(N\)}, our naïve parallel quicksort (\cref{fig:qsort-naive}) with \(\namesty{budget} = N\), dividing the whole vector into at most \(N\) pieces;
and \textbf{WS \(N\)}, our work-stealing parallel quicksort (\cref{fig:work-stealing}) with \hsinline{thresh = 16} and \(N\) worker threads.
The benchmarks were run with the \texttt{tasty-bench} framework \cite{tasty-bench} on a MacBook Pro 2024 with an Apple M3 chip (4 performance cores and 6 efficiency cores) and 16 GB of RAM.
All cases were run with the \texttt{+RTS -N10} option, enabling parallelism on 10 cores.

\pgfplotsset{
  xlabel=\(N\), xlabel style={yshift=.45em},
  every x tick scale label/.append style={yshift=1.625em, xshift=.7em, font=\small},
  ylabel style={yshift=-.05em}
}
\begin{figure}
  \hspace{-1.0em}
  \begin{tikzpicture}
    \begin{axis}[width=.5\textwidth,
      ylabel={Wall Clock Time $[\mathrm{ms}]$}]
      \addplot[black,sharp plot] table[x=size,y=introMean,header=has colnames,col sep=comma] {bench.csv};

      \addplot[red,sharp plot] table[x=size,y=sequentialMean,header=has colnames,col sep=comma] {bench.csv};

      \addplot[blue,mark=triangle,sharp plot] table[x=size,y=parallel4Mean,header=has colnames,col sep=comma] {bench.csv};
      \addplot[green,mark=triangle,sharp plot] table[x=size,y=parallel16Mean,header=has colnames,col sep=comma] {bench.csv};
      \addplot[black,mark=triangle,sharp plot] table[x=size,y=parallel32Mean,header=has colnames,col sep=comma] {bench.csv};

      \addplot[blue,mark=+,sharp plot] table[x=size,y=workSteal4Mean,header=has colnames,col sep=comma] {bench.csv};
      \addplot[green,mark=+,sharp plot] table[x=size,y=workSteal8Mean,header=has colnames,col sep=comma] {bench.csv};
      \addplot[black,mark=+,sharp plot] table[x=size,y=workSteal10Mean,header=has colnames,col sep=comma] {bench.csv};
    \end{axis}
  \end{tikzpicture}
  \hspace{1.1em}
  \begin{tikzpicture}
    \begin{axis}[width=.5\textwidth,
      ylabel={Allocation $[\text{MB}]$}]
      \addplot[black,sharp plot] table[x=size,y=introAlloc,header=has colnames,col sep=comma] {bench.csv};

      \addplot[red,sharp plot] table[x=size,y=sequentialAlloc,header=has colnames,col sep=comma] {bench.csv};

      \addplot[blue,mark=triangle,sharp plot] table[x=size,y=parallel4Alloc,header=has colnames,col sep=comma] {bench.csv};
      \addplot[green,mark=triangle,sharp plot] table[x=size,y=parallel16Alloc,header=has colnames,col sep=comma] {bench.csv};
      \addplot[black,mark=triangle,sharp plot] table[x=size,y=parallel32Alloc,header=has colnames,col sep=comma] {bench.csv};

      \addplot[blue,mark=+,sharp plot] table[x=size,y=workSteal4Alloc,header=has colnames,col sep=comma] {bench.csv};
      \addplot[green,mark=+,sharp plot] table[x=size,y=workSteal8Alloc,header=has colnames,col sep=comma] {bench.csv};
      \addplot[black,mark=+,sharp plot] table[x=size,y=workSteal10Alloc,header=has colnames,col sep=comma] {bench.csv};
    \end{axis}
  \end{tikzpicture}\vspace{.05em}
  \begin{tikzpicture}
    \begin{customlegend}[legend columns=-1,
      legend style={draw=none, /tikz/every even column/.append style={column sep=8pt},font=\small},
      legend entries={Intro,Sequential,Naïve $4$, Naïve $16$, Naïve $32$, WS $4$, WS $8$, WS $10$}]
      \addlegendimage{black}
      \addlegendimage{red}
      \addlegendimage{blue,mark=triangle}
      \addlegendimage{green,mark=triangle}
      \addlegendimage{black,mark=triangle}
      \addlegendimage{blue,mark=+}
      \addlegendimage{green,mark=+}
      \addlegendimage{black,mark=+}
    \end{customlegend}
  \end{tikzpicture}

  \caption{The performance of parallel quicksort implementations\label{fig:bench-qsort}.}
\end{figure}

\Cref{fig:bench-qsort} shows the benchmark results.
Overall, our parallel quicksort implementations have considerably faster runtimes than the sequential quicksort, and our work-stealing versions largely outperform the naïve ones in runtime.
These results are promising and demonstrate the eligibility of our Pure Borrow APIs as building blocks for sophisticated parallelism.
Although our current parallel quicksort implementations do not outperform the industrial-level introsort, we believe that more sophisticated implementations can further improve performance.

\subsection{Summary}
\label{sect:casestudy-summary}

Our Pure Borrow framework enables writing safe, pure, and efficient algorithms with parallelized mutation concisely in Linear Haskell.
Our parallel quicksort implementations written in our Pure Borrow APIs showed promising performance.

\section{Formal Calculus and Metatheory}
\label{sect:metatheory}

We have seen above that our Pure Borrow framework is powerful, providing safe, rich abstraction for efficient mutation.
But is this powerful extension to Linear Haskell really \emph{safe}?

To answer this question, we formalize Linear Haskell extended with Pure Borrow (\cref{sect:metatheory-program}, \cref{sect:metatheory-type}) and establish a metatheory for its safety, leak freedom and confluence (\cref{sect:metatheory-metatheory}).
For purity, we develop two operational semantics, one for physical behavior and the other for mathematical modeling (\cref{sect:metatheory-opsem}), and formulate a type-based relation between the two (\cref{sect:metatheory-assoc}).
\iffull\else
The appendix on our formal calculus and metatheory is included in the extended version of this paper \cite{MatsushitaIshii26-Pure-Borrow-arXiv}.
\fi

\subsection{Program Syntax}
\label{sect:metatheory-program}

\begin{figure}
  \tightmathenv
  \begin{gather*}
    \stag{Variable} \Var \sni \var, \varB, \varC
  \hspace{3em}
    \stag{Data constructor} \Con
  \\[.3em]
    \stag{Integer operation} \iop
      \sdef \paren{{+}} \arity{2}
      \sor \paren{{\times}} \arity{2} \sor \cdots
  \hspace{3em}
    \stag{Integer relation} \irel
      \sdef \paren{{\le}} \arity{2} \sor \cdots
  \\[.3em]
  \begin{aligned}
    \stag{Operator} \op
      & \sdef \iop \sor \irel
      \sor \tpar \arity{2}
      \sor \consume \arity{1}
      \sor \move \arity{1}
      \sor \linearly \arity{1}
      \sor \withLinearly \arity{1}
  \\[-.1em] & \sskip
      \sor \newRef \arity{2} \sor \freeRef \arity{1}
      \sor \newLifetime \arity{2} \sor \tendLifetime \arity{1}
  \\[-.1em] & \sskip
      \sor \borrow \arity{2}
      \sor \share \arity{1}
      \sor \tcopy \arity{1} \sor \joinMut \arity{1}
      \sor \reclaim \arity{2} \sor \execBO \arity{2}
  \end{aligned}
  \\[.3em]
  \begin{aligned}
    \stag{Monad constr.} \mo
      & \sdef \pure \arity{1} \sor \paren{{\bind}} \arity{2}
      \sor \sexecBO \arity{2}
      \sor \parBO \arity{2}
      \sor \deref \arity{1} \sor \updateRef \arity{2}
  \end{aligned}
  \\[.3em]
  \begin{aligned}
    \stag{Term} \Term \sni \term, \termB
      & \sdef \var \sor \klet \Seq{\var \teq \term} \ktin \termB
      \sor \tlamx{\var} \term
      \sor \term\kd \termB
      \sor \kseq \var\kd \term
      \sor \num
  \\[-.3em] & \sskip
      \sor \Con\kd \seq{\term}
      \sor \kcase \term \kof
        \kcurly{ \Seq{\Con\kd \seq{\var} \tto \termB} }
      \sor \op\kd \seq{\term}
      \sor \mo\kd \seq{\term}
  \end{aligned}
  \end{gather*}
  \caption{Program syntax.}
  \label{fig:program}
\end{figure}

\Cref{fig:program} shows the program syntax.\footnote{
  The small numbers (e.g., \(2\) in \(\paren{{+}} \arity{2}\)) show the arity.
  For simplicity, we don't have partial application to \(\op\), \(\mo\), and \(\Con\).
}
We keep the language minimal while supporting the core functionalities of borrowing.
Values of the BO monad are constructed with monad constructors \(\mo\).

The biggest difference from \cref{sect:overview} is that, for mutable state, our language supports only references \hsinline{Ref} and omits vectors \hsinline{Vector} for simplicity.
Note that we can still emulate the functionality of \hsinline{Vector a} by a list of references \hsinline{[Ref a]}.

\subsection{Type System}
\label{sect:metatheory-type}

\begin{figure}
  \tightmathenv
  \begin{gather*}
    \stag{Lifetime variable}\nk5 \lftvar, \lftvarB
  \hspace{2.75em}
    \stag{Lifetime id}\nk5 \lftid
  \hspace{2.75em}
    \stag{Lifetime} \lft, \lftB
      \sdef \al \lftid
      \sor \static
      \sor \lft \lftand \lftB
      \sor \lftvar
  \\[.3em]
    \stag{Multiplicity variable}\nk5 \multvar, \multvarB
  \hspace{3em}
    \stag{Multiplicity} \textstyle
    \mult, \multB
    \sdef \one \sor \many \sor \prod \Seq{\multvar}
  \\[.3em]
    \stag{Type variable} \Tvar, \TvarB
  \hspace{3em}
    \stag{Variable}
    \anyvar \sdef \lftvar \sor \lftid \sor \multvar \sor \Tvar
  \\[.3em]
    \stag{Type constructor} \Tcon
  \hspace{3em}
    \stag{Borrower constructor} \Bcon \sdef \Mut \sor \Share
  \\[.3em]
  \begin{aligned}
    \stag{Type} \Type \sni \Ty, \TyB
      & \sdef \all{\anyvar} \Ty \sor \Tvar
      \sor \Ty \to_\mult \TyB
      \sor \Tcon\, \seq{\Ty}
      \sor \Int
      \sor \Linearly
  \\[-.1em] & \sskip
      \sor \kRef \Ty
      \sor \Now^\lft \sor \End^\lft
      \sor \Bcon^\lft\kc \Ty
      \sor \Lend^\lft\kc \Ty
      \sor \BO^\lft\kc \Ty
  \end{aligned}
  \\[.3em]
    \Ty \lto \TyB \kg\defeq\kg \Ty \to_\one \TyB
  \hspace{2.5em}
    \Ty \to \TyB \kg\defeq\kg \Ty \to_\many \TyB
  \\[.3em]
    \stag{Data type declaration} \datatemplate
  \hspace{2.75em}
    \stag{Typing context} \Tctx, \TctxB \sdef
      \Seq{\var \ccoll{\mult} \Ty}
  \end{gather*}
  \caption{Syntax for the type system.}
  \label{fig:type-syntax}
\end{figure}

We also formalize the type system.
This is largely straightforward, following the approach of \citet{BernardyBNJS18-Linear-Haskell}.
\Cref{fig:type-syntax} presents the syntax for it.\footnote{
  For each type constructor \(\Tcon\), we assume a data type declaration, which can be recursive.
}
The function type \(\to_\mult\) is parameterized by the multiplicity \(\mult\), representing linear (\(\one\)) or unrestricted (\(\many\)) binding.
For simplicity, we do not introduce type classes, and
the linearity witness \(\Linearly\) is just a linear data type.
Notably, we support first-class polymorphism,  including rank-2 polymorphism.
Typing rules are quite straightforward formalization of \cref{sect:overview-details}.
Please see \appendixref{appx:formal-type} for the details.

\subsection{Two Operational Semantics}
\label{sect:metatheory-opsem}

Why is Linear Haskell extended with Pure Borrow still \emph{pure}, while supporting flexible mutation with Rust-style borrowing?

To answer this question, we introduce \emph{two} execution models for our formal calculus.
A natural one is \emph{mutative} operational semantics, which explicitly owns the global memory state and destructively mutates the state as the computation progresses.
What is notable is the other one, \emph{denotational} operational semantics.
It is free of the global memory state and instead locally manages information about each part of the state according to ownership.
This localizes mutation and executes much like purely functional languages.
Indeed, denotational operational semantics is designed to be \emph{strongly confluent} (\cref{thm:dos-confl}, \cref{sect:metatheory-metatheory}),
and thus serves as the witness for the \emph{purity} of our Pure Borrow.
Later, we establish an \emph{association system} (\cref{sect:metatheory-assoc}), or a type-based relation between the two execution models, and state that the relation is actually a \emph{bisimulation} (\cref{thm:bisim}, \cref{sect:metatheory-metatheory}), which entails that well-typed programs behave the same way in the two execution models.
This overall scheme largely follows the purity proof for original Linear Haskell by \citet{BernardyBNJS18-Linear-Haskell},\footnote{
  The ordinary and denotational dynamic semantics of \citet{BernardyBNJS18-Linear-Haskell} correspond to our mutative and denotational operational semantics, respectively.
}
but further introduces a new mechanism for justifying Rust-style borrowing: \emph{histories}.

\paragraph{Histories for borrowing}

As explained in \cref{sect:overview-taste} (and illustrated in \cref{fig:borrow-illust}), Rust-style borrowing is great in that no direct communication is needed for borrowers to return ownership to the lender.
However, this poses a challenge in building a \emph{pure} execution model.
How can the lender know the final result of mutation by the borrowers, when it reclaims the ownership?

To solve this, we propose a new approach: recording the \emph{history} of updates on borrowed data.\footnote{
  Other studies used different methods, such as prophecies, to model mutable borrowing functionally.
  Please refer to \cref{sect:related} for a comparison of our history-based approach with these existing methods.
}
More concretely, our denotational operational semantics goes as follows.
Because each mutation to borrowed data (more specifically, \(\updateRef\)) happens inside the BO monad, executed with \(\execBO\)/\(\sexecBO\) threading live lifetime tokens \(\Now^\lft\), we can record the mutation history to the tokens.
The history can then be inherited by dead lifetime tokens \(\End^\lft\) at \(\tendLifetime\).
When the lender \(\reclaim\)s the ownership, we can restore the final version of the borrowed content from the history saved in the dead lifetime token.
This successfully models global communication from the borrowers to the lender in a \emph{local} way.

\paragraph{Overview}

\begin{figure}
  \subcaptionsetup{aboveskip=1pt}
  \begin{subcaptionblock}{\textwidth}
    \tightmathenv
    \begin{gather*}
      \stag{Value} \val
        \sdef \tlamx{\var} \term
        \sor \num
        \sor \Con\kd \seq{\var}
        \sor \mo\kd \seq{\var}
    \\[-.1em]
      \stag{Evaluation context} \Ktx
        \sdef \khole \sor \khole\kd \var
        \sor \kseq \khole\kd \var
        \sor \kcase \khole \kof
          \kcurly{ \Seq{\Con\kd \seq{\var} \tto \term} }
        \sor \op\kd \seq{\var}\kd \khole\kd \seq{\varB}
    \end{gather*}\vspace{-.5em}
    \caption{Common things.}
    \label{fig:opsem-syntax-common}
  \end{subcaptionblock}
  \begin{subcaptionblock}{.46\textwidth}
    \tightmathenv
    \begin{gather*}
      \stag{Location}\kx \loc
    \hspace{2em}
      \stag{Memory}\kx \Mem \sdef \Seq{\loc \mMapsto \var}
    \\[.1em]
      \stag{Extra value} \mxval
        \sdef \nul
        \sor \kRef \loc
        \sor \Done\kd \var
    \\[.1em]
      \stag{Extra operator} \mop
        \sdef \cdots
    \\[.1em]
      \stag{Term} \mterm, \mtermB
        \sdef \term
        \sor \mxval
        \sor \mop\kd \seq{\var}
    \\[.1em]
      \stag{Value} \mval, \mvalB
        \sdef \val \sor \mxval
    \\[.1em]
      \stag{Evaluation context} \MKtx
        \sdef \Ktx
        \sor \mop\kd \seq{\var}\kd \khole\kd \seq{\varB}
    \\[.1em]
      \stag{Env.}\nk8 \MEnv \ky\sdef\ky \Seq{\var \teq \mterm}
    \hspace{1.6em}
      \stag{Config.}\nk8 \MCfg \ky\sdef\ky \MEnv; \Mem; \var
    \end{gather*}\vspace{-.5em}
    \caption{For mutative semantics.}
    \label{fig:opsem-syntax-mos}
  \end{subcaptionblock}
  \hspace{.3em}
  \begin{subcaptionblock}{.5\textwidth}
    \tightmathenv
    \begin{gather*}
      \stag{Borrower constr.} \DBcon
        \sdef \Mut^{\seq{\bpath}} \sor \Share
    \\[.1em]
    \begin{aligned}
      \stag{Extra value} \dxval
        & \sdef \nul \sor \nul_\Hist
        \sor \kRef \var
    \\[-.6em] & \sskip
        \sor \Lend^\bid\kb \var
        \sor \Done_\Hist\kc \var
    \end{aligned}
    \\[.1em]
      \stag{Extra operator} \dop
        \sdef \cdots
    \\[.1em]
      \stag{Term} \dterm, \dtermB
        \sdef \term
        \sor \dxval
        \sor \dop\kd \seq{\var}
        \sor \DBcon\kd \dterm
    \\[-.1em]
      \stag{Value} \dval, \dvalB
        \sdef \val \sor \dxval
        \sor \DBcon\kd \dval
    \\[.0em]
      \stag{Evaluation ctx.}\kx \DKtx
        \sdef \Ktx
        \sor \dop\kd \seq{\var}\kd \khole\kd \seq{\varB}
        \sor \DBcon\kd \DKtx
    \\[.0em]
      \stag{Env.}\nk8 \DEnv \ky\sdef\ky \Seq{\var \teq \dterm},\ka \Seq{\bid}
    \hspace{1.6em}
      \stag{Config.}\nk8 \DCfg \ky\sdef\ky \DEnv; \var
    \end{gather*}\vspace{-.5em}
    \caption{For denotational semantics. Also see \cref{fig:opsem-syntax-hist}.}
    \label{fig:opsem-syntax-dos}
  \end{subcaptionblock}
  \caption{Core syntax for the two operational semantics.}
  \label{fig:opsem-syntax}
\end{figure}

\begin{figure}
  \tightmathenv
  \begin{gather*}
    \stag{Borrow id} \BorId \sni \bid
  \hspace{3em}
    \stag{Borrow path} \BorPath \sni \bpath, \bpathB \sdef \bid \sor \bpath.\idx
  \\[.1em]
    \stag{History} \dHist \sni \Hist \sdef \Seq{\bpath \hstore \var}
  \end{gather*}
  \caption{Syntax for histories in denotational semantics.}
  \label{fig:opsem-syntax-hist}
\end{figure}

\begin{figure}
  \tightmathenv
  \begin{gather*}
    \begin{aligned}
      \Env \ \ \ &\defeq\ \ \
        \var = 3,\kc
        \varMain \teq \execBO\kd \varNow\kd
          (\modifyRef\kd (\tlamx{\var} \var + 4)\kd \varRef)
    \\[-.1em]
      \Env' \ \ \ &\defeq\ \ \
        \var = 3,\kc
        \varB = 7,\kc
        \varMain \teq (\varNow', \varRef')
    \end{aligned}
  \\[.3em]
    \begin{matrix*}[r]
      \varRef \teq \kRef \loc,\kc
      \varNow \teq \nul,\kc \Env;\kd
        \loc \mMapsto \var;\kd
        \varMain
    &
      \varRef \teq \Mut^\bid\kb (\kRef \var),\kc
      \varNow \teq \nul_\emphist,\kc \Env;\kd
        \varMain
    \\[.3em]
      \to^*
      \varRef' \teq \kRef \loc,\kc
      \varNow' \teq \nul,\kc \Env';\kd
        \loc \mMapsto \varB;\kd
        \varMain
    & \hspace{1.0em}
      \to^*
      \varRef' \teq \Mut^\bid\kb (\kRef \varB),\kc
      \varNow' \teq \nul_{\bid \hstore \varB},\kc \Env';\kd
        \varMain
    \end{matrix*}
  \end{gather*}
  \caption{Reduction example in mutative and denotational semantics.}
  \label{fig:reduce-example}
\end{figure}

Now that we have a big picture, let's dive into the two operational semantics.
\Cref{fig:opsem-syntax,fig:opsem-syntax-hist} present the syntax for them.
More details are presented in \appendixref{appx:formal-opsem}.

To model Haskell's lazy, call-by-need evaluation, both semantics are formulated in the style of \citet{Launchbury93-opsem-lazy}, where the environment consists of variable-term bindings.
Notably, the operational semantics has a nondeterministic evaluation order, as clarified by the syntax of evaluation contexts (\(\Ktx\) etc.), where the arguments of the operators (\(\op\) etc.) are evaluated in \emph{parallel}.
Despite this freedom, our language satisfies the confluence and behavior uniqueness, as discussed later (\cref{thm:dos-confl,cor:mos-behav-uniq}, \cref{sect:metatheory-metatheory}).

One major difference between the two semantics is that mutative semantics is equipped with the global memory \(\Mem\), while denotational semantics lacks it.
A reference \(\TRef\) is modeled as the location \(\loc\) in the former and the content \(\var\) in the latter.

To support borrowing, denotational semantics handles the \emph{history} \(\Hist\), which is defined in \cref{fig:opsem-syntax-hist}.
A \emph{borrow id} \(\bid\) uniquely identifies a borrow created with each \(\borrow\).
A \emph{borrow path} \(\bpath\) represents a sub-component of a borrowed object.
A \emph{history} \(\Hist\) records updates \(\bpath \hstore \var\) of the content of a reference at \(\bpath\) into \(\var\) (by \(\updateRef\)).\footnote{
  A history only records the latest update for each borrow path \(\bpath\).
  So a history can have at most one item of the form \(\bpath \hstore \var\) for each borrow path \(\bpath\).
  The order of items in a history is ignored.
  We write \(\emphist\) for the empty history.
}
In denotational semantics, live and dead lifetime tokens are models as \(\nul_\Hist\) with a history \(\Hist\), and each mutable borrower puts on \(\Mut^{\seq{\bpath}}\), where \(\seq{\bpath}\) is initially set to the borrow id \(\bid\) just after \(\borrow\).
Notably, we can freely handle nested borrowers (i.e., borrowers that contain borrowers), thanks to the multiple paths \(\seq{\bpath}\) in the constructor \(\Mut^{\seq{\bpath}}\) (each path \(\bpath_{\!\idxB}\) corresponds to a different layer of borrowing).
Arbitrary (possibly recursive) algebraic data types within borrowers are also supported, using field accesses \(.\idx\) in borrow paths \(\bpath\) for splitting (\cref{fig:apis-split}).\footnote{
  As for functions, our calculus does not offer special operations on mutable borrowers over functions.
  Extending Pure Borrow with a feature like Rust's \rsinline{FnMut} is left for future work.
}

Based on this syntax, we introduce the small-step reduction relations \(\MCfg \to \MCfg'\) and \(\DCfg \to \DCfg'\) for the two operational semantics.
The reduction rules are rather straightforward, with the big picture in mind.
Please refer to \appendixref{appx:formal-opsem} for the details.
\Cref{fig:reduce-example} presents a reduction example in mutative and denotational semantics,\footnote{
  Here, \(\modifyRef\) is a function derived from \(\updateRef\), just like \hsinline{modifyAt} in \cref{fig:apis-vec}, \cref{sect:overview-details}.
  For readability, we omit some irrelevant variables in the environments.
  For both semantics, the variables \(\varRef\) and \(\varNow\) are actually included in the final environment, with the same binding as the initial environment.
}
showcasing how mutation works.

\subsection{Association System}
\label{sect:metatheory-assoc}

Mutative and denotational operational semantics are designed so that their execution is parallel and related, as informally observed in \cref{fig:reduce-example}.
We formalize this idea into what we dub the \emph{association system}.
Roughly speaking, it is a variant of the type system that is dynamic---i.e., considering intermediate execution steps---and relational---i.e., associating the two execution models.

The high-level association judgment we provide is \(\MCfg \assoc_\Ty \DCfg\), stating that the mutative and denotational configurations \(\MCfg\), \(\DCfg\) are related, under the return type \(\Ty\).
To define this, we introduce a term-level association judgment \(\Rctx \kvdash \mterm \assoc \dterm \kccol \RTy\), which is defined compositionally much like the typing judgment \(\Tctx \kvdash \term \kccol \RTy\).
Here, \(\Rctx\) and \(\RTy\) are an extended version of \(\Tctx\) and \(\Ty\), where \(\Rctx\) can contain some `ghost resources' that account for advanced types (e.g., \(\gmut^{\seq{\bpath}}\) for \(\Mut^\lft\kc \Ty\)).
The design of our association system is highly subtle, primarily due to lazy evaluation, first-class polymorphism, and the complex nature of Rust-style borrowing.
Please refer to \appendixref{appx:assoc} for the details.

\subsection{Metatheory}
\label{sect:metatheory-metatheory}

Finally, we introduce our metatheory for Pure Borrow, which works toward establishing memory safety, leak freedom and confluence (or behavior uniqueness).
With properly designed formalization, especially denotational operational semantics and the association system, we can aim for the desired properties.
Please refer to \appendixref{appx:metatheory} for the omitted details.

\paragraph{Strong confluence of denotational operational semantics}

Remarkably, by design, our denotational operational semantics is \emph{strongly confluent}, i.e., satisfies the \emph{diamond property}:
\begin{theorem}[Strong confluence of denotational operational semantics]
\label{thm:dos-confl}
  For any denotational configurations \(\kb\DCfg\), \(\DCfg_1\) and \(\kb\DCfg_2\) such that \(\kb\DCfg_1 \ne \DCfg_2\), \(\kb\DCfg \to \DCfg_1\) and \(\kb\DCfg \to \DCfg_2\), there exists a configuration \(\kb\DCfg_+\) such that \(\kb\DCfg_1 \to \DCfg_+\) and \(\kb\DCfg_2 \to \DCfg_+\).
\end{theorem}

\paragraph{On the association system}

Crucially, well-typedness implies association:
\begin{theorem}[Well-typedness implies association]
\label{thm:typed-assoc}
  If \(\vdash \term \ccol \Ty\), then \(\paren{\var \teq \term;; \var} \assoc_\Ty \paren{\var \teq \term; \var}\).
\end{theorem}

We use the term \emph{safety} on configurations for the existence of an association.
More concretely, a mutative configuration \(\MCfg\) is safe if \(\MCfg \assoc_\Ty \DCfg\) for some \(\DCfg\) and \(\Ty\),
and a denotational configuration \(\DCfg\) is safe if \(\MCfg \assoc_\Ty \DCfg\) holds for some \(\MCfg\) and \(\Ty\).

Remarkably, safety ensures \emph{leak freedom} for a mutative configuration:
\begin{theorem}[Safety ensures leak freedom]
\label{thm:safe-leakfree}
  If \(\var \teq \mval, \MEnv; \Mem; \var\) is safe, then \(\Mem\) is empty.
\end{theorem}

\paragraph{Progress and bisimulation}

Crucially, we expect the association system to satisfy \emph{progress} and \emph{bisimulation}.
This is a relational version of the well-known progress and preservation, the syntactic approach to type soundness \cite{WrightF94-syntax-type-sound}.

We carefully designed our formalization to satisfy these properties and wrote rough proof sketches for them.
That said, they are currently conjectures, whose rigorous proofs are left as future work and possibly require minor fixes to the association system.

\begin{conjecture}[Safety ensures progress]
\label{thm:progress}
  For any safe configuration \(\MCfg\) that is not a normal form,
  there exists \(\kb\MCfg'\) such that \(\kb\MCfg \to \MCfg'\).
  Also, for any safe configuration \(\kb\DCfg\) that is not a normal form,
  there exists \(\kb\DCfg'\) such that \(\kb\DCfg \to \DCfg'\).
\end{conjecture}

\begin{conjecture}[Association is a bisimulation]
\label{thm:bisim}
  For any type \(\Ty\), the relation \(\assoc_\Ty\) is a bisimulation with respect to mutative and denotational operational semantics.
\end{conjecture}

Bisimulation clearly implies preservation of safety:
\begin{corollary}[Preservation of safety]
\label{cor:preserv}
  Assume \cref{thm:bisim}.
  If \(\kb\MCfg\) is safe and \(\kb\MCfg \to \MCfg'\) holds, then \(\ka\MCfg'\) is safe.
  Analogously, if \(\kb\DCfg\) is safe and \(\kb\DCfg \to \DCfg'\) holds, then \(\kb\DCfg'\) is safe.
\end{corollary}

\paragraph{Behavior uniqueness}

Finally, with the bisimulation (\cref{thm:bisim}) and the confluence of denotational operational semantics (\cref{thm:dos-confl}), we can establish the \emph{behavior uniqueness} of mutative operational semantics under safety.\footnote{
  One may wonder if the \emph{strong confluence} of mutative operational semantics under safety derives from \cref{thm:bisim} and \cref{thm:dos-confl}.
  Unfortunately, that is not the case, because the association \(\assoc_\Ty\) is not left-unique.
  Indeed, the strong confluence does not hold in a usual sense;
  for example, two reductions of distinct \(\newRef\)s may unfortunately allocate the same location \(\loc\).
  So instead, we establish the uniqueness of behavior.
}

\begin{definition}[Behavior in mutative operational semantics]
\label{def:mos-behav}
  We say \(\MCfg\) may be non-terminating if there exists an infinite reduction sequence from \(\MCfg\).
  We say \(\MCfg\) may terminate if there exists a reduction sequence from \(\MCfg\) to some normal form.
  We say \(\MCfg\) may return \(\num\) if there exists a reduction sequence from \(\MCfg\) to some configuration of the form \(\var \teq \num,\ka \MEnv;\kb \Mem;\kb \var\).
\end{definition}

\begin{corollary}[Behavior uniqueness in mutative operational semantics]
\label{cor:mos-behav-uniq}
  Assume \cref{thm:bisim}.
  Take any safe configuration \(\MCfg\).
  It is not the case that \(\kb\MCfg\) may be non-terminating and also may terminate.
  If \(\kb\MCfg\) may return \(\num\) and \(\numB\), then \(\num = \numB\).
\end{corollary}
\begin{proof}
  By \cref{thm:bisim} and \cref{thm:dos-confl}.
\end{proof}

\section{Comparison with Existing Approaches}
\label{sect:compare}

We compare our Pure Borrow with other existing approaches to borrowing and related mechanisms for flexible resource management.

\paragraph{Programmable lifetime control}

Rust's borrow checker is highly automated but is largely a black box hard to intervene with.
In contrast, Pure Borrow provides programmers with great control of lifetimes.
A core idea is to use lifetime tokens \hsinline{Now^α}, \hsinline{End^α} that represent lifetime state, inherited from the semantic model of RustBelt \cite{JungKD18-RustBelt}.
Thanks to Haskell's advanced features, such as first-class polymorphism, programmers can build abstractions as observed in \cref{sect:overview}.
That said, more automation in Pure Borrow would also be helpful, which is left for future work.

\paragraph{First-class lenders}

In Rust, lenders are not first-class citizens and cannot be passed as arguments to or returned from functions.
This imposes practical limitations, e.g., in handling self-borrowing data types.
In contrast, our Pure Borrow treats lenders as first-class citizens, providing the first-class lender type \hsinline{Lend^α T}.
First-class lenders are quite underexplored so far,\footnote{
  First-class lenders were partly supported by \citet{NakayamaMSSK24-borrow-consort} in the context of type-based automated verification, but they did not even support nested references.
}
and it would be interesting to explore their practical usage.
For example, Pure Borrow can represent an integer vector self-borrowing some element of it as \hsinline{∃α. (Mut^α Int, Lend^α (Vector Int), Now^α)}.

\paragraph{Leak freedom}

In Rust, types are affine, not linear.
One can accidentally cause leaks in Rust, for example by constructing cyclic references under reference-counting GC \cite{BenYehuda15-JoinGuard-leak}---a finding nicknamed as \emph{Leakpocalypse}.
A possible improvement is to keep track of leakable types with a special trait \rsinline{T: Leak},\footnote{
  For example, integers \rsinline{i32} and mutable references \rsinline{&mut U} satisfy \rsinline{Leak}.
  Leak freedom would then be guaranteed by imposing \rsinline{T: Leak} on smart pointer types under reference-counting GC such as \rsinline{Rc<T>} and \rsinline{Arc<T>}.
}
but this has not been adopted so far for technical reasons \cite{Matsakis25-Forget}.
In contrast, our Pure Borrow is carefully designed to guarantee leak freedom.\footnote{
  To be fair, Rust appears to be leak-free when restricted to the subset of features currently supported by Pure Borrow.
  Still, we are not aware of any existing work that formally establishes the leak freedom of Rust for such a subset, and we conjecture that Pure Borrow could be safely extended to features like reference-counting GC with leak freedom by introducing a type class for leakable types.
}

\paragraph{Purity}

A crucial design choice of our Pure Borrow is that it focuses on a pure fragment.
As a result, unlike Rust, it does not support mutable state shared across threads such as a mutex, because the result of concurrent mutation depends on nondeterministic thread scheduling.
Still, this does not preclude sequentially shared mutable state, like \hsinline{Vector_ST^s a} in the ST monad (\cref{sect:background-st}).
Indeed, we observe that Pure Borrow can support shared mutable state by requiring a non-shareable token for each mutation to such state, an approach akin to linear constraints in Haskell \cite{SpiwackKBWE22-linear-constraint, SpiwackKBWE26-linear-constraint-arXiv} and GhostCell in Rust \cite{YanovskiDJD21-GhostCell}.

\paragraph{Oxidized OCaml}

\newcommand*{\modesty}[1]{\textsf{\textbf{#1}}}
\newcommand*{\leup}{\rotatebox{90}{\(\le\)}}
\begin{figure}
  \centering\vspace{-.3em}
  \begin{tikzpicture}[
    tagnode/.style = {font = \sffamily\small}]
    \matrix[matrix of nodes,
      row 1/.style = {row sep = -.35em},
      row 2/.style = {row sep = .5em},
      row 3/.style = {row sep = -.5em},
      column sep = -.1em,
      column 1/.style = {column sep = .8em},
      column 4/.style = {column sep = .8em}] {
      & \node[tagnode]{comonad}; & & & &
        \node[tagnode, yshift=.1em]{\clap{comonad}}; \\
      \sffamily Linearity & \modesty{many} & \(\le\) & \modesty{once} &
        \hsinline{Ur a -o a} & &
        \hsinline{Ur a -o Ur (Ur a)} \\
      \sffamily Uniqueness & \modesty{unique} & \(\le\) & \modesty{aliased} &
        \hsinline{a -o (Mut^α a, Lend^α a)} & &
        \hsinline{Mut^α (Mut^β a) -o Mut^α∧β a} \\
      & & & \node[tagnode]{monad}; & &
        \node[tagnode, yshift=-.1em]{\clap{monad-like}}; \\
    };
  \end{tikzpicture}\vspace{-.3em}
  \caption{Linearity and uniqueness axes in OxCaml and Pure Borrow.}
  \label{fig:linear-unique}
\end{figure}

Recently, inspired by the success of Rust, there is a movement to ``oxidize'' OCaml, i.e., strengthen OCaml with resource-related information for better safety and performance \cite{LorenzenWDEL24-OxCaml, GeorgesPEWDECPD25-OxCaml}.
This idea has been embodied as a practical language extension, OxCaml \cite{OxCaml}, already in use for production code.

While this has a similar spirit to our Pure Borrow, their goal is quite different.
First, they do not care much about rigorous purity, as OCaml allows side effects anywhere, unlike Haskell.
Also, they aim to keep annotations for resources very lightweight, even maintaining backward compatibility with legacy OCaml.
They add coarse-grained labels about resources, called \emph{modes}.
We note that their lifetime-related modes are very coarse-grained: just \modesty{local} and \modesty{global}.
They provide an API for temporarily using a \modesty{global} object \modesty{local}ly, which is somewhat like simple borrowing but without fine-grained scoping.
Rust-style borrowing of our Pure Borrow is much more expressive and fine-grained, especially with the support of reborrowing.

An interesting high-level comparison between our approach and theirs can be made in the view of linearity and uniqueness, as summarized in \cref{fig:linear-unique}.
As clarified by \citet{MarshallVO22-linear-unique}, linearity is the inability to create (and drop) an alias, and uniqueness is the absence of aliases.
OxCaml introduces modes \modesty{many} vs. \modesty{once} for linearity and \modesty{unique} vs. \modesty{aliased} for uniqueness.
\citet{LorenzenWDEL24-OxCaml} point out that \modesty{many} is a comonad and \modesty{aliased} is a monad.
Linear Haskell has the linearity axis, with \hsinline{Ur} being a comonad and an exact counterpart of \modesty{many}.
Notably, Rust-style borrowing of our Pure Borrow can be understood as a mechanism for allowing aliases, akin to \modesty{aliased}, but with access separated by lifetimes.
Borrowers \hsinline{Mut^α}/\hsinline{Share^α} and lenders \hsinline{Lend^α} assert a possible presence of lifetime-delimited aliases---a lender for \hsinline{Mut^α}, borrowers for \hsinline{Lend^α}, and a lender and shared borrowers for \hsinline{Share^α}.
Roughly speaking, borrowing \hsinline{borrow} is like monadic introduction, and borrower joining \hsinline{joinMut} is like monadic joining.
Still, \hsinline{Mut^α} is not exactly an (indexed) monad, because it is not a functor, being invariant in subtyping (\cref{fig:subty}, \cref{sect:overview-details}).
This view of borrowing as aliasing provides an exciting perspective for future work.

\paragraph{Splitting with reversion}
\label{sect:compare-revert}

\citet{SpiwackKBWE22-linear-constraint} proposed an approach to splitting ownership with a mechanism for reverting that.
They introduce, for example, a vector type \hsinline{Vector a n} with a phantom parameter \hsinline{n} for an abstract address.
In their API, the vector can be split into two vectors \hsinline{Vector a l} and \hsinline{Vector a r} for some \hsinline{l} and \hsinline{r} under a peculiar constraint \hsinline{Slices n l r}, and they can be later joined back under that constraint.
The core drawback of this approach is that the structure of a split should be \emph{statically} determined and tracked, so that phantom post-processing can be performed later to revert the split and join tokens.
In contrast, our Rust-style borrowing enables dynamic, flexible splitting of ownership, without post-processing for reversion.

\paragraph{Destination passing}

Destination passing \cite{Minamide98-data-hole} is a long-known idea for achieving efficient mutation in functional languages, and its type-safe formulation has been recently explored \cite{Bagrel24-destination-Haskell, BagrelS25-destination-calculus}.
At a high level, destination passing works like borrowing:
the type \(S \ltimes (\floor{T_1} \otimes \cdots \otimes \floor{T_n})\) roughly denotes an object \(S\) lending its components \(T_1, \ldots, T_n\) as holes.
Its core difference from our Pure Borrow is that one should directly modify the lender type (on the right-hand side of \(\ltimes\)) each time one returns ownership to fill a hole.

\paragraph{Rust-style borrowing in functional languages}

\citet{WagnerGMLA25-linearity-to-borrowing} presented how to augment a linear type system of an ML-style functional language to support Rust-style borrowing.
They also proved the soundness of the resulting type system semantically using a new separation logic that supports borrowing.
This work is interesting, and their high-level goal is relevant to ours.
However, they did not clarify the purity of Rust-style borrowing, as their calculus is impure and their metatheory does not account for functional behavior.
Also, they developed a new type system with special typing rules for borrowing and lifetimes, whereas our Pure Borrow works in existing Linear Haskell without modification.

\paragraph{Monadic regions}

Some studies on monadic regions \cite{FluetM04-monad-region,FluetM06-monad-region-journal,KiselyovS08-light-monad-region} introduced an extension of the ST monad with region inclusion for better resource management.
While they lack leak freedom, parallelism and borrowing, their monad is, at a high level, akin to our \hsinline{BO} monad, with region inclusion corresponding to (the dual of) lifetime inclusion.

\section{Related and Future Work}
\label{sect:related}

\paragraph{Pure model of mutable borrowing}

Our use of histories for purely modeling mutable borrowing is simple but novel, and interestingly different from existing known approaches.

A notable existing technique proposed by RustHorn \cite{MatsushitaTK20-RustHorn, MatsushitaTK21-RustHorn-journal} is to model Rust's mutable borrowing functionally by \emph{prophecies} \cite{AbadiL88-refinement}.
Prophecies are a kind of dual of histories and allow getting information about the \emph{future}, instead of the past, in advance.
In RustHorn, each mutable borrower is modeled using a prophecy that `foresees' the final value of the borrowed content.
Prophecies can model mutable borrowing in a compact way and thus are useful in automated functional verification, as demonstrated by RustHorn and its descendant Creusot \cite{DenisJM22-Creusot}.
However, prophecies are not suitable for our denotational operational semantics, because they require nondeterminism and a precise analysis of the time point at which the ownership is released.

Another existing approach is to translate Rust programs into purely functional languages.
This has been explored by Electrolysis \cite{Ullrich16-Electrolysis} and Aeneas \cite{HoP22-Aeneas, Ho24-Aeneas}.
However, they are quite restrictive in terms of the supported patterns of mutable borrowing.\footnote{
  Aeneas cannot handle nested mutable borrowers (e.g., \rsinline{&mut &mut i32}) as inputs or outputs of functions \cite[Chapter 10]{Ho24-Aeneas}.
  Electrolysis also has this restriction \cite[\S 6.2]{Ullrich16-Electrolysis}.
}
Unlike such existing approaches, our model of mutable borrowing using histories can purely model arbitrary borrowing patterns.

\paragraph{Type soundness proof of Rust-style borrowing}

Our metatheory works toward proving type soundness using an association system, a binary version of a type system, and establishes progress and bisimulation (a variant of preservation) in the style of \citet{WrightF94-syntax-type-sound}.
At a high level, this strategy is similar to the correctness proof of RustHorn \cite{MatsushitaTK20-RustHorn, MatsushitaTK21-RustHorn-journal}, which proved the equivalence of Rust programs and their prophecy-based semantics through a bisimulation.
Still, technically, our proof strategy is quite different from theirs.
We deal with \emph{histories} instead of prophecies, and also tackle advanced features of Haskell, such as first-class polymorphism and lazy evaluation, unlike for Rust.

RustBelt \cite{JungKD18-RustBelt} and its descendants \cite{DangJKD20-RustBelt-relaxed, MatsushitaDJD22-RustHornBelt, GaherSJKD24-RefinedRust, MatsushitaT25-Nola} give a semantic, extensible soundness proof for Rust's ownership type system extended with some APIs based on the technique of semantic typing, or logical relations \cite{TimanyKDB24-logical-type-sound}, using the separation logic Iris \cite{JungSSSTBD15-Iris, JungKJBBD18-Iris-ground-up}.
Finding such a semantic soundness proof for our Pure Borrow framework is a major challenge left for future work.
All known separation logics that support borrowing, including Iris, are fundamentally affine and do not guarantee leak freedom.
There exists a leak-free separation logic that supports invariants \cite{BizjakGKB19-Iron}, but not yet one that supports borrowing.
Also, little exploration has been done on separation logic supporting lazy evaluation.

\paragraph{Proving purity}

Our metatheory \cref{sect:metatheory-metatheory} works toward proving the behavior uniqueness of the physical execution model, mutative operational semantics, via a bisimulation with a special execution model, denotational operational semantics, which is designed to be strongly confluent.
However, it does not completely establish purity in a strict sense, because denotational operational semantics has a subtle side effect: generation of a fresh borrow id \(\bid\).
The effect of id generation has been theoretically studied for \(\nu\)-calculus \cite{PittsS93-nu-calculus}, exploring non-trivial program equations (contextual equivalences) as well as game semantics \cite{AbramskyGMOS04-nu-calculus-game}.
As for Pure Borrow, we conjecture that purity holds despite id generation, thanks to access to each id \(\bid\) being restricted to the mutable borrowers and the lender, which can live only inside the linear thunk for the borrow.
A significant challenge for future work is to formulate what purity under linearity is in the first place and prove that Pure Borrow indeed satisfies it.

\paragraph{Equational theory}

One benefit of the purity of computation is strong support for an equational theory, which establishes non-trivial contextual equivalences between programs.
It is left for future work to develop a good equational theory for our Pure Borrow framework, prove its soundness, and explore optimization and verification based on the theory.

One challenge is that the denotations of programs naïvely induced from our denotational operational semantics \cref{sect:metatheory-opsem} (or the denotational dynamic semantics of \citet{BernardyBNJS18-Linear-Haskell}) do not satisfy full abstraction (i.e., contextually equivalent programs can have different denotations), due to linearity constraints on well-typed contexts.\footnote{
  For example, let us consider two functions \hsinline{f :: \\ !_ -> 0} and \hsinline{g :: \\ !_ -> 42} of the type \hsinline{Vector () -> Int}.
  They would have different denotations naïvely in our denotational semantics.
  However, they are contextually equivalent under the linear type system, because the type system disallows sharing a linear vector \hsinline{Vector ()} under an unrestricted arrow \hsinline{->} (rather than \hsinline{-o}), making it impossible to call \hsinline{f} and \hsinline{g} to distinguish them.
}

\paragraph{Composability with other linear effects}

It is interesting to explore what kind of linear effects our \hsinline{BO} monad can compose with.
This is non-trivial, especially because \hsinline{BO} supports a \emph{parallel} execution primitive \hsinline{parBO}.
We conjecture that the monad is composable with the \hsinline{Reader r} monad for a duplicable type.
Further exploration is left for future work.

\begin{acks}
We would like to express our sincere gratitude to Yudai Tanabe, Taro Sekiyama and Atsushi Igarashi for their valuable feedback and suggestions during the early stages of this work.
We heartily thank Arnaud Spiwack for his insightful and encouraging feedback.
We also thank the anonymous reviewers for their valuable feedback.
This research was supported in part by the Hakubi Project at Kyoto University and JSPS KAKENHI Grant Numbers JP24KJ0133 and JP20H00582 for the first author, who is also hosted by MPI-SWS as a JSPS Overseas Research Fellow.
\end{acks}

\section*{Data Availability Statement}

Our Haskell implementation of Pure Borrow and our benchmark suite are archived on Zenodo at \cite{MatsushitaIshii26-Pure-Borrow-artifact}.
The latest version is maintained as a public repository at \url{https://github.com/SoftwareFoundationGroupAtKyotoU/pure-borrow}.

\bibliographystyle{ACM-Reference-Format}
\bibliography{./bib.bib}

\iffull
\appendix
\appendixsetup
\section{Formal Calculus}
\label{appx:formal}

We present the omitted details of our formal calculus introduced in \cref{sect:metatheory}.

\subsection{Type System}
\label{appx:formal-type}

We present the omitted details of our type system introduced in \cref{sect:metatheory-type}.

\paragraph{Basics}

The typing context \(\Tctx\) consists of items of the form \(\var \ccoll{\mult} \Ty\), with the multiplicity \(\mult\).\footnote{
  The order of items in a typing context \(\Tctx\) is ignored.
  For each variable \(\var\), a typing context \(\Tctx\) can contain at most one item of the form \(\var \ccoll{\mult} \Ty\).
  We write \(\empseq\) for the empty typing context.
  The domain \(\dom \Tctx\) is defined as \(\set{\ka \var}{\var \ccoll{\mult} \Ty \kb\in\kb \Tctx \ka}\).
}

We define the multiplicity product \(\mult \cdot \multB\) naturally, i.e.,
\(\one \cdot \mult \ka\defeq\ka \mult \cdot \one \ka\defeq\ka \mult\),
\(\many \cdot \mult \ka\defeq\ka \mult \cdot \many \ka\defeq\ka \many\),
\(\prod \Seq{\multvar} \cdot \prod \Seq{\multvarB} \ka\defeq\ka \prod (\Seq{\multvar}, \Seq{\multvarB})\).
This is introduced to define the typing context multiplication \(\mult \Tctx\) (\cref{fig:tctx-ops}).
The formal product of variables \(\prod \Seq{\multvar}\) is introduced for the multiplicity product.\footnote{
  In \(\prod \Seq{\multvar}\), the sequence \(\Seq{\multvar}\) should be non-empty, and its order is ignored.
  The substitution \(\paren{\prod \Seq{\multvar}} \sqbr{\Seq{\multB/\multvarB}}\) is defined by the multiplicity product.
}

\paragraph{Inclusion and subtyping}

The lifetime inclusion \(\lft \le \lftB\) is the partial order over lifetimes generated by the rules \(\lft \le \static\), \(\lft_0 \lftand \lft_1 \le \lft_\idx\) and ``if \(\lftB \le \lft_0\) and \(\lftB \le \lft_1\), then \(\lftB \le \lft_0 \lftand \lft_1\)''.
The multiplicity inclusion \(\mult \le \multB\) is the partial order over multiplicities generated by the rules \(\one \le \mult\), \(\mult \le \many\), \(\mult_\idx \le \mult_0 \cdot \mult_1\), and ``if \(\mult_0 \le \multB\) and \(\mult_1 \le \multB\), then \(\mult_0 \cdot \mult_1 \le \multB\)''.

\begin{figure}
  \small\tightmathenv
  \begin{gather*}
    \Ty \subty \Ty
  \hspace{3em}
    \frac{
      \Ty \subty \Ty'
    \hspace{1.5em}
      \Ty' \subty \Ty''
    }{
      \Ty \subty \Ty''
    }
  \hspace{3em}
    \paren{\all{\anyvar} \Ty} \ka\subty\ka \Ty \ka\sqbr{\any \ky/ \anyvar}
  \hspace{3em}
    \frac{
      \text{\(\anyvar\) is fresh in \(\TyB\)}
    \hspace{1.5em}
      \TyB \subty \Ty
    }{
      \TyB \ka\subty\ka \all{\anyvar} \Ty
    }
  \\[.2em]
    \frac{
      \Ty' \subty \Ty
    \hspace{1em}
      \TyB \subty \TyB'
    \hspace{1em}
      \mult \le \multB
    }{
      \Ty \to_\mult \TyB \ke\subty\ke \Ty' \to_\multB \TyB'
    }
  \hspace{3em}
    \frac{
      \all{\idx} \Ty_\idx \sqbr{\Seq{\TyB \ky/ \Tvar}} \kb\subty\kb
        \Ty_\idx \sqbr{\Seq{\TyB' \ky/ \Tvar}}
    \hspace{1.5em}
      \datatemplate
    }{
      \Tcon\kd \seq{\TyB} \ka\subty\ka \Tcon\kd \seq{\TyB'}
    }
    \tagsp\ruletag{subty-data}
  \\[.2em]
    \frac{
      \Ty \subty \TyB
    }{
      \kRef \Ty \subty \kRef \TyB
    }
  \hspace{3em}
    \frac{
      \lftB \le \lft
    }{
      \End^\lft \ka\subty\kb \End^\lftB
    }
  \hspace{3em}
    \frac{
      \lftB \le \lft
    \hspace{1.5em}
      \Ty \subty \TyB
    \hspace{1.5em}
      \TyB \subty \Ty
    }{
      \Mut^\lft\kd \Ty \kb\subty\kb \Mut^\lftB\kd \TyB
    }
  \\[.2em]
    \frac{
      \lftB \le \lft
    \hspace{1.5em}
      \Ty \subty \TyB
    }{
      \Share^\lft\kd \Ty \kb\subty\kb \Share^\lftB\kd \TyB
    }
  \hspace{3em}
    \frac{
      \lft \le \lftB
    \hspace{1.5em}
      \Ty \subty \TyB
    }{
      \Lend^\lft\kd \Ty \kb\subty\kb \Lend^\lftB\kd \TyB
    }
  \hspace{3em}
    \frac{
      \lftB \le \lft
    \hspace{1.5em}
      \Ty \subty \TyB
    }{
      \BO^\lft\kd \Ty \ka\subty\kb \BO^\lftB\kd \TyB
    }
  \end{gather*}
  \caption{Subtyping rules.}
  \label{fig:subty-rules}
\end{figure}

\begin{figure}
  \tightmathenv
  \[
    \Tctx \subty \Tctx
  \hspace{3em}
    \frac{
      \Tctx \subty \Tctx'
    \hspace{1em}
      \Tctx' \subty \Tctx''
    }{
      \Tctx \subty \Tctx''
    }
  \hspace{3em}
    \frac{
      \multB \le \mult
    \hspace{1.5em}
      \Ty \subty \TyB
    }{
      \var \ccoll{\mult} \Ty,\kb \Tctx \kb\subty\kb \var \ccoll{\multB} \TyB,\kb \Tctx
    }
  \hspace{3em}
    \var \ccoll{\many} \Ty,\kb \Tctx \kb\subty\kb \Tctx
  \]
  \tightcaption
  \caption{Rules for typing context inclusion.}
  \label{fig:tctx-incl}
\end{figure}

Our type system introduces subtyping natively, without a function like \hsinline{upcast} (\cref{fig:subty}, \cref{sect:overview-details}) in our Haskell implementation.
The rules for the subtyping \(\Ty \subty \TyB\) are listed in \cref{fig:subty-rules}.
To be precise, we use mixed induction (i.e., coinduction over induction) to reason about recursive data types, following the approach of \citet{DanielssonA10-subtype-declarative}.
More specifically, the subtyping rule \ref{rule:subty-data} for the user-defined data type \(\Tcon\) can be applied coinductively, whereas the other rules can only be applied inductively.
Overall, the subtyping relation is defined as the greatest fixpoint over the least fixpoint.

We also introduce the typing context inclusion \(\Tctx \subty \TctxB\), which is inductively defined by the rules listed in \cref{fig:tctx-incl}.
We can weaken the type of a variable by subtyping.
A variable of the unrestricted multiplicity \(\many\) can be forgotten and also turned into the multiplicity \(\one\).

\paragraph{Operations on typing contexts}

\begin{figure}
  \tightmathenv
  \begin{gather*}
    \mult\ka \empseq \ke\defeq\ke \empseq
  \hspace{3em}
    \mult\ka \paren{\var \ccoll{\multB} \Ty,\kb \Tctx} \ke\defeq\ke
      \var \ccoll{\mult \cdot \multB} \Ty,\kb \mult\ka \Tctx
  \hspace{3em}
    \empseq + \TctxB \ke\defeq\ke \TctxB
  \\[.2em]
    \paren{\var \ccoll{\many} \Ty,\kb \Tctx} \ka+\ka \paren{\var \ccoll{\many} \Ty,\kb \TctxB} \ke\defeq\ke
      \var \ccoll{\many} \Ty,\kb \paren{\Tctx + \TctxB}
  \hspace{3em}
    \frac{
      \kc\var \ka\notin\ka \dom \TctxB
    }{
      \paren{\var \ccoll{\mult} \Ty,\kb \Tctx} \ka+\ka \TctxB \ke\defeq\ke
        \var \ccoll{\mult} \Ty,\kb \paren{\Tctx + \TctxB}
    }
  \end{gather*}
  \caption{Operations on typing contexts.}
  \label{fig:tctx-ops}
\end{figure}

Following \citet{BernardyBNJS18-Linear-Haskell}, we introduce the multiplication \(\mult \Tctx\) and the sum \(\Tctx + \TctxB\) over typing contexts, defined in \cref{fig:tctx-ops}.
Intuitively, \(\Tctx + \TctxB\) allows sharing unrestricted variables between \(\Tctx\) and \(\TctxB\).

\paragraph{Typing rules}

\begin{figure}
  \tightmathenv
  \begin{gather*}
    \data\kf \paren{} \kf\where\kf \paren{} \ccol \paren{}
  \hspace{4em}
    \data\kf \Bool \kf\where\kh
      \True \ccol \Bool,\kd \False \ccol \Bool
  \\[.2em]
    \data\kf \Ur\ke \Tvar \kf\where\kh
      \Ur \kccol \Tvar \to \Ur\ke \Tvar
  \hspace{4em}
    \data\kf (\Tvar, \TvarB) \kf\where\kh
      \paren{{,}} \kccol \Tvar \lto \TvarB \lto (\Tvar, \TvarB)
  \end{gather*}
  \caption{Default data type declarations.}
  \label{fig:data-decls}
\end{figure}

\begin{figure}
  \small\tightmathenv
  \begin{gather*}
    \frac{
      \TctxB \subty \Tctx
    \hspace{1em}
      \Tctx \vdash \term \ccol \Ty
    }{
      \TctxB \vdash \term \ccol \Ty
    }
  \hspace{3em}
    \frac{
      \Tctx \vdash \term \ccol \Ty
    \hspace{1em}
      \Ty \subty \TyB
    }{
      \Tctx \vdash \term \ccol \TyB
    }
  \hspace{3em}
    \frac{
      \text{\(\anyvar\) is fresh in \(\Tctx\)}
    \hspace{1.5em}
      \Tctx \vdash \term \ccol \Ty
    }{
      \Tctx \vdash \term \ccol \all{\anyvar} \Ty
    }
  \hspace{3em}
    \var \ccoll{\mult} \Ty \kvdash \var \ccol \Ty
  \\[.2em]
    \frac{
      \all{\idx}\ke \Tctx_\idx \kvdash \term \ccol \Ty_\idx
    \hspace{2em}
      \Seq{\var \ccoll{\mult} \Ty},\kb \TctxB \kvdash \termB \ccol \TyB
    }{
      \sum_\idx \mult\ka \Tctx_\idx + \TctxB \kvdash
        \klet \Seq{\var \teq \term} \ktin \termB \kccol \TyB
    }
  \hspace{3em}
    \frac{
      \all{\idx}\ke \Seq{\var \ccoll{\many} \Ty},\kb \Tctx \kvdash \term_\idx \ccol \Ty_\idx
    \hspace{2em}
      \Seq{\var \ccoll{\many} \Ty},\kb \TctxB \kvdash \termB \ccol \TyB
    }{
      \many\ka \Tctx + \TctxB \kvdash
        \klet \Seq{\var \teq \term} \ktin \termB \kccol \TyB
    }
  \\[.2em]
    \frac{
      \var \ccoll{\mult} \Ty,\kb \Tctx \kvdash \term \ccol \TyB
    }{
      \Tctx \kvdash \tlamx{\var} \term \kccol \Ty \to_\mult \TyB
    }
  \hspace{3em}
    \frac{
      \Tctx \kvdash \term \kccol \Ty \to_\mult \TyB
    \hspace{2em}
      \TctxB \kvdash \termB \kccol \Ty
    }{
      \Tctx + \mult\ka \TctxB \kvdash \term\kd \termB \kccol \TyB
    }
  \hspace{3em}
    \frac{
      \var \ka\in\ka \dom \Tctx
    \hspace{2em}
      \Tctx \vdash \term \ccol \Ty
    }{
      \Tctx \kvdash \kseq \var\kd \term \kccol \Ty
    }
  \\[.2em]
    \vdashk \num \ccol \Int
  \hspace{3em}
    \frac{
      \all{\idx}\ke \Tctx_\idx \kvdash \term_\idx \kccol \Ty_{\idxB, \idx}\sqbr{\Seq{\TyB \ky/ \Tvar}}
    \hspace{2em}
      \datatemplate
    }{
      \sum_\idx \mult_{\idxB, \idx}\kb \Tctx_\idx \kvdash
      \Con_\idxB\kc \seq{\term} \kccol \Tcon\kd \seq{\TyB}
    }
  \\[.2em]
    \frac{\begin{gathered}
      \Tctx \vdash \term \ccol \Tcon\kd \seq{\TyB}
    \\[-.45em]
      \all{\idxB}\ke
        \Seq{\var \ccoll{\mult_{\idxB}}\kz \Ty_\idxB
          \sqbr{\vphantom{\hat A}\smash{\Seq{\TyB \ky/ \Tvar}}}},\kb
        \TctxB \kvdash \termB_\idxB \kccol \TyB'
    \\[-.35em]
      \datatemplate
    \end{gathered}}{
      \Tctx + \TctxB \kvdash \kcase \term \kof
        \kcurly{ \Seq{\Con\kd \seq{\var} \tto \termB} } \ccol \TyB'
    }
  \hspace{3em}
    \frac{\begin{gathered}
      \Tctx \vdash \term \ccol \Bcon^\lft\kb \paren{\Tcon\kd \seq{\TyB}}
    \\[-.45em]
      \all{\idxB}\ke
        \Seq{\var \ccoll{\mult_{\idxB}}\kz \Bcon^\lft\ka \paren{\Ty_\idxB
          \sqbr{\vphantom{\hat A}\smash{\Seq{\TyB \ky/ \Tvar}}}}},\kb
        \TctxB \kvdash \termB_\idxB \kccol \TyB'
    \\[-.35em]
      \datatemplate
    \end{gathered}}{
      \Tctx + \TctxB \kvdash \kcase \term \kof
        \kcurly{ \Seq{\Con\kd \seq{\var} \tto \termB} } \ccol \TyB'
    }
    \tagsp\ruletag{ty-case-bor}
  \\[.4em]
    \frac{
      \any \kccol \seq{\Ty};\ka \TyB
    \hspace{1.5em}
      \all{\idx}\kc \Tctx_\idx \vdash \term_\idx \ccol \Ty_\idx
    }{
      \sum_\idx \Tctx_\idx \kvdash \any\kd \seq{\term} \kccol \TyB
    }
  \hspace{3em}
    \iop \kccol \Seq{\Int};\kd \Int
  \hspace{3em}
    \irel \kccol \Seq{\Int};\kd \Bool
  \hspace{3em}
    \tpar \kccol \Ty,\kb \TyB;\kd \paren{\Ty, \TyB}
  \\[.2em]
    \frac{
      \text{\(\Ty\) is \(\Linearly\) or \(\Mut^\lft\kc \TyB\)}
    }{
      \consume \kccol \Ty;\kd \paren{}
    }
  \hspace{3em}
    \frac{
      \text{\(\Ty\) is \(\Int\), \(\End^\lft\) or \(\Share^\lft\kc \TyB\)}
    }{
      \move \kccol \Ty;\kd \Ur\kd \Ty
    }
  \hspace{3em}
    \linearly \kccol \Linearly \lto \Ur\kd \Ty;\kd \Ur\kd \Ty
  \\[.2em]
    \frac{
      \text{\(\Ty\) is \(\Linearly\), \(\kRef \TyB\), \(\Now^\lft\) or \(\Mut^\lft\kc \TyB\)}
    }{
      \withLinearly \kccol \Ty;\kd \paren{\Linearly,\ka \Ty}
    }
  \hspace{3em}
    \newRef \kccol \Ty;\kd \kRef \Ty
  \hspace{3em}
    \freeRef \kccol \kRef \Ty;\kd \Ty
  \\[.2em]
    \newLifetime \kccol \paren{\all{\lftid}\ka \Now^{\ka\al \lftid} \nklto \Ty};\kd \Ty
  \hspace{3em}
    \tendLifetime \kccol \Now^{\ka\al \lftid};\kd
      \End^{\ka\al \lftid}
  \\[.1em]
    \borrow \kccol \Linearly,\kb \Ty;\kd
      \paren{\Mut^\lft\kc \Ty,\kb \Lend^\lft\kc \Ty}
  \hspace{3em}
    \share \kccol \Mut^\lft\kc \Ty;\kd \Ur\kd \paren{\Share^\lft\kc \Ty}
  \\[.0em]
    \frac{
      \text{\(\Ty\) is \(\Int\), \(\End^\lftB\) or \(\Share^\lftB\kc \TyB\)}
    }{
      \tcopy \kccol \Share^\lft\kc \Ty;\kd \Ty
    }
  \hspace{2.5em}
    \joinMut \kccol
      \Bcon^\lft\kb \paren{\Mut^\lftB\kc \Ty};\kd \Bcon^{\lft \lftand \lftB}\kc \Ty
  \hspace{2.5em}
    \reclaim \kccol \Lend^\lft\kc,\kb  \Ty\End^\lft;\kd \Ty
  \\[.3em]
    \execBO \kccol \Now^\lft,\kb \BO^\lft\kc \Ty;\kd
      \paren{\Now^\lft, \Ty}
  \hspace{2.5em}
    \pure \kccol \Ty;\kd \BO^\lft\kd \Ty
  \hspace{2.5em}
    \paren{{\bind}} \kccol
      \BO^\lft\kc \Ty,\kb \Ty \nklto \BO^\lft\kc \TyB;\kd
      \BO^\lft\kc \TyB
  \\[.2em]
    \sexecBO \kccol
      \Now^\lft,\kb \BO^{\lft \lftand \lftB}\kc \TyB;\kd
      \BO^\lftB\kc \paren{\Now^\lft,\ka \TyB}
  \hspace{3em}
    \parBO \kccol \BO^\lft\kc \Ty,\kb \BO^\lft\kc \TyB;\kd
      \BO^\lft\kb \paren{\Ty, \TyB}
  \\[.1em]
    \deref \kccol \Bcon^\lft\kb \paren{\kRef \Ty};\kd
      \BO^\lft\kb \paren{\Bcon^\lft\kc \Ty}
  \hspace{1.5em}
    \frac{
      \lftB \le \lft
    }{
      \updateRef \kccol
        \Ty \nklto \BO^\lftB\ka \paren{\TyB, \Ty},\kc
        \Mut^\lft\ka \paren{\kRef \Ty};\ke
        \BO^\lftB\ka \paren{\TyB,\kb \Mut^\lft\ka \paren{\kRef \Ty}}
    }
  \end{gather*}
  \caption{Typing rules.}
  \label{fig:typing}
\end{figure}

The typing rules are listed in \cref{fig:typing}.
For this, we assume the default data type declarations listed in \cref{fig:data-decls}.
The typing judgment has the form \(\Tctx \vdash \term \ccol \Ty\).
For operators \(\op\) and monad constructors \(\mo\), we introduce an auxiliary judgment \(\any \ccol \seq{\Ty}; \TyB\) (where \(\any\) is either \(\op\) or \(\mo\) and the length of \(\seq{\Ty}\) is the arity of \(\any\)), meaning that \(\any\) linearly inputs arguments of the type \(\seq{\Ty}\) and outputs an object of the type \(\TyB\).

We have two typing rules for typing the \(\case\) term.
The latter \ref{rule:ty-case-bor} is for the case where the pattern-matched term is a borrower \(\Bcon^\lft\kb \paren{\Tcon\kd \seq{\TyB}}\).
We distribute the borrower constructor \(\Bcon^\lft\) to the components \(\seq{\var}\) of the data.
This is an idealized version of the API functions such as \hsinline{splitPair} (\cref{fig:apis-split}, \cref{sect:overview-details}) in our Haskell implementation.

\subsection{Two Operational Semantics}
\label{appx:formal-opsem}

We present the omitted details of our two operational semantics introduced in \cref{sect:metatheory-opsem}.

\paragraph{Basics}

In denotational operational semantics, we use a borrower constructor \(\DBcon\), where a mutable borrower \(\Mut^{\seq{\bpath}}\) is marked with borrow paths \(\seq{\bpath}\), representing which part of which lenders the borrower refers to.
Here, we can have multiple borrow paths due to the reborrowing caused by \(\joinMut\), as clarified by the reduction rule \ref{rule:dred-join-bor} shown later in \cref{fig:red-borrow-dos}.
We use a borrower context \(\DBctx\) for a composite of borrower constructors.
We write \(\DBctx\kd \any\) instead of \(\DBctx \sqbr{\any}\) for readability.

In both semantics, an environment \(\MEnv\), \(\DEnv\) assigns a term \(\mterm\), \(\dterm\) to each variable \(\var\).
A denotational environment \(\DEnv\) additionally manages the borrow ids \(\seq{\bid}\) that have been created.\footnote{
  The order of items in an environment \(\MEnv\), \(\DEnv\) is ignored.
  An environment \(\MEnv\), \(\DEnv\) can have at most one item of the form \(\var \teq \any\) for each variable \(\var\).
  Also, a denotational environment \(\DEnv\) can have at most one occurrence of each borrow id \(\bid\).
}
The memory \(\Mem\) appears only in mutative semantics.\footnote{
  The order of items in a memory \(\Mem\) is ignored.
  A memory \(\Mem\) can have at most one item of the form \(\loc \mMapsto \var\) for each location \(\loc\).
}
A configuration has the form \(\MEnv; \Mem; \var\) in mutative semantics and \(\DEnv; \var\) in denotational semantics, where \(\var\) is the top-level variable.\footnote{
  We assume alpha-equivalence over mutative configurations \(\MEnv; \Mem; \var\) with respect to variable names \(\var\) bound in the environment \(\MEnv\), and assume alpha-equivalence over denotational configurations \(\DEnv; \var\) with respect to variable names \(\var\) and borrow id names \(\bid\) bound in the environment \(\DEnv\).
}
We say that \(\MEnv; \Mem; \var\) is a normal form if \(\MEnv\) has an item of the form \(\var \teq \mval\).
Likewise, we say that \(\DEnv; \var\) is a normal form if \(\DEnv\) has an item of the form \(\var \teq \dval\).

Strictness works as follows in this language.
For simplicity, all the arguments of operators are made strict.
In particular, \(\tpar\kd \term\kd \termB\) forces the evaluation of \(\term\) and \(\termB\) before returning the pair of their results.
Also, we introduce a special term \(\tseq\kd \var\kd \term\) that forces the evaluation of the variable \(\var\) before returning \(\term\);
it amounts to \hsinline{_do !x <- x; t} in Haskell, and we introduce \(\tseq\) as a primitive for the bang \hsinline{!} pattern.
The usual function application is lazy.
Pattern matching of \(\case\) is strict.

\begin{figure}
  \subcaptionsetup{aboveskip=1pt}
  \begin{subcaptionblock}{\textwidth}
    \tightmathenv
    \begin{align*}
      \stag{Denesting context} \Ntx
        & \sdef \khole\kd \term
        \sor \var\kd \khole
        \sor \kseq \var\kd \khole
        \sor \Con\kd \seq{\var}\kd \khole\kd \seq{\term}
    \\[-.2em] & \sskip
        \sor \kcase \khole \kof
          \kcurly{ \Seq{\Con\kd \seq{\var} \tto \term} }
        \sor \op\kd \seq{\var}\kd \khole\kd \seq{\term}
        \sor \mo\kd \seq{\var}\kd \khole\kd \seq{\term}
    \end{align*}\vspace{-.5em}
    \caption{Common things.}
    \label{fig:opsem-syntax-common-more}
  \end{subcaptionblock}
  \begin{subcaptionblock}{.46\textwidth}
    \tightmathenv
    \begin{align*} &
      \stag{Extra operator} \mop
        \sdef \linear \arity{1}
        \sor \exeBO \arity{1}
    \\[-.1em] & \hspace{.8em}
        \sor \execBO^\post \arity{1}
        \sor \paren{{\bind^\post}} \arity{2}
        \sor \sexecBO^\pre \arity{2}
    \\[-.1em] & \hspace{.8em}
        \sor \sexecBO^\post \arity{1}
        \sor \parBO^\post \arity{2}
        \sor \deref^\post \arity{1}
    \\[-.1em] & \hspace{.8em}
        \sor \updateRef^\pre \arity{2}
        \sor \updateRef^\prepost_\loc \arity{1}
    \\[-.1em] & \hspace{.8em}
        \sor \updateRef^\post_\loc \arity{1}
    \end{align*}\vspace{-.2em}
    \caption{For mutative semantics.}
    \label{fig:opsem-syntax-mos-more}
  \end{subcaptionblock}
  \hspace{.3em}
  \begin{subcaptionblock}{.5\textwidth}
    \tightmathenv
    \begin{gather*}
      \stag{Borrower context} \DBctx
        \sdef \khole \sor \DBcon\kd \DBctx
    \\[.1em]
    \begin{aligned} &
      \stag{Extra operator} \dop
        \sdef \linear \arity{1}
        \sor \exeBO_\Hist \arity{1}
    \\[-.3em] & \hspace{.8em}
        \sor \execBO^\post \arity{1}
        \sor \paren{{\bind^\post}} \arity{2}
        \sor \sexecBO^\pre_\Hist \arity{2}
    \\[-.3em] & \hspace{.8em}
        \sor \sexecBO^\post_{\Hist; \Histpr} \arity{1}
        \sor \parBO^\post_\Hist \arity{2}
        \sor \deref^\post_\Hist \arity{1}
    \\[-.3em] & \hspace{.8em}
        \sor \updateRef^\pre_\Hist \arity{2}
        \sor \updateRef^\prepost_{\seq{\bpath}} \arity{1}
    \\[-.3em] & \hspace{.8em}
        \sor \updateRef^\post_{\seq{\bpath};\ka \Hist} \arity{1}
    \end{aligned}
    \end{gather*}
    \caption{For denotational semantics.}
    \label{fig:opsem-syntax-dos-more}
  \end{subcaptionblock}
  \caption{More syntax for the two operational semantics.}
  \label{fig:opsem-syntax-more}
\end{figure}

\Cref{fig:opsem-syntax-more} shows additional syntax for the two operational semantics, omitted in \cref{fig:opsem-syntax}.

The denesting context \(\Ntx\) is used for turning a nested term into an unnested one, by binding subterms to fresh variables (see \ref{rule:mred-denest} and \ref{rule:dred-denest} later shown in \cref{fig:red-basic}).\footnote{
  We impose left-to-right order in denesting to avoid unnecessary nondeterminism.
}
Denesting is a dynamic version of conversion to the A-normal (or K-normal) form.

We introduce extra values \(\mxval\), \(\dxval\) and extra operators \(\mop\), \(\dop\) for the two semantics.
For monadic extra operators such as \(\execBO^\post\), the superscript \(\post\) indicates post-processing, \(\pre\) pre-processing, and \(\prepost\) pre-processing before post-processing.
The value \(\nul\)/\(\nul_\Hist\) represents the value of linearity witness tokens \(\Linearly\) and lifetime tokens \(\Now^\lft\), \(\End^\lft\);
in denotational semantics, the lifetime tokens carry the history \(\Hist\).

\paragraph{Composing histories}

\begin{figure}
  \tightmathenv
  \[
    \frac{
      \kc\dom \Hist \ka\cap\ka \dom \Histpr \kb\teq\kb \empset
    }{
      \Hist \ka\hpar\ka \Histpr \ke\defeq\ke \Hist, \Histpr
    }
  \hspace{4.5em}
    \begin{gathered}
      \paren{\bpath \hstore \var,\kb \Hist} \kc\hseq\kc \paren{\bpath \hstore \varB,\kb \Histpr} \ke\defeq\ke
        \bpath \hstore \varB,\kb \paren{\Hist \hseq \Histpr}
    \\[.3em]
      \frac{
        \kb\bpath \notin \dom \Histpr
      }{
        \paren{\bpath \hstore \var,\kb \Hist} \kc\hseq\kc \Histpr \ke\defeq\ke
          \bpath \hstore \var,\kb \paren{\Hist \hseq \Histpr}
      }
    \hspace{3em}
      \emphist \ka\hseq\ka \Hist \ke\defeq\ke \Hist
    \end{gathered}
  \]
  \caption{Parallel \(\Hist \hpar \Histpr\) and sequential \(\Hist \hseq \Histpr\) composition of histories.}
  \label{fig:hist-comp}
\end{figure}

We define the parallel \(\Hist \hpar \Histpr\) and sequential \(\Hist \hseq \Histpr\) compositions of histories, as defined in \cref{fig:hist-comp}.
The parallel composition \(\Hist \hpar \Histpr\) treats both histories equally (and thus is commutative), while the sequential composition \(\Hist \hseq \Histpr\) overwrites older records \(\Hist\) with newer ones \(\Histpr\).
The domain \(\dom \Hist\) is defined as \(\set{\ka \bpath}{\bpath \hstore \var \kb\in\kb \Hist \ka}\).

\paragraph{Reduction rules}

\begin{figure}
  \tightmathenv
  \begin{gather*}
    \frac{
      \var \teq \MKtx\sqbr{\varB} \kin \MEnv
    \hspace{2em}
      \MEnv; \Mem; \varB \hasloop
    }{
      \MEnv; \Mem; \var \hasloop
    }
  \hspace{4.5em}
    \frac{
      \var \teq \DKtx\sqbr{\varB} \kin \DEnv
    \hspace{2em}
      \DEnv; \varB \hasloop
    }{
      \DEnv; \var \hasloop
    }
  \end{gather*}
  \caption{The coinductive rule of the forcing loop predicate \(\MCfg \hasloop\), \(\DCfg \hasloop\) for each semantics.}
  \label{fig:forceloop}
\end{figure}

\begin{figure}
  \begin{subcaptionblock}{\textwidth}
    \small\tightmathenv
    \begin{gather*}
      \frac{
        \var \teq \MKtx\sqbr{\varB} \kin \MEnv
      \hspace{2em}
        \MEnv; \Mem; \varB \kd\to\kd \MEnv'; \Mem'; \varB
      }{
        \MEnv; \Mem; \var \kd\to\kd \MEnv'; \Mem'; \var
      }
      \tagsp\ruletag{mred-ktx}
    \hspace{3em}
      \frac{
        \MCfg \hasloop
      }{
        \MCfg \to \MCfg
      }
      \tagsp\ruletag{mred-loop}
    \\[.0em]
    \begin{aligned} &
      \var \kteq \klet \Seq{\varB \teq \term} \ktin \termB,\kb \xMCfg \kto
    \\[-.4em] & \hspace{6em}
        \Seq{\varB \teq \term},\kb \var \teq \termB,\kb \xMCfg
    \end{aligned}
    \hspace{3em}
      \frac{
        \term \notin \Var
      }{
        \var \teq \Ntx\sqbr{\term},\ka \xMCfg \kto
          \varB \teq \term,\ka \var \teq \Ntx\sqbr{\varB},\ka \xMCfg
      }
      \tagsp\ruletag{mred-denest}
    \\[.0em]
      \frac{
        \varB \teq \mval \kin \MEnv
      }{
        \var \teq \varB,\ka \xMCfg \kto
          \var \teq \mval,\ka \xMCfg
      }
      \tagsp\ruletag{mred-var}
    \hspace{3em}
      \frac{
        \varF \teq \tlamx{\varB} \term \kin \MEnv
      }{
        \var \teq \varF\kc \varB,\ka \xMCfg \kto
          \var \teq \term,\ka \xMCfg
      }
    \\[.0em]
      \frac{
        \varB \teq \mval \kin \MEnv
      }{
        \var \teq \kseq \varB\kd \varC,\ka \xMCfg \kto
          \var \teq \varC,\ka \xMCfg
      }
    \hspace{2.25em}
      \frac{
        \varB \teq \Con_\idx\kb \Seq{\varC'} \kin \MEnv
      }{
        \var \teq \kcase \varB \kof
          \kcurly{ \Seq{\Con\kd \seq{\varC} \tto \term} },\kb \xMCfg \kto
          \Seq{\varC_\idx \teq \varC'},\ka
          \var \teq \term_\idx,\ka \xMCfg
      }
    \end{gather*}
    \caption{For mutative operational semantics.}
    \label{fig:red-basic-mos}
  \end{subcaptionblock}
  \begin{subcaptionblock}{\textwidth}
    \small\tightmathenv
    \begin{gather*}
      \frac{
        \var \teq \DKtx\sqbr{\varB} \kin \DEnv
      \hspace{2em}
        \DEnv; \varB \kd\to\kd \DEnv'; \varB
      }{
        \DEnv; \var \kd\to\kd \DEnv'; \var
      }
      \tagsp\ruletag{dred-ktx}
    \hspace{3em}
      \frac{
        \DCfg \hasloop
      }{
        \DCfg \to \DCfg
      }
      \tagsp\ruletag{dred-loop}
    \\[.0em]
    \begin{aligned} &
      \var \kteq \klet \Seq{\varB \teq \term} \ktin \termB,\kb
      \DEnv;\kc \kc \var \kto
    \\[-.5em] & \hspace{7em}
        \Seq{\varB \teq \term},\kb \var \teq \termB,\kb \xDCfg
    \end{aligned}
    \hspace{3em}
      \frac{
        \term \notin \Var
      }{
        \var \teq \Ntx\sqbr{\term},\ka \xDCfg \kto
          \varB \teq \term,\ka \var \teq \Ntx\sqbr{\varB},\ka \xDCfg
      }
      \tagsp\ruletag{dred-denest}
    \\[.0em]
      \frac{
        \varB \teq \dval \kin \DEnv
      }{
        \var \teq \DBctx\kd \varB,\ka \xDCfg \kto
          \var \teq \DBctx\kd \dval,\ka \xDCfg
      }
      \tagsp\ruletag{dred-var}
    \hspace{3em}
      \frac{
        \varF \teq \tlamx{\varB} \term \kin \DEnv
      }{
        \var \teq \varF\kc \varB,\ka \xDCfg \kto
          \var \teq \term,\ka \xDCfg
      }
    \\[.0em]
      \frac{
        \varB \teq \dval \kin \DEnv
      }{
        \var \teq \kseq \varB\kd \varC,\ka \DEnv; \var \kto
          \var \teq \varC,\ka \xDCfg
      }
    \hspace{3em}
      \frac{
        \varB \teq \Con_\idx\kb \Seq{\varC'} \kin \DEnv
      }{
        \var \teq \kcase \varB \kof
          \kcurly{ \Seq{\Con\kd \seq{\varC} \tto \term} },\kb
          \DEnv;\kc \var \kto
          \Seq{\varC_\idx \teq \varC'},\ka
          \var \teq \term_\idx,\ka \xDCfg
      }
    \\[.0em]
      \frac{
        \varB \teq \DBcon\kc \paren{\kz\Con_\idx\kb \Seq{\varC'}} \kin \DEnv
      \hspace{2em}
        \all{\idxB}\ke \DBcon'_\idxB \kc=\kc \begin{cases} \kd
          \Mut^{\Seq{\bpath.\idxB}} \ku & \DBcon \teq \Mut^{\seq{\bpath}}
        \\[-.4em] \kd
          \Share \ku & \DBcon \teq \Share
        \end{cases}
      }{
        \var \teq \kcase \varB \kof
        \kcurly{ \Seq{\Con\kd \seq{\varC} \tto \term} },\kb
        \DEnv;\kc \var \kto
          \Seq{\varC_\idx \teq \DBcon'\kb \varC'},\ka
          \var \teq \term_\idx,\ka \xDCfg
      }
      \tagsp\ruletag{dred-case-bor}
    \end{gather*}
    \caption{For denotational operational semantics.}
    \label{fig:red-basic-dos}
  \end{subcaptionblock}
  \tightcaption
  \caption{Basic reduction rules.}
  \label{fig:red-basic}
\end{figure}

\begin{figure}
  \begin{subcaptionblock}{\textwidth}
    \small\tightmathenv
    \begin{gather*}
      \frac{
        \Seq{\varB \teq \num} \kin \MEnv
      \hspace{1.5em}
        \numB = \iop\kd \seq{\num}
      }{\begin{aligned} &
        \var \teq \iop\kd \seq{\varB},\ka \xMCfg \kto
      \\[-.5em] & \hspace{5.5em}
          \var \teq \numB,\ka \xMCfg
      \end{aligned}}
    \hspace{2.5em}
      \frac{
        \Seq{\varB \teq \num} \kin \MEnv
      \hspace{1.5em}
        \bool = \asBool\paren{\irel\kd \seq{\num}}
      }{\begin{aligned} &
        \var \teq \irel\kd \seq{\varB},\ka \xMCfg \kto
      \\[-.5em] & \hspace{6.5em}
          \var \teq \bool,\ka \xMCfg
      \end{aligned}}
    \hspace{2.5em}
      \frac{
        \varB \teq \mval,\kc \varC \teq \mvalB \kin \MEnv
      }{\begin{aligned} &
        \var \teq \tpar\kd \varB\kd \varC,\ka \xMCfg \kto
      \\[-.5em] & \hspace{5em}
          \var \teq \paren{\varB, \varC},\ka \xMCfg
      \end{aligned}}
    \\[.3em]
      \frac{
        \varB \teq \mval \kin \MEnv
      }{
        \var \teq \consume\kd \varB,\ka \xMCfg \kto
          \var \teq \paren{},\ka \xMCfg
      }
    \hspace{3em}
      \frac{
        \varB \teq \mval \kin \MEnv
      }{
        \var \teq \move\kd \varB,\ka \xMCfg \kto
          \var \teq \Ur\kd \varB,\ka \xMCfg
      }
    \\[.3em]
      \frac{
        \varF \teq \mval \kin \MEnv
      }{
        \var \teq \linearly\kd \varF,\ka \xMCfg \kto
          \varLi \teq \nul,\ka \varB \teq \varF\kc \varLi,\ka
          \var \teq \linear\kd \varB,\ka \xMCfg
      }
    \hspace{2em}
      \frac{
        \varB \teq \mval \kin \MEnv
      }{
        \var \teq \linear\kd \varB,\ka \xMCfg \kto
          \var \teq \mval,\ka \xMCfg
      }
    \\[.3em]
      \frac{
        \varB \teq \mval \kin \MEnv
      }{
        \var \teq \withLinearly\kd \varB,\ka \xMCfg \kto
          \varLi \teq \nul,\ka \var \teq \paren{\varLi, \varB},\ka \xMCfg
      }
    \end{gather*}
    \caption{For mutative operational semantics.}
    \label{fig:red-basicop-mos}
  \end{subcaptionblock}
  \begin{subcaptionblock}{\textwidth}
    \small\tightmathenv
    \begin{gather*}
      \frac{
        \Seq{\varB \teq \num} \kin \DEnv
      \hspace{1.5em}
        \numB = \iop\kd \seq{\num}
      }{
        \var \teq \iop\kd \seq{\varB},\ka \xDCfg \kto
          \var \teq \numB,\ka \xDCfg
      }
    \hspace{1.5em}
      \frac{
        \Seq{\varB \teq \num} \kin \DEnv
      \hspace{1em}
        \bool = \asBool\paren{\irel\kd \seq{\num}}
      }{
        \var \teq \irel\kd \seq{\varB},\ka \xDCfg \kto
          \var \teq \bool,\ka \xDCfg
      }
    \hspace{1.5em}
      \frac{
        \varB \teq \dval,\kc \varC \teq \dvalB \kin \DEnv
      }{
        \var \teq \tpar\kd \varB\kd \varC,\ka \xDCfg \kto
          \var \teq \paren{\varB, \varC},\ka \xDCfg
      }
    \\[.3em]
      \frac{
        \varB \teq \dval \kin \DEnv
      }{
        \var \teq \consume\kd \varB,\ka \xDCfg \kto
          \var \teq \paren{},\ka \xDCfg
      }
    \hspace{3em}
      \frac{
        \varB \teq \dval \kin \DEnv
      }{
        \var \teq \move\kd \varB,\ka \xDCfg \kto
          \var \teq \Ur\kd \varB,\ka \xDCfg
      }
    \\[.3em]
      \frac{
        \varF \teq \dval \kin \DEnv
      }{\begin{aligned} &
        \var \teq \linearly\kd \varF,\ka \xDCfg \kto
      \\[-.5em] & \hspace{2em}
          \varLi \teq \nul,\ka \varB \teq \varF\kc \varLi,\ka
          \var \teq \linear\kd \varB,\ka \xDCfg
      \end{aligned}}
    \hspace{2em}
      \frac{
        \varB \teq \dval \kin \DEnv
      }{\begin{aligned} &
        \var \teq \linear\kd \varB,\ka \xDCfg \kto
      \\[-.5em] & \hspace{5.5em}
          \var \teq \dval,\ka \xDCfg
      \end{aligned}}
    \hspace{2em}
      \frac{
        \varB \teq \dval \kin \DEnv
      }{\begin{aligned} &
        \var \teq \withLinearly\kd \varB,\ka \xDCfg \kto
      \\[-.5em] & \hspace{4em}
          \varLi \teq \nul,\ka \var \teq \paren{\varLi, \varB},\ka \xDCfg
      \end{aligned}}
    \end{gather*}
    \caption{For denotational operational semantics.}
    \label{fig:red-basicop-dos}
  \end{subcaptionblock}
  \tightcaption
  \caption{Reduction rules for basic operators.}
  \label{fig:red-basicop}
\end{figure}

\begin{figure}
  \begin{subcaptionblock}{.52\textwidth}
    \small\tightmathenv
    \begin{gather*}
      \frac{
        \varLi \teq \nul,\kd \varB \teq \mval \kin \MEnv
      }{
        \var \teq \newRef\kd \varLi\kd \varB,\ka \xMCfg \kto
          \var \teq \kRef \loc,\ka \MEnv;\kb
          \loc \mMapsto \varB,\ka \Mem;\kb \var
      }
    \\[-.1em]
      \frac{
        \varB \teq \kRef \loc \kin \MEnv
      }{
        \var \teq \freeRef\kd \varB,\ka \MEnv;\kb
        \loc \mMapsto \varC,\ka \Mem;\kb \var \kto
          \var \teq \varC,\ka \xMCfg
      }
    \end{gather*}
    \caption{For mutative semantics.}
    \label{fig:red-ref-mos}
  \end{subcaptionblock}
  \hspace{.7em}
  \begin{subcaptionblock}{.41\textwidth}
    \small\tightmathenv
    \begin{gather*}
      \frac{
        \varLi \teq \nul,\kd \varB \teq \dval \kin \DEnv
      }{
        \var \teq \newRef\kd \varLi\kd \varB,\ka \xDCfg \kto
          \var \teq \kRef \varB,\ka \xDCfg
      }
    \\[.1em]
      \frac{
        \varB \teq \kRef \varC \kin \DEnv
      }{
        \var \teq \freeRef\kd \varB,\ka \xDCfg \kto
          \var \teq \varC,\ka \xDCfg
      }
    \end{gather*}
    \caption{For denotational semantics.}
    \label{fig:red-ref-dos}
  \end{subcaptionblock}
  \tightcaption
  \caption{Basic reference-related reduction rules.}
  \label{fig:red-ref}
\end{figure}

\begin{figure}
  \begin{subcaptionblock}{\textwidth}
    \small\tightmathenv
    \begin{gather*}
      \frac{
        \varLi \teq \nul,\kc \varF \teq \mval \kin \MEnv
      }{\begin{aligned} &
        \var \teq \newLifetime\kd \varLi\kd \varF,\ka \xMCfg \kto
      \\[-.6em] & \hspace{5em}
          \varNow \teq \nul,\ka \var \teq \varF\kc \varNow,\ka \xMCfg
      \end{aligned}}
    \hspace{3em}
      \frac{
        \varNow \teq \nul \kin \MEnv
      }{\begin{aligned} &
        \var \teq \tendLifetime\, \varNow,\ka \xMCfg \kto
      \\[-.6em] & \hspace{8em}
          \var \teq \Ur\kd \varNow,\ka \xMCfg
      \end{aligned}}
    \\[-.2em]
      \frac{
        \varLi \teq \nul,\kc \varB \teq \mval \kin \MEnv
      }{\begin{aligned} &
        \var \teq \borrow\kd \varLi\kd \varB,\ka \xMCfg \kto
      \\[-.6em] & \hspace{3.5em}
          \varBm \teq \mval,\ka \varBl \teq \mval,\ka \var \teq \paren{\varBm, \varBl},\ka \xMCfg
      \end{aligned}}
    \hspace{3em}
      \frac{
        \varB \teq \mval \kin \MEnv
      }{
        \var \teq \share\kd \varB,\ka \xMCfg \kto
          \var \teq \Ur\kd \varB,\ka \xMCfg
      }
    \\[.1em]
      \frac{
        \varB \teq \mval \kin \MEnv
      }{\begin{aligned} &
        \var \teq \tcopy\kd \varB,\ka \xMCfg \kto
      \\[-.6em] & \hspace{6em}
          \var \teq \mval,\ka \xMCfg
      \end{aligned}}
    \hspace{2.5em}
      \frac{
        \varB \teq \mval \kin \MEnv
      }{\begin{aligned} &
        \var \teq \joinMut\kd \varB,\ka \xMCfg \kto
      \\[-.6em] & \hspace{6.5em}
          \var \teq \mval,\ka \xMCfg
      \end{aligned}}
    \hspace{2.5em}
      \frac{
        \varEnd \teq \nul,\kc \varB \teq \mval \kin \MEnv
      }{\begin{aligned} &
        \var \teq \reclaim\kd \varB\kd \varEnd,\ka \xMCfg \kto
      \\[-.6em] & \hspace{8.5em}
          \var \teq \varB,\ka \xMCfg
      \end{aligned}}
    \\[-1.8em]
    \end{gather*}
    \caption{For mutative operational semantics.}
    \label{fig:red-borrow-mos}
  \end{subcaptionblock}
  \begin{subcaptionblock}{\textwidth}
    \small\tightmathenv
    \begin{gather*}
      \frac{
        \varLi \teq \nul,\kc \varF \teq \dval \kin \DEnv
      }{\begin{aligned} &
        \var \teq \newLifetime\kd \varLi\kd \varF,\ka \xDCfg \kto
      \\[-.5em] & \hspace{6.5em}
          \varNow \teq \nul_\emphist,\ka \var \teq \varF\kc \varNow,\ka \xDCfg
      \end{aligned}}
    \hspace{3em}
      \frac{
        \varNow \teq \nul_\Hist \kin \DEnv
      }{
        \var \teq \tendLifetime\, \varNow,\ka \xDCfg \kto
          \var \teq \Ur\kd \varNow,\ka \xDCfg
      }
    \\[.1em]
      \frac{
        \varLi \teq \nul,\kc \varB \teq \dval \kin \DEnv
      }{\begin{aligned} &
        \var \teq \borrow\kd \varLi\kd \varB,\ka \xDCfg \kto
          \varBm \teq \Mut^\bid\kb \dval,\ka
      \\[-.5em] & \hspace{6.5em}
          \varBl \teq \Lend^\bid\kb \varB,\ka
          \var \teq \paren{\varBm, \varBl},\ka
          \bid,\ka \xDCfg
      \end{aligned}}
    \hspace{3em}
      \frac{
        \varB \teq \Mut^{\seq{\bpath}}\kc \dval \kin \DEnv
      }{
        \var \teq \share\kd \varB,\ka \xDCfg \kto
          \varB' \teq \Share\kd \dval,\ka
          \var \teq \Ur\kd \varB,\ka \xDCfg
      }
    \\[.0em]
      \frac{
        \varB \teq \Share\kd \dval \kin \DEnv
      }{
        \var \teq \tcopy\kd \varB,\ka \xDCfg \kto
          \var \teq \dval,\ka \xDCfg
      }
    \hspace{2.75em}
      \frac{
        \varB \teq \DBcon\kc \paren{\Mut^{\seq{\bpathB}}\ka \dval} \kin \DEnv
      \hspace{1.5em}
        \DBcon' \kb=\kb \begin{cases} \kd
          \Mut^{\ka\seq{\bpath},\ka \seq{\bpathB}} & \nk{11} \DBcon = \Mut^{\seq{\bpath}}
        \\[-.4em] \kd
          \Share & \nk{11} \DBcon = \Share
        \end{cases}
      }{
        \var \teq \joinMut\kd \varB,\ka \xDCfg \kto
          \var \teq \DBcon'\kb \dval,\ka \xDCfg
      }
      \tagsp\ruletag{dred-join-bor}
    \\[.15em]
      \frac{
        \varEnd \teq \nul_\Hist,\kc \varB \teq \Lend^\bid\kb \varC \kin \DEnv
      \hspace{2em}
        \restore_\Hist\kc \bid\ke \DEnv \kb\turnsto{\varC}{\varC'} \DEnv'
      }{
        \var \teq \reclaim\kd \varB\kd \varEnd,\ka \xDCfg \kto
          \var \teq \varC',\ka \DEnv';\kb \var
      }
      \tagsp\ruletag{dred-reclaim}
    \end{gather*}
    \caption{For denotational operational semantics.}
    \label{fig:red-borrow-dos}
  \end{subcaptionblock}
  \begin{subcaptionblock}{\textwidth}
    \small\tightmathenv
    \begin{gather*}
      \frac{
        \all{\seq{\idx}}\kc \bpath\Seq{.\idx} \ka\notin\ka \dom \Hist
      }{
        \restore_\Hist\ke \bpath\ke \DEnv
          \turnsto{\var}{\var} \DEnv
      }
    \hspace{3em}
      \frac{
        \bpath \hstore \varB \kin \Hist
      \hspace{1.5em}
        \var \teq \kRef \varC \kin \DEnv
      \hspace{1.5em}
        \restore_\Hist\kf \bpath.0\kf \DEnv
          \kb\turnsto{\varB}{\varB'} \DEnv'
      }{
        \restore_\Hist\kf \bpath\kf \DEnv \kb\turnsto{\var}{\var'}
          \var' \teq \kRef \varB',\kb \DEnv'
      }
    \\[.3em]
      \frac{\begin{gathered}
        \bpath \ka\notin\ka \dom \Hist
      \hspace{1.5em}
        \bpath\Seq{.\idxB} \in \dom \Hist
      \\[-.5em]
        \var \kteq \DBctx\kb \paren{\kRef \varB} \kin \DEnv
      \hspace{1.5em}
        \restore_\Hist\kf \bpath.0\kf \DEnv
          \turnsto{\varB}{\varB'} \DEnv'
      \end{gathered}}{
        \restore_\Hist\kf \bpath\kf \DEnv \kb\turnsto{\var}{\var'}
          \var' \kteq \DBctx\kb \paren{\kRef \varB},\kb \DEnv'
      }
    \hspace{2.5em}
      \frac{\begin{gathered}
        \bpath \ka\notin\ka \dom \Hist
      \hspace{1.5em}
        \bpath\Seq{.\idxB} \in \dom \Hist
      \hspace{1.5em}
        \var \kteq \DBctx\kb \paren{\Con\kd \seq{\varB}^\num}
          \kin \DEnv
      \\[-.5em]
        \DEnv_0 \teq \DEnv
      \hspace{1.5em}
        \all{\idx}\kb
        \restore_\Hist\kf \bpath.\idx\kf \DEnv_\idx
          \turnsto{\varB_{\kz\idx}}{\varB'_{\kz\idx}} \DEnv_{\idx + 1}
      \end{gathered}}{
        \restore_\Hist\kf \bpath\kf \DEnv \kb\turnsto{\var}{\var'}
          \var' \kteq \DBctx\kb \paren{\Con\kd \seq{\varB'}},\kb \DEnv_\num
      }
    \end{gather*}
    \caption{The restoration-by-history predicate \(\restore_\Hist\ke \bpath\ke \DEnv \turnsto{\var}{\varB} \DEnv'\).}
    \label{fig:red-borrow-restore}
  \end{subcaptionblock}
  \tightcaption
  \caption{Borrow-related reduction rules.}
  \label{fig:red-borrow}
\end{figure}

\begin{figure}
  \begin{subcaptionblock}{\textwidth}
    \small\tightmathenv
    \begin{gather*}
      \frac{
        \varNow \teq \nul,\kc \varBo \teq \mval \kin \MEnv
      }{\begin{aligned} &
        \var \teq \execBO\kd \varNow\kd \varBo,\ka \xMCfg \kto
      \\[-.5em] & \hspace{3em}
          \varRes \teq \exeBO\kd \varBo,\ka
          \var \teq \execBO^\post\kc \varRes,\ka \xMCfg
      \end{aligned}}
    \hspace{3em}
      \frac{
        \varRes \teq \Done\kd \varB \kin \MEnv
      }{\begin{aligned} &
        \var \teq \execBO^\post\kc \varRes,\ka \xMCfg \kto
      \\[-.5em] & \hspace{3em}
          \varNow \teq \nul,\ka \var \teq \paren{\kz\varNow, \varB}\kc,\ka \xMCfg
      \end{aligned}}
    \\[.3em]
      \frac{
        \varBo \teq \pure\kd \varB \kin \MEnv
      }{\begin{aligned} &
        \var \teq \exeBO\kd \varBo,\ka \xMCfg \kto
      \\[-.5em] & \hspace{7em}
          \var \teq \Done\kd \varB,\ka \xMCfg
      \end{aligned}}
    \hspace{3em}
      \frac{
        \varBo \teq \varBo' \bind \varKo \kin \MEnv
      }{\begin{aligned} &
        \var \teq \exeBO\kd \varBo,\ka \xMCfg \kto
      \\[-.5em] & \hspace{3em}
          \varRes \teq \exeBO\kd \varBo',\ka
          \var \teq \varRes \bind^\post \varKo,\ka \xMCfg
      \end{aligned}}
    \\[.3em]
      \frac{
        \varRes \teq \Done_\Hist\kc \varB \kin \MEnv
      }{\begin{aligned} &
        \var \teq \varRes \bind^\post \varKo;\ka \xMCfg \kto
      \\[-.5em] & \hspace{3em}
          \varBo \teq \varKo\kd \varB,\ka
          \var \teq \exeBO\kd \varBo,\ka \xMCfg
      \end{aligned}}
    \hspace{3em}
      \frac{
        \varBo \teq \sexecBO\kd \varNow\kd \varBo' \kin \MEnv
      }{\begin{aligned} &
        \var \teq \exeBO\kd \varBo,\ka \xMCfg \kto
      \\[-.5em] & \hspace{3em}
          \var \teq \sexecBO^\pre\kc \varNow\kd \varBo',\ka \xMCfg
      \end{aligned}}
    \\[.3em]
      \frac{
        \varNow \teq \nul,\kb \varBo' \teq \mval \kin \MEnv
      }{\begin{aligned} &
        \var \teq \sexecBO^\pre\kc \varNow\kd \varBo',\ka \xMCfg \kto
      \\[-.5em] & \hspace{3em}
          \varRes \teq \exeBO\kd \varBo',\ka
          \var \teq \sexecBO^\post\kc \varRes,\ka \xMCfg
      \end{aligned}}
    \hspace{3em}
      \frac{
        \varRes \teq \Done\kd \varB \kin \MEnv
      }{\begin{aligned} &
        \var \teq \sexecBO^\post\kc \varRes,\ka \xMCfg \kto
      \\[-.5em] & \hspace{7em}
          \var \teq \Done\kd \varRes,\ka \xMCfg
      \end{aligned}}
    \\[.3em]
      \frac{
        \varBo \teq \parBO\kd \varBo_0\kd \varBo_1 \kin \MEnv
      }{\begin{aligned} &
        \var \teq \exeBO\kd \varBo,\ka \xMCfg \kto
          \varRes_0 \teq \exeBO\kd \varBo_0,\ka
      \\[-.4em] & \hspace{2em}
          \varRes_1 \teq \exeBO\kd \varBo_1,\ka
          \var \teq \parBO^\post\kc \varRes_0\kd \varRes_1,\ka \xMCfg
      \end{aligned}}
    \hspace{3em}
      \frac{
        \varRes_0 \teq \Done\kd \varB,\kc
        \varRes_1 \teq \Done\kd \varC \kin \MEnv
      }{\begin{aligned} &
        \var \teq \parBO^\post\kc \varRes_0\kd \varRes_1;\ka \xMCfg \kto
      \\[-.5em] & \hspace{3em}
          \varPr \teq \paren{\varB, \varC},\ka \var \teq \Done\kd \varPr,\ka \xMCfg
      \end{aligned}}
    \end{gather*}
    \caption{For mutative operational semantics.}
    \label{fig:red-monad-basic-mos}
  \end{subcaptionblock}
  \begin{subcaptionblock}{\textwidth}
    \small\tightmathenv
    \begin{gather*}
      \frac{
        \varNow \teq \nul_\Hist,\kc \varBo \teq \dval \kin \DEnv
      }{\begin{aligned} &
        \var \teq \execBO\kd \varNow\kd \varBo,\ka \xDCfg \kto
      \\[-.5em] & \hspace{3em}
          \varRes \teq \exeBO_\Hist\kc \varBo,\ka
          \var \teq \execBO^\post\kc \varRes,\ka \xDCfg
      \end{aligned}}
    \hspace{3em}
      \frac{
        \varRes \teq \Done_\Hist\kc \varB \kin \DEnv
      }{\begin{aligned} &
        \var \teq \execBO^\post\kc \varRes,\ka \xDCfg \kto
      \\[-.5em] & \hspace{3em}
          \varNow \teq \nul_\Hist,\ka \var \teq \paren{\kz\varNow, \varB}\kc,\ka \xDCfg
      \end{aligned}}
    \\[.2em]
      \frac{
        \varBo \teq \pure\kd \varB \kin \DEnv
      }{
        \var \teq \exeBO_\Hist\kc \varBo,\ka \xDCfg \kto
          \var \teq \Done_\Hist\kc \varB,\ka \xDCfg
      }
    \hspace{3em}
      \frac{
        \varBo \teq \varBo' \bind \varKo \kin \DEnv
      }{\begin{aligned} &
        \var \teq \exeBO_\Hist\kc \varBo,\ka \xDCfg \kto
      \\[-.5em] & \hspace{3em}
          \varRes \teq \exeBO_\Hist\kd \varBo',\ka
          \var \teq \varRes \bind^\post \varKo,\ka \xDCfg
      \end{aligned}}
    \\[.3em]
      \frac{
        \varRes \teq \Done_\Hist\kc \varB \kin \DEnv
      }{\begin{aligned} &
        \var \teq \varRes \bind^\post \varKo;\ka \xDCfg \kto
      \\[-.5em] & \hspace{3em}
          \varBo \teq \varKo\kd \varB,\ka
          \var \teq \exeBO_\Hist\kc \varBo,\ka \xDCfg
      \end{aligned}}
    \hspace{3em}
      \frac{
        \varBo \teq \sexecBO\kd \varNow\kd \varBo' \kin \DEnv
      }{\begin{aligned} &
        \var \teq \exeBO_\Hist\kc \varBo,\ka \xDCfg \kto
      \\[-.5em] & \hspace{4em}
          \var \teq \sexecBO^\pre_\Hist\kc \varNow\kd \varBo',\ka \xDCfg
      \end{aligned}}
    \\[.3em]
      \frac{
        \varNow \teq \nul_{\Histpr},\kc \varBo' \teq \dval \kin \DEnv
      }{\begin{aligned} &
        \var \teq \sexecBO^\pre_\Hist\kc \varNow\kd \varBo',\ka \xDCfg \kto
      \\[-.5em] & \hspace{1.5em}
          \varRes \teq \exeBO_\emphist\kc \varBo',\ka
          \var' \teq \sexecBO^\post_{\Hist; \Histpr}\kb \varRes,\ka \xDCfg
      \end{aligned}}
    \hspace{2.5em}
      \frac{
        \varRes \teq \Done_{\Hist_0}\kc \varB \kin \DEnv
      }{\begin{aligned} &
        \var \teq \sexecBO^\post_{\Hist; \Histpr}\kb \varRes,\ka \xDCfg \kto
          \varNow \teq \nul_{\Histpr \hseq \Hist_0},\ka
      \\[-.4em] & \hspace{5.5em}
          \varPr \teq \paren{\kz\varNow, \varB},\ka
          \var \teq \Done_{\Hist \hseq \Hist_0}\kc \varPr,\ka \xDCfg
      \end{aligned}}
    \\[.3em]
      \frac{
        \varBo \teq \parBO\kd \varBo_0\kd \varBo_1 \kin \DEnv
      }{\begin{aligned} &
        \var \teq \exeBO_\Hist\kc \varBo,\ka \xDCfg \kto
          \varRes_0 \teq \exeBO_\emphist\kc \varBo_0,\ka
      \\[-.5em] & \hspace{2.5em}
          \varRes_1 \teq \exeBO_\emphist\kc \varBo_1,\ka
          \var \teq \parBO^\post_\Hist\kd \varRes_0\kd \varRes_1,\ka \xDCfg
      \end{aligned}}
    \hspace{2.5em}
      \frac{
        \varRes_0 \teq \Done_{\Hist_0}\kc \varB,\kc
        \varRes_1 \teq \Done_{\Hist_1}\kc \varC \kin \DEnv
      }{\begin{aligned} &
        \var \teq \parBO^\post_\Hist\kd \varRes_0\kd \varRes_1;\ka \xDCfg \kto
      \\[-.5em] & \hspace{2.5em}
          \varPr \teq \paren{\varB, \varC},\ka
          \var \teq \Done_{\Hist \hseq \paren{\Hist_0 \hpar \Hist_1}}\kc \varPr,\ka \xDCfg
      \end{aligned}}
    \end{gather*}
    \caption{For denotational operational semantics.}
    \label{fig:red-monad-basic-dos}
  \end{subcaptionblock}
  \tightcaption
  \caption{Basic monadic reduction rules.}
  \label{fig:red-monad-basic}
\end{figure}

\begin{figure}
  \begin{subcaptionblock}{\textwidth}
    \small\tightmathenv
    \begin{gather*}
      \frac{
        \varBo \teq \deref\kd \varRef \kin \MEnv
      }{\begin{aligned} &
        \var \teq \exeBO\kd \varBo,\ka \xMCfg \kto
      \\[-.6em] & \hspace{6em}
          \var \teq \deref^\post\kb \varRef,\ka \xMCfg
      \end{aligned}}
    \hspace{3em}
      \frac{
        \varRef \teq \kRef \loc \kin \MEnv
      \hspace{2em}
        \loc \mMapsto \varB \kin \Mem
      }{\begin{aligned} &
        \var \teq \deref^\post\kb \varRef,\ka \xMCfg \kto
      \\[-.6em] & \hspace{4em}
          \varB' \teq \varB,\ka
          \var \teq \Done\kd \varB',\ka \xMCfg
      \end{aligned}}
    \\[-.1em]
      \frac{
        \varBo \teq \updateRef\kd \varKo\kd \varRef \kin \MEnv
      }{\begin{aligned} &
        \var \teq \exeBO\kd \varBo,\ka \xMCfg \kto
      \\[-.5em] & \hspace{2em}
          \var \teq \updateRef^\pre\kb \varKo\kd \varRef,\ka \xMCfg
      \end{aligned}}
    \hspace{3em}
      \frac{
        \varKo \teq \mval,\kc \varRef \teq \kRef \loc \kin \MEnv
      \hspace{2em}
        \loc \mMapsto \varB \kin \Mem
      }{\begin{aligned} &
        \var \teq \updateRef^\pre\kb \varKo\kd \varRef,\ka \xMCfg \kto
          \varBo \teq \varKo\kd \varB,\ka
      \\[-.5em] & \hspace{4em}
          \varRes \teq \execBO\kd \varBo,\ka
          \var \teq \updateRef^\prepost_\loc\kd \varRes,\ka \xMCfg
      \end{aligned}}
    \\[-.1em]
      \frac{
        \varRes \teq \Done\kd \varPr \kin \MEnv
      }{\begin{aligned} &
        \var \teq \updateRef^\prepost_\loc\kd \varRes,\ka \xMCfg \kto
      \\[-.6em] & \hspace{5em}
          \var \teq \updateRef^\post_\loc\kd \varPr,\ka \xMCfg
      \end{aligned}}
    \hspace{3em}
      \frac{
        \varPr \teq \paren{\varC, \varB} \kin \MEnv
      }{\begin{aligned} &
        \var \teq \updateRef^\post_\loc\kd \varPr,\ka \xMCfg \kto
          \varRef \teq \kRef \loc,\ka
      \\[-.6em] & \hspace{4em}
          \varPr' \teq \paren{\varC, \varRef},\ka
          \var \teq \Done\kd \varPr',\ka
          \MEnv;\kb \Mem \mapupd{\loc}{\varB};\kb \var
      \end{aligned}}
    \end{gather*}
    \caption{For mutative operational semantics.}
    \label{fig:red-monad-ref-mos}
  \end{subcaptionblock}
  \begin{subcaptionblock}{\textwidth}
    \small\tightmathenv
    \setlength{\abovedisplayskip}{-4pt}
    \begin{gather*}
      \frac{
        \varBo \teq \deref\kd \varRef \kin \DEnv
      }{
        \var \teq \exeBO_\Hist\kc \varBo,\ka \xDCfg \kto
          \var \teq \deref^\post_\Hist\kb \varRef,\ka \xDCfg
      }
    \hspace{2.25em}
      \frac{
        \varRef \teq \DBcon\kb \paren{\kRef \varB} \kin \DEnv
      \hspace{1.5em}
        \DBcon' \kc=\kc \begin{cases} \kd
          \Mut^{\Seq{\bpath.0}} \ku & \DBcon \teq \Mut^{\seq{\bpath}}
        \\[-.4em] \kd
          \Share \ku & \DBcon \teq \Share
        \end{cases}
      }{
        \var \teq \deref^\post_\Hist\kb \varRef,\ka \xDCfg \kto
          \varB' \teq \DBcon'\kb \varB,\ka
          \var \teq \Done_\Hist\kc \varB',\ka \xDCfg
      }
    \\[.4em]
      \frac{
        \varBo \teq \updateRef\kd \varKo\kd \varRef \kin \DEnv
      }{\begin{aligned} &
        \var \teq \exeBO_\Hist\kd \varBo,\ka \xDCfg \kto
      \\[-.5em] & \hspace{3em}
          \var \teq \updateRef^\pre_\Hist\kc \varKo\kd \varRef,\ka \xDCfg
      \end{aligned}}
    \hspace{3em}
      \frac{
        \varKo \teq \dval,\kc \varRef \teq \Mut^{\seq{\bpath}}\kb \paren{\kRef \varB} \kin \DEnv
      }{\begin{aligned} &
        \var \teq \updateRef^\pre_\Hist\kc \varKo\kd \varRef,\ka \xDCfg \kto
          \varBo \teq \varKo\kd \varB,\ka
      \\[-.5em] & \hspace{4em}
          \varRes \teq \exeBO_\Hist\kc \varBo,\ka
          \var \teq \updateRef^\prepost_{\seq{\bpath}}\kc \varRes,\ka \xDCfg
      \end{aligned}}
    \\[.3em]
      \frac{
        \varRes \teq \Done_\Hist\kc \varPr \kin \DEnv
      }{\begin{aligned} &
        \var \teq \updateRef^\prepost_{\seq{\bpath}}\kc \varRes,\ka \xDCfg \kto
      \\[-.4em] & \hspace{4em}
          \var \teq \updateRef^\post_{\seq{\bpath};\ka \Hist}\kc \varPr,\ka \xDCfg
      \end{aligned}}
    \hspace{3em}
      \frac{
        \varPr \teq \paren{\varC, \varB} \kin \DEnv
      }{\begin{aligned} &
        \var \teq \updateRef^\post_{\seq{\bpath};\ka \Hist}\kb \varPr,\ka \xDCfg \kto
          \varRef \teq \Mut^{\seq{\bpath}}\kb \paren{\kRef \varB},\kb
      \\[-.5em] & \hspace{8em}
          \varPr' \teq \paren{\varC, \varRef},\ka
          \var \teq \Done_{\Hist \hseq \Seq{\bpath \ka\hstore\ka \varB}}\kc \varPr',\ka \xDCfg
      \end{aligned}}
    \end{gather*}
    \caption{For denotational operational semantics.}
    \label{fig:red-monad-ref-dos}
  \end{subcaptionblock}
  \tightcaption
  \caption{Reference-related monadic reduction rules.}
  \label{fig:red-monad-ref}
\end{figure}

Now we inductively define the small-step reduction relations \(\MCfg \to \MCfg'\) and \(\DCfg \to \DCfg'\) for mutative and denotational operational semantics, respectively.

The basic reduction rules for the two semantics are listed in \cref{fig:red-basic}.\footnote{
  The function \(\asBool\) maps the truth and falsehood to \(\True\) and \(\False\), respectively.
}
The rules \ref{rule:mred-loop} and \ref{rule:dred-loop} use the forcing loop predicate \(\MCfg \hasloop\), \(\DCfg \hasloop\), meaning that the environment contains variables that circularly forcing one another, \emph{coinductively} defined by the rules in \cref{fig:forceloop}.
The two semantics are mostly the same in these rules, but there are some notable differences due to borrower constructors \(\DBcon\), appearing only in denotational semantics.
The rule \ref{rule:dred-var} for a variable term carries the borrower context \(\DBctx\).
Also, the rule \ref{rule:dred-case-bor} distributes the borrower constructor \(\DBcon\) over the data components (recall the typing rule \ref{rule:ty-case-bor} in \cref{fig:typing}).

The reduction rules for basic operators for denotational semantics are listed in \cref{fig:red-basicop}.
The rules for mutative semantics are just analogous to them.
The operator \(\linearly\) allows binding the result of a linear computation to an unrestricted variable.
We use an extra operator \(\linear\) to clarify that boundary for our association system \cref{appx:assoc}.

The basic reference-related reduction rules for the two semantics are presented in \cref{fig:red-ref}.
In mutative semantics, a reference is modeled as \(\kRef \loc\) with a location \(\loc\), whose content \(\var\) is separately stored as \(\loc \mMapsto \var\) in the global memory \(\Mem\);
in denotational semantics, a reference is modeled as \(\kRef \var\), where the body content \(\var\) is locally stored in the value.
The ownership justifies this.

The borrow-related reduction rules for the two semantics are listed in \cref{fig:red-borrow}.
Here we clearly see the differences between the two semantics.
In denotational semantics, for each new borrow, a fresh borrow id \(\bid\) is created, and the lender is modeled as \(\Lend^\bid\kb \var\), where the \(\var\) represents the original version of the borrowed object.
Notably, when the lender reclaims the ownership of the borrowed object by \ref{rule:dred-reclaim}, the latest version of the object is restored from the history.
For this, we introduce the restoration-by-history predicate \(\restore_\Hist\ke \bpath\ke \DEnv \turnsto{\var}{\varB} \DEnv'\), inductively defined by the rules listed in \cref{fig:red-borrow-restore}.
The predicate adds new items to the environment \(\DEnv\) to get \(\DEnv'\), reflecting updates recorded in the history \(\Hist\) on the borrow path \(\bpath\) for the variable \(\var\), returning a new variable \(\var'\).
The predicate does nothing and just returns \(\DEnv\) and \(\var\) if the history does not contain any relevant update,
while the predicate takes a fresh variable \(\var'\) for the new version of \(\var\) otherwise.

Basic monadic reduction rules for denotational semantics are listed in \cref{fig:red-monad-basic}.
The rules for mutative semantics are just analogous to them, just without the history information \(\Hist\).
We internally use the extra operator \(\exeBO\)/\(\exeBO_\Hist\), a utility variant of \(\execBO\), to execute a monad without carrying around the live lifetime token.
It returns \(\Done\kd \var\)/\(\Done_\Hist\kc \var\) for the result of the computation.
What is important is how to handle histories.
The scoped execution \(\sexecBO\) executes the monad \(\varBo\) with the empty history, whose resulting history \(\Hist_0\) is then sequentially composed with the original histories \(\Hist\), \(\Histpr\).
The parallel execution \(\parBO\) executes the monads \(\varBo_0\), \(\varBo_1\) respectively with the empty history and then composes the resulting histories \(\Hist_0\), \(\Hist_1\) \emph{in parallel}, whose composite is sequentially composed with the original history \(\Hist\).

Reference-related monadic reduction rules for the two semantics are listed in \cref{fig:red-monad-ref}.
Like in \cref{fig:red-ref}, mutative operational semantics accesses the memory \(\Mem\) while denotational operational semantics locally operates on the body of \(\TRef\).
Crucially, in denotational semantics, the reference update \(\updateRef\) records the update \(\bpath \hstore \varB\) to the history \(\Hist\) for each borrow path \(\bpath\) of the mutable borrower.
Also, the dereference \(\deref\) distributes the borrower constructor, modifying the borrow path, just like in \ref{rule:dred-case-bor}.

\section{Association System}
\label{appx:assoc}

We present the omitted details of our association system introduced in \cref{sect:metatheory-assoc}.

\subsection{Overview}
\label{appx:assoc-overview}

\paragraph{Syntax}

\begin{figure}
  \tightmathenv
  \begin{gather*}
    \stag{Lifetime path} \lftpath \sdef \lftid \sor \lftpath.\idx
  \hspace{3em}
    \stag{Extended type} \RTy \sdef \Ty
      \sor \Done^{\seq{\lftpath}}\kc \Ty
      \sor \kSkel \Ty
  \\[.1em]
  \begin{aligned}
    \stag{Ghost resource} \gho
      & \sdef \glinear
      \sor \gnow^{\lftpath;\ka \num}_{\Hist;\ka \Bpaths}
      \sor \gend^\lftid_\Hist
      \sor \glifetime^\bid_\lft
      \sor \gborrow^{\seq{\bpath}}_{\var;\ka \Ty;\ka \seq{\bpathB}}
  \\[-.4em] & \hspace{1em}
      \sor \gmut^{\seq{\bpath}}
      \sor \gshare^{\seq{\bpath}}
      \sor \glend^\bid_{\var;\ka \Ty}
      \sor \gdepo^{\seq{\bpath}}
      \sor \ggdepo^\bid_{\var;\ka \Ty;\ka \Hist}
  \\[-.4em] & \hspace{1em}
      \sor \ghole^{\seq{\bpath}}
      \sor \grelend^{\seq{\bpath}}_\Hist
      \sor \loc \borMapsto{\seq{\bpath}} \var
      \sor \loc \bormapsto{\seq{\bpath}} \var
      \sor \gskel^\var_\Ty
      \sor \gvar^\var_\Ty
  \end{aligned}
  \\[.1em]
    \stag{Atomic resource} \atom
      \sdef \var \ccolul{\Bpaths}{\mult} \RTy
      \sor \bid
      \sor \var \teq \mterm \assoc \dterm
      \sor \var \teq \dval
      \sor \loc \mMapsto \var
      \sor \loc \dmapsto \var
      \sor \all{\anyvar} \Rctx
      \sor \gho
  \\[.1em]
    \stag{Resource context} \Rctx, \RctxB
      \sdef \Seq{\atom}
  \end{gather*}
  \caption{Syntax for the association system.}
  \label{fig:assoc-syntax}
\end{figure}

The syntax for the association system is shown in \cref{fig:assoc-syntax}.

A lifetime path \(\lftpath = \lftid\Seq{.\idx}\) represents a part of an atomic lifetime \(\al \lftid\).
We introduce extended types \(\RTy\), which include \(\Done^{\seq{\lftpath}}\kc \Ty\) to represent the result of \(\exeBO\)/\(\exeBO_\Hist\)
and the skeleton type \(\kSkel \Ty\) to represent a borrower whose lender has already reclaimed.
The resource context \(\Rctx\) extends the typing context \(\Tctx\) with more information about the dynamic ownership.\footnote{
  The order of items in a resource context \(\Rctx\) is ignored.
}
A variable-type binding \(\var \ccolul{\Bpaths}{\mult} \RTy\) in a resource context is marked with the set of borrow paths \(\Bpaths\) that the variable may be able to update (we omit the superscript \(\Bpaths\) if it is empty).
This information is vital to justify \(\parBO\), ensuring that the borrow paths updated by the two monad arguments are mutually disjoint.
A resource context also contains borrow ids \(\bid\), variable-term bindings \(\var \teq \mterm \assoc \dterm\), \(\var \teq \dval\), and points-to tokens \(\loc \mMapsto \var\) from the environments and memory.
Importantly, bindings \(\var \teq \mval \assoc \dval\), \(\var \teq \dval\) over values are shareable.
A captured points-to token \(\loc \dmapsto \var\) represents a usual points-to token captured inside \(\linear\) (please see \ref{rule:assoc-linear} shown later).

Interestingly, a resource context can even contain polymorphic resource contexts \(\all{\anyvar} \Rctx\) to reason about \emph{first-class polymorphism} of Haskell.
For example, we can create a polymorphic reference \(\all{\Tvar} \kRef \sqbr{\Tvar}\) to a list of any type from the nil list \(\sqbr{} \ccol \all{\Tvar}\kz \sqbr{\Tvar}\), and then borrow it to create a polymorphic pair \(\all{\lftvar, \Tvar} \paren{\Mut^\lftvar\ka \paren{\kRef \sqbr{\Tvar}},\kb \Lend^\lftvar\ka \paren{\kRef \sqbr{\Tvar}}}\) of a borrower and a lender.
Because the variables \(\lftvar\), \(\Tvar\) are instantiated dynamically, we have to keep the polymorphism over \(\lftvar\), \(\Tvar\) in creating the ghost state for the borrow.

A resource context can contain ghost resources \(\gho\) to represent various auxiliary information.
The meanings of ghost resources are as follows.
A linearity witness \(\glinear\) is used for \(\Linearly\).
A live lifetime token \(\gnow^{\lftid\Seq{.\idx};\kb \num}_{\Hist;\ka \Bpaths}\) owns the part \(\Seq{.\idx}\) of the lifetime \(\al \lftid\), asserting that it is alive with the partial history \(\Hist\), has \(\num\) children tokens, and has the right to perform updates of borrow paths in \(\Bpaths\).
A dead lifetime token \(\gend^\lftid_\Hist\) asserts that the lifetime \(\al \lftid\) has ended with the history \(\Hist\).
A borrow lifetime token \(\glifetime^\bid_\lft\) asserts that the lifetime of the borrow of the id \(\bid\) is \(\lft\).
A borrow path token \(\gborrow^{\seq{\bpath}}_{\var;\ka \Ty;\ka \seq{\bpathB}}\) asserts that the borrow of the borrow path sequence \(\seq{\bpath}\) owns an object of a variable \(\var\) typed \(\Ty\).
When \(\seq{\bpathB}\) is non-empty, it means that the object is also reborrowed under \(\seq{\bpathB}\).
When \(\seq{\bpathB}\) is empty, we omit the last part `\(; \seq{\bpathB}\)' of the subscript.
We have tokens for a mutable borrower \(\gmut^{\seq{\bpath}}\), a shared borrower \(\gshare^{\seq{\bpath}}\), and a lender \(\glend^\bid_{\mval;\ka \var}\).
We also have tokens for a (local) deposit \(\gdepo^{\seq{\bpath}}\) and a global deposit \(\ggdepo^\bid_{\var;\ka \Ty;\ka \Hist}\) for the borrowed contents.
The hole token \(\ghole^{\seq{\bpath}}\) is used inside a global deposit to indicate a hole caused by a deposit \(\gdepo^{\seq{\bpath}}\).
The relender token \(\grelend^{\seq{\bpath}}_\Hist\) asserts that the relender of \(\seq{\bpath}\) is under the updates of \(\Hist\).
We also have a shared-borrowed points-to token \(\loc \borMapsto{\vphantom{y}\smash{\seq{\bpath}}} \var\) and the evidence of the ownership \(\loc \bormapsto{\vphantom{y}\smash{\seq{\bpath}}} \var\).
Also, we have a shareable token \(\gskel^\var_\Ty\) for the skeleton type \(\kSkel \Ty\) and a variable token \(\gvar^\var_\Ty\) providing skeletons with the witness that the variable \(\var\) is typed \(\Ty\).

\paragraph{Assocaition rules}

\begin{figure}
  \small\tightmathenv
  \begin{gather*}
    \frac{
      \Seq{\varB \teq \mtermB \assoc \dtermB},\kb
      \Seq{\varB' \ky\teq \mval \assoc \dval},\kb
      \Seq{\varB' \ky\teq \dval},\kb \Seq{\varC \teq \dvalB},\kb
      \Seq{\bid},\kb \Mem \kvdash
        \mterm \assoc \dterm \kccol \Ty
    }{
      \var \teq \mterm,\ka \Seq{\varB \teq \mtermB},\kb
      \Seq{\varB' \teq \mval},\kb
        \MEnv;\kb \Mem;\kb \var
      \ke\assoc_\Ty\ke
      \var \teq \dterm,\kb \Seq{\varB \teq \dtermB},\kb
      \Seq{\varB' \teq \dval},\kb
      \Seq{\varC \teq \dvalB},\kb \Seq{\bid},\kb
        \DEnv;\kb \var
    }
    \tagsp\ruletag{config-assoc}
  \\[.3em]
    \frac{
      \RctxB \subty \Rctx
    \hspace{1.5em}
      \Rctx \vdash \mterm \assoc \dterm \ccol \RTy
    }{
      \RctxB \vdash \mterm \assoc \dterm \ccol \RTy
    }
    \tagsp\ruletag[assoc-incl]{assoc-\(\subty\)}
  \hspace{3em}
    \frac{
      \Rctx \kvdash \mterm \assoc \dterm \kccol \Ty
    }{
      \paren{\all{\anyvar} \Rctx} \kvdash \mterm \assoc \dterm \kccol \all{\anyvar} \Ty
    }
    \tagsp\ruletag[assoc-all]{assoc-\(\forall\)}
  \\[.3em]
    \frac{
      \all{\idx}\kc \Rctx_\idx \kreal
    \hspace{2em}
      \all{\idx}\ke
      \Seq{\var \ccoll{\many} \RTyB},\ka
      \Rctx_\idx \ka+\ka \RctxB \kvdash
        \mtermB_\idx \assoc \dtermB_\idx \kccol \RTyB_\idx
    \hspace{2em}
      \Seq{\var \ccoll{\many} \RTyB},\ka \RctxB' \kvdash
        \mterm \assoc \dterm \kccol \RTy
    }{
      \Seq{\var \teq \mtermB \assoc \dtermB} \ka+\ka
        \sum_\idx \Rctx_\idx \ka+\ka \many\ka \RctxB \ka+\ka \RctxB' \kvdash
        \mterm \assoc \dterm \kccol \RTy
    }
    \tagsp\ruletag[assoc-var-ur]{assoc-var-\(\textup{\many}\)}
  \\[.3em]
    \frac{
      \Rctx \kvalid
    \hspace{1.75em}
      \Rctx' \kreal
    \hspace{1.75em}
      \Rctx \ka+\ka \Rctx' \kvdash \mtermB \assoc \dtermB \kccol \RTyB
    \hspace{1.75em}
      \var \ccolul{\ka\BorPaths\kd \Rctx}{\mult} \RTyB,\ka \RctxB \kvdash
        \mterm \assoc \dterm \kccol \RTy
    }{
      \var \teq \mtermB \assoc \dtermB \ka+\ka
        \Rctx' \ka+\ka \mult\ka \Rctx \ka+\ka \RctxB \kvdash
        \mterm \assoc \dterm \kccol \RTy
    }
    \tagsp\ruletag{assoc-var}
  \\[.3em]
    \frac{
      \Rctx \kghost,\ka \gvalid
    \hspace{1.5em}
      \Rctx,\ka \Seq{\loc \dmapsto \var},\ka \RctxB \kvdash
        \var \assoc \var \kccol \Ur\kd \Ty
    }{
      \Seq{\loc \dmapsto \var},\ka \RctxB \kvdash
        \linear\kd \var \kassoc \linear\kd \var \kkcol \Ur\kd \Ty
    }
    \tagsp\ruletag{assoc-linear}
  \hspace{3em}
    \glinear \kvdash \nul \kassoc \nul \kkcol \Linearly
  \\[.4em]
    \loc \dmapsto \var,\kb \var \ccolul{\Bpaths}{\one} \Ty \kvdash
      \kRef \loc \kassoc \kRef \var \kkcol \kRef \Ty
  \hspace{2.25em}
    \gnow^{\lftid;\ka 0}_{\Hist;\ka \BorPath} \kvdash
      \nul \kassoc \nul_\Hist \kkcol \Now^{\ka\al \lftid}
  \hspace{2.25em}
    \frac{
      \lft \ka\le\ka \al \lftid
    }{
      \gend^\lftid_\Hist \kvdash
        \nul \kassoc \nul_\Hist \kkcol \End^\lft
    }
  \\[-.4em]
    \frac{
      \all{\idx}\kc \Bpaths['] \subseteq \Bpaths_\idx
    }{
      \begin{aligned} &
        \Seq{\gnow^{\lftpath;\ka 0}_{\Hist;\ka \Bpaths}}\kb,\kb
        \varBo \ccolul{\Bpaths[']}{\one} \BO^{\kb\bigwedge \Seq{\al \floor{\lftpath}}}\kd \Ty \kkvdash {}
      \\[-.5em] & \hspace{4em}
          \exeBO\kd \varBo \kassoc \exeBO_\Hist\kc \varBo
          \kkcol \Done^{\seq{\lftpath}}\kc \Ty
      \end{aligned}
    }
  \hspace{3em}
  \begin{gathered}
  \\[-.2em]
    \Seq{\gnow^{\lftpath;\ka 0}_{\Hist;\ka \Bpaths}}\kb,\kb
    \var \ccolul{\Bpaths[']}{\one}\ky \Ty \kvdash
      \Done\kd \var \kassoc \Done_\Hist\kc \var
      \kkcol \Done^{\seq{\lftpath}}\kc \Ty
  \end{gathered}
  \\[.1em]
    \frac{
      \lftB \ka\le\ka \bigwedge_\idx \lft_\idx
    \hspace{1.5em}
      \Ty \subty \TyB
    \hspace{1.5em}
      \TyB \subty \Ty
    }{
      \gmut^{\seq{\bpath}},\kb \Seq{\glifetime^{\floor{\bpath}}_\lft},\kb
      \gborrow^{\seq{\bpath}}_{\var;\ka \Ty} \kkvdash
        \var \kassoc \Mut^{\seq{\bpath}}\kc \var
        \kkcol \Mut^\lftB\kc \TyB
    }
  \hspace{3em}
    \frac{
      \lft \le \al \lftid
    \hspace{2em}
      \Rctx \kvdash \mterm \assoc \dterm \kccol \kSkel \Ty
    }{
      \gend^\lftid_\Hist \ka+\ka \Rctx \kkvdash
        \mterm \kassoc \Mut^{\seq{\bpath}}\kc \dterm
        \kkcol \Mut^\lft\kc \Ty
    }
  \\[.2em]
    \frac{
      \lftB \ka\le\ka \bigwedge_\idx \lft_\idx
    \hspace{1.5em}
      \Ty \subty \TyB
    }{
      \gshare^{\seq{\bpath}},\kb \Seq{\glifetime^{\floor{\bpath}}_\lft},\kb
      \gborrow^{\seq{\bpath}}_{\var;\ka \Ty} \kkvdash
        \var \kassoc \Share\kd \var
        \kkcol \Share^\lftB\kc \TyB
    }
  \hspace{3em}
    \frac{
      \lft \le \al \lftid
    \hspace{2em}
      \Rctx \kvdash \mterm \assoc \dterm \kccol \kSkel \Ty
    }{
      \gend^\lftid_\Hist \ka+\ka \Rctx \kkvdash
        \mterm \kassoc \Share\kd \dterm
        \kkcol \Share^\lft\kc \Ty
    }
  \\[.2em]
    \frac{
      \lft \le \lftB
    \hspace{1.5em}
      \Ty \subty \TyB
    }{
      \glend^\bid_{\var;\ka \Ty},\kb \bid,\kb
      \glifetime^\bid_\lft\ka,\kb \var \teq \mval \assoc \dval \kkvdash
        \mval \kassoc \Lend^\bid\kb \var
        \kkcol \Lend^\lftB\kc \TyB
    }
  \\[.2em]
    \frac{
      \Rctx \kvdash \var \assoc \var \kccol \Ty
    \hspace{1.5em}
      \RctxB \kvdash \mterm \assoc \dterm \kkcol \RTyB
    }{
      \gdepo^{\seq{\bpath}},\kb
      \gborrow^{\seq{\bpath}}_{\var;\ka \Ty} \ky+\ka
      \Rctx \ka+\ka \RctxB \kkvdash
        \mterm \assoc \dterm \kkcol \RTyB
    }
  \hspace{3em}
    \frac{
      \Rctx \kvdashul{\bid}{\Hist} \var \kccol \Ty
    \hspace{1.5em}
      \RctxB \kvdash \mterm \assoc \dterm \kkcol \RTyB
    }{
      \ggdepo^\bid_{\var;\ka \Ty;\ka \Hist},\kb
      \Rctx \ka+\ka \RctxB \kkvdash
        \mterm \assoc \dterm \kkcol \RTyB
    }
  \\[.2em]
    \frac{
      \Rctx \kvdashul{\seq{\bpath}}{\Hist} \mterm \assoc \dterm \kccol \RTy
    }{
      \var \teq \mterm \assoc \dterm \ka+\ka \Rctx \kvdashul{\seq{\bpath}}{\Hist}
        \var \kccol \RTy
    }
  \hspace{3em}
    \frac{
      \Rctx \kkvdashul{\seq{\bpath}}{\Hist} \dval \kleadsto \dvalB
    }{
      \ghole^{\seq{\bpath}},\kb
      \gborrow^{\seq{\bpath}}_{\var;\ka \Ty;\ka \seq{\bpathB}},\kb
      \var \teq \mval \assoc \dvalB,\kb \Rctx \kvdashul{\seq{\bpath}}{\Hist}
        \mval \assoc \dval \kccol \Ty
    }
  \hspace{3em}
    \frac{
      \Rctx \kvdash \mterm \assoc \dterm \kccol \RTy
    }{
      \Rctx \kvdashul{\seq{\bpath}}{\emphist}
        \mterm \assoc \dterm \kccol \RTy
    }
  \\[.2em]
    \frac{
      \Rctx \kvdashul{\Seq{\bpath.0}}{\Hist} \varB \kccol \Ty
    }{
      \loc \dmapsto \varB,\kb
      \loc \bormapsto{\seq{\bpath}} \varB,\kb
      \Rctx \kkvdashul{\seq{\bpath}}{\Seq{\bpath \hstore \varB},\ka \Hist}
        \kRef \loc \kassoc \kRef \var \kkcol \kRef \Ty
    }
  \hspace{3em}
    \frac{
      \all{\idx}\kc \bpath_\idx \ka\notin\ka \dom \Hist
    \hspace{1.5em}
      \Hist \ne \emphist
    \hspace{1.5em}
      \Rctx \kvdashul{\Seq{\bpath.0}}{\Hist} \var \kccol \Ty
    }{
      \loc \dmapsto \var,\kb
      \loc \bormapsto{\seq{\bpath}} \var,\kb
      \Rctx \kkvdashul{\seq{\bpath}}{\Hist}
        \kRef \loc \kassoc \kRef \var \kkcol \kRef \Ty
    }
  \end{gather*}
  \caption{Selected association rules.}
  \label{fig:assoc-rules}
\end{figure}

Selected rules for the association judgment are shown in \cref{fig:assoc-rules}.\footnote{
  We define \(\floor{\lftpath}\) by \(\floor{\lftid\Seq{.\idx}} \ka\defeq\ka \lftid\) and
  \(\floor{\bpath}\) by \(\floor{\bid\Seq{.\idx}} \ka\defeq\ka \bid\).
}
\Cref{appx:assoc-details} presents the omitted rules and formal definitions.

The association relation over mutative and denotational configurations is \(\MCfg \assoc_\Ty \DCfg\).
Its only rule is \ref{rule:config-assoc}, which puts into the resource context the variable-term bindings, the borrow ids of the environments, and the points-to tokens of the memory, discarding some old variables.

The association judgment has the form \(\Rctx \kvdash \mterm \assoc \dterm \kccol \RTy\).
The association judgment is closed under resource context inclusion (\ref{rule:assoc-incl}), defined naturally similarly to the typing context inclusion \cref{fig:tctx-incl}.
For polymorphic types, we use polymorphic resource contexts (\ref{rule:assoc-all}).
For variable introduction, we use the condition \(\Rctx \kreal\), meaning that every item in \(\Rctx\) is of the form \(\var \teq \mterm \assoc \dterm\) or \(\loc \mMapsto \var\).

Unlike the unrestricted case (\ref{rule:assoc-var-ur}), introduction of a possibly linear variable (\ref{rule:assoc-var}) is tricky, because the set of accessible borrow paths should be determined from the resource context \(\Rctx\), where we define \(\BorPaths\kd \Rctx\) as the set of the subpaths of the borrow path of each mutable borrower owned by \(\Rctx\).
For this purpose, we introduce the \emph{validity} side condition \(\Rctx \kvalid\).
For example, it ensures that \(\Rctx\) is free of a linear variable-type binding \(\var \ccolul{\Bpaths}{\one} \Ty\) and of overlapping live lifetime tokens.
Also, to handle \emph{nested borrowing} properly (which is much trickier than it might look), the validity ensures that \(\Rctx\) owns all the mutably borrowed objects accessible by \(\Rctx\).

The key is the rule \ref{rule:assoc-linear} for \(\linear\), an intermediate state of \(\linearly\).
Here, we can capture some fragment \(\Mem\) of the memory and an arbitrary ghost resource context \(\Rctx\),\footnote{
  The condition \(\Rctx \kghost\) means that \(\Rctx\) consists only of ghost resources \(\gho\) after unfolding polymorphism.
}
under the \emph{global validity} side condition \(\Rctx \kgvalid\), stronger than the usual validity \(\valid\).
For example, the global validity ensures that \(\Rctx\) contains a borrowed object for each lender in \(\Rctx\).
The association rule for \(\linear\) itself (shown next) is straightforward.

A linearity token \(\Linearly\) captures \(\glinear\),
a reference \(\kRef \Ty\) captures a points-to token \(\loc \mMapsto \var\) and a body,
and live and dead lifetime tokens \(\Now^\lft\), \(\End^\lft\) capture \(\gnow\) and \(\gend\) of the observed history \(\Hist\).
The type \(\Done^{\seq{\rho}}\kc \Ty\) also captures \(\gnow\) of the lifetime paths \(\seq{\rho}\).
A mutable borrower is usually obtained from the token \(\gmut^{\seq{\bpath}}\) accompanied with some shared information.
A shared borrower is similarly obtained using \(\gshare^{\seq{\bpath}}\).
Also, borrowers whose lifetime has ended can be obtained through a skeleton \(\kSkel \Ty\).
A lender is obtained simply from the token \(\glend^\bid_{\var;\ka \Ty}\).

The content of a deposit \(\gdepo^{\seq{\bpath}}\), \(\ggdepo^\bid_{\var;\ka \Ty;\ka \Hist}\) is owned by some variable.
The tricky point is which variable should own each deposit.
We carefully design the validity condition (\ref{rule:assoc-var}) so that the borrower or the lender with a sufficient lifetime token should own the deposit, so that nested borrowers work properly in the presence of \(\parBO\).
For deposits, we introduce an auxiliary judgment \(\Rctx \kvdashul{\seq{\bpath}}{\Hist} \mterm \assoc \dterm \kccol \RTy\), asserting that the object \(\mterm \assoc \dterm\) deposited at \(\seq{\bpath}\) becomes an object of type \(\RTy\) after the updates of \(\Hist\).
Its rules are defined using an auxiliary judgment \(\Rctx \kvdashul{\seq{\bpath}}{\Hist} \dterm \leadsto \dtermB\), meaning that a term \(\dterm\) at \(\seq{\bpath}\) turns into \(\dtermB\) after the updates of \(\Hist\).

\subsection{Details}
\label{appx:assoc-details}

Here we present more details of the association system.

\paragraph{Basics}

The subtyping over extended types \(\RTy \subty \RTy'\) extends the usual subtyping \(\Ty \subty \Ty\) by the following rules:
{\small\[
  \frac{
    \Ty \subty \TyB
  }{
    \Done^\lft\kc \Ty \ka\subty\ka \Done^\lft\kc \TyB
  }
\hspace{3em}
  \frac{
    \Ty \subty \TyB
  }{
    \kSkel \Ty \ka\subty\ka \kSkel \TyB
  }
\]}%

We say an atomic resource \(\atom\) is shareable and write \(\atom \kshareable\) if \(\atom\) is of the form \(\var \ccolul{\Bpaths}{\many} \RTy\),
\(\var \teq \mval \assoc \dval\), \(\var \teq \dval\), \(\bid\),
\(\gend^\lftid_\Hist\), \(\glifetime^\bid_\lft\), \(\gborrow^{\seq{\bpath}}_{\var;\ka \Ty;\ka \seq{\bpathB}}\),
\(\gshare^{\seq{\bpath}}\), \(\loc \borMapsto{\seq{\bpath}} \var\), or \(\gskel^\var_\Ty\).

The resource context inclusion \(\Rctx \subty \RctxB\) is inductively defined by the following rules:
{\small\begin{gather*}
  \Rctx \subty \Rctx
\hspace{3em}
  \frac{
    \Rctx \subty \Rctx'
  \hspace{1.5em}
    \Rctx' \subty \Rctx''
  }{
    \Rctx \subty \Rctx''
  }
\hspace{3em}
  \frac{
    \multB \le \mult
  \hspace{1.5em}
    \Ty \subty \TyB
  }{
    \var \ccolul{\Bpaths}{\mult} \Ty,\kb \Rctx \kb\subty\kb
      \var \ccolul{\Bpaths}{\multB} \TyB,\kb \Rctx
  }
\\[.3em]
  \frac{
    \atom \kshareable
  }{
    \atom,\ka \Rctx \kb\subty\kb \Rctx
  }
\hspace{3em}
  \paren{\all{\anyvar} \Rctx},\kb \RctxB \kc\subty\kc
    \Rctx \sqbr{\any / \anyvar},\kb \RctxB
\hspace{3em}
  \frac{
    \text{\(\any\) is fresh in \(\RctxB\) and \(\RctxB'\)}
  \hspace{2em}
    \RctxB' \ka\subty\kb \Rctx,\ka \RctxB
  }{
    \RctxB' \ka\subty\kb
      \paren{\all{\anyvar} \Rctx},\ka \RctxB
  }
\end{gather*}}%

The multiplication \(\mult\ka \Rctx\) over a resource context is defined as follows:
{\small\begin{gather*}
  \one\ka \Rctx \ke\defeq\ke \Rctx
\hspace{3em}
  \mult\ka \paren{\var \ccoll{\multB} \Ty,\kb \Rctx} \ke\defeq\ke
    \var \ccoll{\mult \cdot \multB} \Ty,\kb \mult\ka \Rctx
\\[.3em]
  \mult\ka \paren{\paren{\all{\anyvar} \Rctx},\kb \RctxB} \ke\defeq\ke
    \paren{\all{\anyvar} \mult\ka \Rctx},\kb \mult\ka \RctxB
\hspace{3em}
  \frac{
    \atom \kshareable
  }{
    \mult\ka \paren{\atom,\ka \Rctx} \ke\defeq\ke
      \atom,\kb \mult\ka \Rctx
  }
\hspace{3em}
  \mult\ka \empseq \ke\defeq\ke \empseq
\end{gather*}}%

The addition \(\Rctx + \RctxB\) of resource contexts is defined as follows:
{\small\[
  \frac{
    \atom \kshareable
  }{
    \paren{\atom, \Rctx} \ka+\ka \paren{\atom, \RctxB} \ke\defeq\ke
      \atom,\ka \paren{\Rctx + \RctxB}
  }
\hspace{3em}
  \frac{
    \neg\kc \atom \kshareable
  }{
    \paren{\atom, \Rctx} \ka+\ka \RctxB \ke\defeq\ke
      \atom,\ka \paren{\Rctx + \RctxB}
  }
\hspace{3em}
  \empseq + \RctxB \ke\defeq\ke \RctxB
\]}%

The predicate \(\Rctx \kghost\) is inductively defined by the following rules:
{\small\[
  \Seq{\gho} \kghost
\hspace{3em}
  \frac{
    \Rctx,\ka \RctxB \kghost
  }{
    \paren{\all{\anyvar} \Rctx},\ka \RctxB \kghost
  }
\]}%

\paragraph{On histories}

The set of borrow paths \(\BorPaths\kd \Rctx\) for a resource context \(\Rctx\) is inductively defined by the following rules:
{\small\[
  \frac{
    \gmut^{\seq{\bpath}}\kc \Ty \ka\in\ka \Rctx
  }{
    \bpath_\idx\Seq{.\idxB} \ka\in\ka \BorPaths\kd \Rctx
  }
\hspace{3em}
  \frac{
    \bpath \ka\in\ka \BorPaths\kc \paren{\Rctx,\ka \RctxB}
  }{
    \bpath \ka\in\ka \BorPaths\kc \paren{\paren{\all{\anyvar} \Rctx},\ka \RctxB}
  }
\]}%

The predicate \(\Now^\lftpath\kc \Rctx \leadsto \Hist\), extracting the history \(\Hist\) for the lifetime path \(\lftpath\) from the resource context \(\Rctx\) is inductively defined by the following rules:
{\small\[
  \frac{
    \all{\idx < \num}\kd
    \Now^{\lftpath.\idx}\kc \Rctx_\idx \leadsto \Histpr_\idx
  }{
    \Now^\lftpath\ka \paren{\gnow^{\lftpath;\ka \num}_{\Hist;\ka \Bpaths},\ka \Seq{\Rctx}} \leadsto
    \Hist \hseq \bighpar_\idx \Histpr_\idx
  }
\hspace{3em}
  \frac{
    \Now^\lftpath\kc \Rctx \leadsto \Hist
  \hspace{1.5em}
    \text{\(\atom\) is not of the form \(\gnow^{\floor{\lftpath}\Seq{.\idx};\kb \num}_{\Histpr;\ka \Bpaths}\)}
  }{
    \Now^\lftpath\kc \paren{\atom, \Rctx} \leadsto \Hist
  }
\]}%
Importantly, the predicate is right-unique (i.e., unique in the history \(\Hist\)).
We also use the shorthand \(\Now^\lftpath\kc \Rctx\) for \(\ex{\Hist} \Now^\lftpath\kc \Rctx\).

The restriction \(\Hist \rvert_\bpath\) of a history \(\Hist\) with respect to a borrow path \(\bpath\) is defined as follows:
{\small\[
  \paren{\bpath\Seq{.\idx} \hstore \var,\ka \Hist} \ka\rvert_\bpath \ke\defeq\ke
    \bpath\Seq{.\idx} \hstore \var,\ka \paren{\Hist \rvert_\bpath}
\hspace{3em}
  \frac{
    \text{\(\bpathB\) is not of the form \(\bpath\Seq{.\idx}\)}
  }{
    \paren{\bpathB \hstore \var,\ka \Hist} \ka\rvert_\bpath \ke\defeq\ke
      \Hist \rvert_\bpath
  }
\hspace{3em}
  \emphist \rvert_\bpath \ke\defeq\ke \emphist
\]}

\paragraph{Validity}

The validity \(\Rctx \kvalid\) of a resource context is defined as the negation of the invalidity \(\Rctx \kinvalid\), inductively defined by the following rules:
{\small\begin{gather*}
  \frac{
    \text{\(\anyvar\) is fresh in \(\RctxB\)}
  \hspace{1.75em}
    \Rctx,\ka \RctxB \kinvalid
  }{
    \paren{\all{\anyvar} \Rctx},\ka \RctxB \kinvalid
  }
\hspace{3em}
  \var \ccolul{\Bpaths}{\one} \Ty,\ka \Rctx \kinvalid
\\[.3em]
  \frac{
    \dom \Hist \not\subseteq \Bpaths
  }{
    \gnow^{\lftpath;\ka \num}_{\Hist;\ka \Bpaths}\kb,\kb \Rctx \kinvalid
  }
\hspace{3em}
  \frac{
    \Bpaths['] \not\subseteq \Bpaths
  }{
    \gnow^{\lftpath;\ka \num}_{\Hist;\ka \Bpaths}\kb,\kb
    \gnow^{\lftpath\Seq{.\idx};\kb \numB}_{\Histpr;\ka \Bpaths[']}\kb,\kb
    \Rctx \kinvalid
  }
\hspace{3em}
  \frac{
    \idx \ne \idx'
  \hspace{1.5em}
    \Bpaths \ka\cap\ka \Bpaths['] \kb\ne\kb \empset
  }{
    \gnow^{\lftpath.\idx\Seq{.\idxB};\kb \num}_{\Hist;\ka \Bpaths},\kb
    \gnow^{\lftpath.\idx'\Seq{.\idxB'};\kb \numB}_{\Histpr;\ka \Bpaths[']},\kb
    \Rctx \kinvalid
  }
\\[.3em]
  \gnow^{\lftpath;\ka \num}_{\Hist;\ka \Bpaths}\kb,\kb
  \gend^{\floor{\lftpath}}_{\Histpr},\kb \Rctx \kinvalid
\hspace{3em}
  \gend^\lftid_\Hist,\kb \gend^\lftid_{\Histpr},\kb \Rctx \kinvalid
\\[.3em]
  \glifetime^\bid_\lft,\kb \glifetime^\bid_\lftB,\kb \Rctx \kinvalid
\hspace{3em}
  \gborrow^{\seq{\bpath}}_{\var;\ka \Ty},\kb
  \gborrow^{\seq{\bpath}}_{\varB;\ka \TyB},\kb \Rctx \kinvalid
\\[.3em]
  \gmut^{\seq{\bpath}},\kb \gmut^{\seq{\bpath}},\kb \Rctx \kinvalid
\hspace{3em}
  \gmut^{\seq{\bpath}},\kb \gshare^{\seq{\bpath}},\kb \Rctx \kinvalid
\hspace{3em}
  \glend^\bid_{\var;\ka \Ty}\kb,\kb
  \glend^\bid_{\varB;\ka \TyB}\kb,\kb \Rctx \kinvalid
\\[.3em]
  \gdepo^{\seq{\bpath}},\kb
  \gdepo^{\seq{\bpath}},\kb \Rctx \kinvalid
\hspace{3em}
  \ggdepo^\bid_{\var;\ka \Ty;\ka \Hist}\kb,\kb
  \ggdepo^\bid_{\varB;\ka \TyB;\ka \Histpr}\kb,\kb \Rctx \kinvalid
\hspace{3em}
  \ghole^{\seq{\bpath}},\kb
  \ghole^{\seq{\bpath}},\kb \Rctx \kinvalid
\\[.3em]
  \frac{
    \lft \ka\le\ka \al \lftid
  \hspace{1.5em}
    \Now^\bid\kc \Rctx \leadsto \Histpr
  \hspace{1.5em}
    \Hist \ka\ne\ka \Histpr \rvert_\bid
  }{
    \ggdepo^\bid_{\var;\ka \Ty;\ka \Hist}\kb,\kb
    \glifetime^\bid_\lft\ka,\kb \Rctx \kinvalid
  }
\hspace{3em}
  \frac{
    \lft \ka\le\ka \al \lftid
  \hspace{1.5em}
    \Hist \ka\ne\ka \Histpr \rvert_\bid
  }{
    \ggdepo^\bid_{\var;\ka \Ty;\ka \Hist}\kb,\kb
    \glifetime^\bid_\lft\ka,\kb \gend^\lftid_{\Histpr}\kb,\kb \Rctx \kinvalid
  }
\\[.3em]
  \frac{
    \lft \ka\le\ka \al \lftid
  \hspace{1.5em}
    \Now^\bid\kc \Rctx \leadsto \Histpr
  \hspace{1.5em}
    \Hist \rvert_{\bpath_\idx} \ka\ne\ka \Histpr \rvert_{\bpath_\idx}
  }{
    \grelend^{\seq{\bpath}}_\Hist\kb,\kb
    \glifetime^{\floor{\bpath_\idx}}_\lft\ka,\kb \Rctx \kinvalid
  }
\hspace{3em}
  \frac{
    \lft \ka\le\ka \al \lftid
  \hspace{1.5em}
    \Hist \rvert_{\bpath_\idx} \ka\ne\ka \Histpr \rvert_{\bpath_\idx}
  }{
    \grelend^{\seq{\bpath}}_\Hist\kb,\kb
    \glifetime^{\floor{\bpath_\idx}}_\lft\ka,\kb
    \gend^\lftid_{\Histpr}\kb,\kb \Rctx \kinvalid
  }
\\[.3em]
  \frac{
    \bigwedge_\idxB \al \floor{\lftpath_\idxB} \kb\le\kb
      \bigwedge_\idx \lft_\idx
  \hspace{2em}
    \gdepo^{\seq{\bpath}} \kb\notin\kb \Rctx
  }{
    \gmut^{\seq{\bpath}},\kb
    \gborrow^{\seq{\bpath}}_{\var;\ka \Ty},\kb
    \Seq{\glifetime^{\floor{\bpath}}_\lft},\kb
    \Seq{\gnow^{\lftpath;\ka \num}_{\Hist;\ka \Bpaths}}\kb,\kb \Rctx \kinvalid
  }
\\[.3em]
  \frac{
    \lft \ka\le\ka \al \lftid
  \hspace{2em}
    \all{\Hist}\kb \ggdepo^\bid_{\var;\ka \Ty;\ka \Hist} \kb\notin\kb \Rctx
  }{
    \glend^\bid_{\var;\ka \Ty},\kb \glifetime^\bid_\lft,\kb
    \gend^\lftid_{\Histpr}\kb,\kb \Rctx \kinvalid
  }
\hspace{3em}
  \frac{
    \lft \ka\le\ka \al \lftid
  \hspace{2em}
    \Now^\bid\kc \Rctx
  \hspace{2em}
    \all{\Hist}\kb \ggdepo^\bid_{\var;\ka \Ty;\ka \Hist} \kb\notin\kb \Rctx
  }{
    \glend^\bid_{\var;\ka \Ty},\kb \glifetime^\bid_\lft\kb,\kb \Rctx \kinvalid
  }
\\[.3em]
  \frac{
    \lft \ka\le\ka \al \lftid
  \hspace{1.5em}
    \floor{\bpath} \ka=\ka \bid
  \hspace{1.5em}
    \gdepo^{\seq{\bpathB}, \bpath} \kb\notin\kb \Rctx
  }{
    \glend^\bid_{\var;\ka \Ty},\kb \glifetime^\bid_\lft,\kb
    \gend^\lftid_{\Histpr}\kb,\kb
    \gborrow^{\seq{\bpathB}, \bpath}_{\var;\ka \Ty;\ka \seq{\bpathB'}}\kb,\kb
    \Rctx \kinvalid
  }
\hspace{3em}
  \frac{
    \lft \ka\le\ka \al \lftid
  \hspace{1.5em}
    \Now^\bid\kc \Rctx
  \hspace{1.5em}
    \floor{\bpath} \ka=\ka \bid
  \hspace{1.5em}
    \gdepo^{\seq{\bpathB}, \bpath} \kb\notin\kb \Rctx
  }{
    \glend^\bid_{\var;\ka \Ty},\kb \glifetime^\bid_\lft\kb,\kb
    \gborrow^{\seq{\bpathB}, \bpath}_{\var;\ka \Ty;\ka \seq{\bpathB'}}\kb,\kb
    \Rctx \kinvalid
  }
\end{gather*}}%

The global validity \(\Rctx \kgvalid\) of a resource context is defined as the negation of the global invalidity \(\Rctx \kginvalid\), inductively defined by the following rules:
{\small\begin{gather*}
  \frac{
    \text{\(\anyvar\) is fresh in \(\RctxB\)}
  \hspace{1.5em}
    \Rctx,\ka \RctxB \kginvalid
  }{
    \paren{\all{\anyvar} \Rctx},\ka \RctxB \kginvalid
  }
\hspace{3em}
  \frac{
    \Rctx \kinvalid
  }{
    \Rctx \kginvalid
  }
\hspace{3em}
  \frac{
    \gnow^{\lftpath;\ka \num}_{\Histpr;\ka \Bpaths} \kb\in\kb \Rctx
  \hspace{1.5em}
    \neg\kd \Now^{\floor{\lftpath}}\kc \Rctx
  }{
    \Rctx \kginvalid
  }
\\[.3em]
  \frac{
    \gdepo^{\seq{\bpath}} \kb\notin\kb \Rctx
  }{
    \gborrow^{\seq{\bpath}}_{\var;\ka \Ty;\ka \seq{\bpathB}}\kb,\kb \Rctx \kginvalid
  }
\hspace{3em}
  \frac{
    \ghole^{\seq{\bpath}} \kb\notin\kb \Rctx
  }{
    \gdepo^{\seq{\bpath}},\kb \Rctx \kginvalid
  }
\hspace{3em}
  \frac{
    \loc \bormapsto{\seq{\bpath}} \var \kb\notin\kb \Rctx
  }{
    \loc \borMapsto{\seq{\bpath}} \var,\kb \Rctx \kginvalid
  }
\\[.3em]
  \frac{
    \all{\Hist}\kb \ggdepo^\bid_{\var;\ka \Ty;\ka \Hist} \kb\notin\kb \Rctx
  }{
    \glend^\bid_{\var;\ka \Ty}\kb,\kb \Rctx \kginvalid
  }
\hspace{3em}
  \frac{
    \glend^\bid_{\var;\ka \Ty} \kb\notin\kb \Rctx
  }{
    \ggdepo^\bid_{\var;\ka \Ty;\ka \Hist}\kb,\kb \Rctx \kginvalid
  }
\hspace{3em}
  \frac{
    \all{\var, \Ty}\kd
    \glend^{\floor{\bpath_\idx}}_{\var;\ka \Ty} \kb\notin\kb \Rctx
  }{
    \gdepo^{\seq{\bpath}},\kb \Rctx \kginvalid
  }
\hspace{3em}
  \frac{
    \gvar^\var_\Ty \notin \Rctx
  }{
    \gskel^\var_\Ty,\kb \Rctx \kginvalid
  }
\end{gather*}}%

\paragraph{Association rules}

The omitted rules for the association judgment \(\Rctx \kvdash \mterm \assoc \dterm \kccol \RTy\) are as follows.
Here, we introduce an auxiliary judgment \(\any \assoc \ccol \seq{\RTy};\ka \RTyB\) for operators and monad constructors \(\any\).
Also, we write \(\Rctx \freshadd \RctxB\) for \(\Rctx, \RctxB\) with the following side condition:
for every item of the form \(\var \ccoll{\mult} \RTy\) in \(\Rctx\), \(\var\) is fresh in \(\RctxB\).
{\small\begin{gather*}
  \var \ccoll{\mult} \RTy \kvdash \var \kassoc \var \kkcol \RTy
\\[.3em]
  \frac{
    \all{\idx}\ke \Rctx_\idx \kvdash \term_\idx \assoc \term_\idx \kccol \Ty_\idx
  \hspace{1.5em}
    \Seq{\varB \ccoll{\one} \Ty_\idx} \freshadd \RctxB \kvdash \termB \assoc \termB \kccol \TyB
  }{
    \sum_\idx \Rctx_\idx \ka+\ka \RctxB \kvdash
      \klet \Seq{\varB \teq \term} \ktin \termB
      \kassoc \klet \Seq{\varB \teq \term} \ktin \termB
      \kkcol \TyB
  }
\hspace{2em}
  \frac{
    \all{\idx}\ke \Seq{\varB \ccoll{\many} \Ty_\idx},\kb \Rctx \kvdash
      \term_\idx \assoc \term_\idx \kccol \Ty_\idx
  \hspace{1.5em}
    \Seq{\varB \ccoll{\many} \Ty_\idx} \freshadd \RctxB \kvdash
      \termB \assoc \termB \kccol \TyB
  }{
    \many\ka \Rctx \ka+\ka \RctxB \kvdash
      \klet \Seq{\varB \teq \term} \ktin \termB
      \kassoc \klet \Seq{\varB \teq \term} \ktin \termB
      \kkcol \TyB
  }
\\[.3em]
  \frac{
    \var \ccoll{\mult} \Ty,\ka \Rctx \kvdash
      \term \kassoc \term \kkcol \TyB
  }{
    \Rctx \kvdash
      \tlamx{\var} \term \kassoc \tlamx{\var} \term \kkcol \Ty
      \to_\mult \TyB
  }
\hspace{2em}
  \frac{
    \Rctx \kvdash \term \assoc \term \kccol \Ty \to_\mult \TyB
  \hspace{1em}
    \RctxB \kvdash \termB \assoc \termB \kccol \Ty
  }{
    \Rctx \ka+\ka \mult\ka \RctxB \kvdash
      \term\kd \termB \kassoc \term\kd \termB \kkcol \TyB
  }
\hspace{2em}
  \frac{
    \var \ka\in\ka \dom \Rctx
  \hspace{1em}
    \Rctx \vdash \term \kassoc \term \ccol \Ty
  }{
    \Rctx \kvdash \kseq \var\kd \term \kassoc \kseq \var\kd \term \kccol \Ty
  }
\\[.3em]
  \vdashk \num \ccol \Int
\hspace{3em}
  \frac{
    \all{\idx}\ke \Rctx_\idx \kvdash \term_\idx \assoc \term_\idx
      \kccol \Ty_{\idxB, \idx} \sqbr{\Seq{\TyB \ky/ \Tvar}}
  \hspace{2em}
    \datatemplate
  }{
    \sum_\idx \mult_{\idxB, \idx}\ka \Rctx_\idx \kvdash
      \Con_\idxB\kc \seq{\term} \kassoc \Con_\idxB\kc \seq{\term} \kkcol \Tcon\kd \seq{\TyB}
  }
\\[.3em]
  \frac{
    \Rctx \vdash \term \assoc \term \ccol \Tcon\kd \seq{\TyB}
  \hspace{2em}
    \all{\idxB}\ke
      \Seq{\var \ccoll{\mult_{\idxB}}\kz \Ty_\idxB
        \sqbr{\vphantom{\hat A}\smash{\Seq{\TyB \ky/ \Tvar}}}} \freshadd
      \RctxB \kvdash \termB_\idxB \kccol \TyB'
  \hspace{2em}
    \datatemplate
  }{
    \Rctx \ka+\ka \RctxB \kvdash
      \kcase \term \kof \kcurly{ \Seq{\Con\kd \seq{\var} \tto \termB} }
      \kassoc \kcase \term \kof \kcurly{ \Seq{\Con\kd \seq{\var} \tto \termB} }
      \kccol \TyB'
  }
\\[.3em]
  \frac{
    \Rctx \vdash \term \assoc \term \ccol \Bcon^\lft\kb \paren{\Tcon\kd \seq{\TyB}}
  \hspace{2em}
    \all{\idxB}\ke
      \Seq{\var \ccoll{\mult_{\idxB}}\kz \Bcon^\lft\ka \paren{\Ty_\idxB
        \sqbr{\vphantom{\hat A}\smash{\Seq{\TyB \ky/ \Tvar}}}}} \freshadd
      \RctxB \kvdash \termB_\idxB \kccol \TyB'
  \hspace{2em}
    \datatemplate
  }{
    \Rctx \ka+\ka \RctxB \kvdash
      \kcase \term \kof \kcurly{ \Seq{\Con\kd \seq{\var} \tto \termB} }
      \kassoc \kcase \term \kof \kcurly{ \Seq{\Con\kd \seq{\var} \tto \termB} }
      \kccol \TyB'
  }
\\[.3em]
  \frac{
    \any \assoc \kccol \seq{\Ty};\ka \TyB
  \hspace{1.5em}
    \all{\idx}\kc \Rctx_\idx \vdash \term_\idx \assoc \term_\idx \ccol \Ty_\idx
  }{
    \sum_\idx \Rctx_\idx \kvdash \any\kd \seq{\term} \assoc \any\kd \seq{\term} \kccol \TyB
  }
\hspace{2em}
  \iop \assoc \kccol \Seq{\Int};\kd \Int
\hspace{2em}
  \irel \assoc \kccol \Seq{\Int};\kd \Bool
\hspace{2em}
  \tpar \assoc \kccol \Ty,\kb \TyB;\kd \paren{\Ty, \TyB}
\\[.3em]
  \frac{
    \text{\(\Ty\) is \(\Linearly\), \(\Now^\lft\) or \(\Mut^\lft\kc \TyB\)}
  }{
    \consume \assoc \kccol \Ty;\kd \paren{}
  }
\hspace{2.5em}
  \frac{
    \text{\(\Ty\) is \(\Int\), \(\End^\lft\) or \(\Share^\lft\kc \TyB\)}
  }{
    \move \assoc \kccol \Ty;\kd \Ur\kd \Ty
  }
\hspace{2.5em}
  \linearly \assoc \kccol \Linearly \lto \Ur\kd \Ty;\kd \Ur\kd \Ty
\\[.3em]
  \frac{
    \text{\(\Ty\) is \(\Linearly\), \(\kRef \TyB\), \(\Now^\lft\) or \(\Mut^\lft\kc \TyB\)}
  }{
    \withLinearly \assoc \kccol \Ty;\kd \paren{\Linearly,\ka \Ty}
  }
\hspace{3em}
  \newRef \assoc \kccol \Ty;\kd \kRef \Ty
\hspace{3em}
  \freeRef \assoc \kccol \kRef \Ty;\kd \Ty
\\[.3em]
  \newLifetime \assoc \kccol \paren{\all{\lftid}\ka \Now^{\ka\al\kd \lftid} \ky\lto\ky \Ty};\kd \Ty
\hspace{3em}
  \frac{
    \lft \ka\le\ka \al\kd \lftid
  }{
    \tendLifetime \assoc \kccol \Now^{\ka\al\kd \lftid};\kd
      \End^\lft
  }
\\[.3em]
  \frac{
    \lft \le \lftB
  \hspace{1.5em}
    \Ty \subty \TyB
  }{
    \borrow \assoc \kccol \Linearly,\kb \Ty;\kd
      \paren{\Mut^\lft\kc \Ty,\kb \Lend^\lftB\kc \TyB}
  }
\hspace{3em}
  \frac{
    \Ty \subty \TyB
  }{
    \share \assoc \kccol \Mut^\lft\kc \Ty;\kd \Ur\kd \paren{\Share^\lft\kc \TyB}
  }
\\[.3em]
  \frac{
    \text{\(\Ty\) is \(\Int\), \(\End^\lftB\) or \(\Share^\lftB\kc \TyB\)}
  \hspace{1.5em}
    \Ty \subty \Ty'
  }{
    \tcopy \assoc \kccol
      \Mut^\lft\kb \Ty;\kd \Ty'
  }
\hspace{3em}
  \frac{
    \text{\(\Ty\) is \(\Int\), \(\End^\lftB\) or \(\Share^\lftB\kc \TyB\)}
  }{
    \tcopy \assoc \kccol
      \Share^\lft\kb \Ty;\kd \Ty
  }
\\[.3em]
  \frac{
    \lft' \le \lft \lftand \lftB
  }{
    \joinMut \assoc \kccol
      \Mut^\lft\kb \paren{\Mut^\lftB\kc \Ty};\kd \Mut^{\lft'}\kb \Ty
  }
\hspace{3em}
  \frac{
    \lft' \le \lft \lftand \lftB
  \hspace{1.5em}
    \Ty \subty \TyB
  }{
    \joinMut \assoc \kccol
      \Share^\lft\kb \paren{\Mut^\lftB\kc \Ty};\kd \Share^{\lft'}\kb \TyB
  }
\\[.3em]
  \reclaim \kccol \Lend^\lft\kc \Ty,\kb \End^\lft;\kd \Ty
\hspace{3em}
  \execBO \assoc \kccol \Now^\lft,\kb \BO^\lft\kc \Ty;\kd
    \paren{\Now^\lft, \Ty}
\\[.3em]
  \pure \assoc \kccol \Ty;\kd \BO^\lft\kd \Ty
\hspace{3em}
  \paren{{\bind}} \assoc \kccol
    \BO^\lft\kc \Ty,\kb \Ty \ky\lto\ky \BO^\lft\kc \TyB;\kd
    \BO^\lft\kc \TyB
\\[.3em]
  \sexecBO \assoc \kccol
    \Now^\lft,\kb \BO^{\lft \lftand \lftB}\kc \TyB;\kd
    \BO^\lftB\kc \paren{\Now^\lft,\ka \TyB}
\hspace{3em}
  \parBO \assoc \kccol \BO^\lft\kc \Ty,\kb \BO^\lft\kc \TyB;\kd
    \BO^\lft\kb \paren{\Ty, \TyB}
\\[.3em]
  \frac{
    \lftB \le \lft
  }{
    \deref \assoc \kccol \Bcon^\lft\kb \paren{\kRef \Ty};\kd
      \BO^\lftB\kb \paren{\Bcon^\lft\kc \Ty}
  }
\hspace{3em}
  \frac{
    \lftB \le \lft
  }{\begin{aligned} &
    \updateRef \assoc \kccol
      \Ty \ky\lto\ky \BO^\lftB\ka \paren{\TyB, \Ty},\kc
      \Mut^\lft\ka \paren{\kRef \Ty};
  \\[-.5em] & \hspace{14em}
      \BO^\lftB\ka \paren{\TyB,\ka \Mut^\lft\ka \paren{\kRef \Ty}}
  \end{aligned}}
\\[.3em]
  \execBO^\post \assoc \kccol
    \Done^{\lftpath}\kc \Ty;\kd \paren{\Now^{\ka\al \floor{\lft}},\ka \Ty}
\hspace{3em}
  \frac{
    \lft \ka\le\ka \al \floor{\lftpath}
  }{
    \paren{{\bind^\post}} \assoc \kccol
      \Done^{\lftpath}\kc \Ty,\kb
      \Ty \nklto \BO^{\ka\al \floor{\lftpath}}\kc \TyB;\kd
      \BO^\lft\kc \TyB
  }
\\[.3em]
  \frac{
    \all{\idx}\kc \Bpaths['] \subseteq \Bpaths_\idx
  }{
    \Seq{\gnow^{\lftpath;\ka 0}_{\Hist;\ka \Bpaths}}\kb,\kb
    \varNow \ccolul{\Bpaths['']}{\one}\kx \Now^\lft,\kb
    \varBo \ka\ccolul{\Bpaths[']}{\one}\kx
      \BO^{\ka\bigwedge \Seq{\ka \al \floor{\lftpath}} \ka\lftand\ka \lft}\kd \Ty \kkvdash
      \sexecBO^\pre\kc \varNow\kd \varBo
      \kkassoc \sexecBO^\pre_\Hist\kc \varNow\kd \varBo
      \kkkcol \Done^{\seq{\lftpath}}\kc \Ty
  }\
\\[.3em]
  \Seq{\gnow^{\lftpath;\ka 1}_{\Hist;\ka \Bpaths}}\kb,\kb
  \gnow^{\lftpath';\ka 1}_{\Histpr;\ka \Bpaths[']}\kb,\kb
  \varRes \ccolul{\Bpaths['']}{\one}\kx \Done^{\Seq{\lftpath.0},\ka \lftpath'\ky.0}\kb \Ty \kkvdash
    \sexecBO^\post\kc \varRes
    \kkassoc \sexecBO^\post_{\Hist, \Histpr}\kc \varRes
    \kkkcol \Done^{\seq{\lftpath}}\kc \Ty
\\[.3em]
  \frac{
    \all{\idx}\kc \Bpaths['], \Bpaths[''] \kb\subseteq\kb \Bpaths_\idx
  }{\begin{aligned} &
    \Seq{\gnow^{\lftpath;\ka 2}_{\Hist;\ka \Bpaths}}\kb,\kb
    \varRes_0 \ccolul{\Bpaths[']}{\one}\kx \Done^{\Seq{\lftpath.0},\ka \lftpath'\ky.0}\kc \Ty,\kb
    \varRes_1 \ccolul{\Bpaths['']}{\one}\kx \Done^{\Seq{\lftpath.1},\ka \lftpath'\ky.1}\kc \TyB \kkvdash {}
  \\[-.3em] & \hspace{7em}
      \parBO^\post\kc \varRes_0\kd \varRes_1
      \kkassoc \parBO^\post_\Hist\kc \varRes_0\kd \varRes_1
      \kkkcol \Done^{\seq{\lftpath}}\kc \paren{\Ty, \TyB}
  \end{aligned}}
\\[.3em]
  \frac{
    \bigwedge \Seq{\al \floor{\lftpath}} \kb\le\kb \lft
  }{
    \Seq{\gnow^{\lftpath;\ka 0}_{\Hist;\ka \Bpaths}}\kb,\kb
    \varRef \ccolul{\Bpaths[']}{\one}\ky \Bcon^\lft\kb \paren{\kRef \Ty} \kkvdash
      \deref^\post\kc \varRef
      \kkassoc \deref^\post_\Hist\kc \varRef
      \kkkcol \Done^{\seq{\lftpath}}\kb \paren{\Bcon^\lft\kc \Ty}
  }
\\[.3em]
  \frac{
    \bigwedge \Seq{\al \floor{\lftpath}} \kb\le\kb \lft
  }{\begin{aligned} &
    \Seq{\gnow^{\lftpath;\ka 0}_{\Hist;\ka \Bpaths}}\kb,\kb
    \varKo \ka\ccolul{\Bpaths['']}{\one} \Ty \lto
      \BO^{\ka \Seq{\al \floor{\lftpath}}}\kc \paren{\TyB, \Ty},\kb
    \varRef \ccolul{\Bpaths[']}{\one}\kx \Mut^\lft\kb \paren{\kRef \Ty} \kkvdash {}
  \\[-.4em] & \hspace{7em}
      \updateRef^\pre\kc \varKo\kd \varRef
      \kkassoc \updateRef^\pre_\Hist\kc \varKo\kd \varRef
      \kkkcol \Done^{\seq{\lftpath}}\kb \paren{\TyB,\ka \Mut^\lft\kb \paren{\kRef \Ty}}
  \end{aligned}}
\\[.3em]
  \frac{
    \all{\idx}\ka \paren[\big]{\kb
      \bigwedge \Seq{\al \floor{\lftpath}} \kb\le\kb \lft_\idx
    \kb}
  \hspace{2em}
    \all{\idx}\ka \paren[\big]{\kb \lftB \le \lft_\idx \kb}
  \hspace{2em}
    \all{\idx}\ka \paren[\big]{\kb \Ty'_\idx \subty \kRef \Ty \kb}
  \hspace{2em}
    \all{\idx}\ka \paren[\big]{\kb \kRef \Ty \subty \Ty'_\idx \kb}
  }{\begin{aligned} & \textstyle
    \var \teq \kRef \loc \assoc \kRef \varB,\kb
    \loc \mMapsto \varB,\kb
    \gmut^{\seq{\bpath}},\kb \gdepo^{\seq{\bpath}},\kb
    \Seq{\glifetime^{\floor{\bpath}}_\lft},\kb
    \gborrow^{\seq{\bpath}}_{\var;\ka \Ty}\kb,\kb
    \varRes \ccolul{\Bpaths}{\one} \Done^{\seq{\lftpath}}\kc \paren{\TyB, \Ty} \kkvdash {}
  \\[-.4em] & \hspace{15em}
      \updateRef^\prepost_\loc\kc \varRes
      \kkassoc \updateRef^\prepost_{\seq{\bpath}}\kc \varRes
      \kkkcol \Done^{\seq{\lftpath}}\kb \paren{\TyB,\ka \Mut^\lftB\kb \paren{\kRef \Ty}}
  \end{aligned}}
\\[.3em]
  \frac{
    \all{\idx}\ka \paren[\big]{\kb
      \bigwedge \Seq{\al \floor{\lftpath}} \kb\le\kb \lft_\idx
    \kb}
  \hspace{2em}
    \all{\idx}\ka \paren[\big]{\kb \lftB \le \lft_\idx \kb}
  \hspace{2em}
    \all{\idx}\ka \paren[\big]{\kb \Ty'_\idx \subty \kRef \Ty \kb}
  \hspace{2em}
    \all{\idx}\ka \paren[\big]{\kb \kRef \Ty \subty \Ty'_\idx \kb}
  }{\begin{aligned} & \textstyle
    \var \teq \kRef \loc \assoc \kRef \varB,\kb
    \loc \mMapsto \varB,\kb
    \gmut^{\seq{\bpath}},\kb \gdepo^{\seq{\bpath}},\kb
    \Seq{\glifetime^{\floor{\bpath}}_\lft},\kb
    \gborrow^{\seq{\bpath}}_{\var;\ka \Ty}\kb,\kb
    \Seq{\gnow^{\lftpath;\ka 0}_{\Hist;\ka \Bpaths}}\kb,\kb
    \varRes \ccolul{\Bpaths[']}{\one}\kx \paren{\TyB, \Ty} \kkvdash {}
  \\[-.4em] & \hspace{16em}
      \updateRef^\post_\loc\kc \varPr
      \kkassoc \updateRef^\post_{\seq{\bpath}; \Hist}\kc \varPr
      \kkkcol \Done^{\seq{\lftpath}}\kb \paren{\TyB,\ka \Mut^\lftB\kb \paren{\kRef \Ty}}
  \end{aligned}}
\\[.3em]
  \frac{
    \lftB \ka\le\ka \bigwedge_\idx \lft_\idx
  \hspace{1.5em}
    \Ty \subty \TyB
  \hspace{1.5em}
    \TyB \subty \Ty
  }{\begin{aligned} &
    \gmut^{\seq{\bpath}},\kb \Seq{\glifetime^{\floor{\bpath}}_\lft},\kb
    \gborrow^{\seq{\bpath}}_{\var;\ka \Ty}\kb,\kb
    \var \teq \mval \assoc \dval \kkvdash {}
  \\[-.4em] & \hspace{10em}
      \mval \kassoc \Mut^{\seq{\bpath}}\kc \dval
      \kkcol \Mut^\lftB\kc \TyB
  \end{aligned}}
\hspace{3em}
  \frac{
    \lftB \ka\le\ka \bigwedge_\idx \lft_\idx
  \hspace{1.5em}
    \Ty \subty \TyB
  }{\begin{aligned} &
    \gshare^{\seq{\bpath}},\kb \Seq{\glifetime^{\floor{\bpath}}_\lft},\kb
    \gborrow^{\seq{\bpath}}_{\var;\ka \Ty}\kb,\kb
    \var \teq \mval \assoc \dval \kkvdash {}
  \\[-.4em] & \hspace{10.5em}
      \mval \kassoc \Share\kd \dval
      \kkcol \Share^\lftB\kc \TyB
  \end{aligned}}
\\[.3em]
  \frac{
    \all{\idx}\ke \Rctx_\idx \kkvdash
      \var_\idx \kassoc \Share\kd \var_\idx
      \kkcol \Share^\lft\kb \paren{\Ty_{\idxB, \idx} \sqbr{\Seq{\TyB \ky/ \Tvar}}}
  \hspace{2em}
    \datatemplate
  }{
    \sum_\idx \mult_{\idxB, \idx}\ka \Rctx_\idx \kkvdash
      \Con_\idxB\kc \seq{\var} \kassoc \Share\kd \Con_\idxB\kc \seq{\var} \kkcol \Share^\lft\kb \paren{\Tcon\kd \seq{\TyB}}
  }
\\[.3em]
  \frac{
    \Rctx \kkvdash \var \kassoc \Share\kd \var \kkcol \Share^\lft\kc \Ty
  }{
    \loc \borMapsto{\seq{\bpath}} \var \ka+\ka \Rctx \kkvdash
      \kRef \loc \kassoc \Share\kc \paren{\kRef \var}
      \kkcol \Share^\lft\kb \paren{\kRef \Ty}
  }
\\[.3em]
  \frac{
    \var \ccolul{\Bpaths}{\mult} \Ty,\kb
    \Rctx \kvdash \mterm \assoc \dterm \kccol \RTyB
  }{
    \gvar^\var_{\Seq{\all{\anyvar}\kw}\kc \Ty}\kb,\kb
    \var \ccolul{\Bpaths}{\mult} \Ty,\kb
    \Rctx \kvdash \mterm \assoc \dterm \kccol \RTyB
  }
\hspace{3em}
  \frac{
    \Ty \subty \TyB
  }{
    \gskel^\var_\Ty \kvdash
      \var \assoc \var \kccol \kSkel \TyB
  }
\hspace{3em}
  \frac{
    \Ty \subty \TyB
  }{
    \gskel^\var_\Ty,\kb \var \teq \mval \assoc \dval \kkvdash
      \mval \kassoc \dval \kkcol \kSkel \TyB
  }
\\[.3em]
  \var \ccolul{\Bpaths}{\many} \Ty \kvdash
    \var \assoc \var \kccol \kSkel \Ty
\hspace{3em}
  \var \ccolul{\Bpaths}{\many} \Ty,\kb
  \var \teq \mval \assoc \dval \kkvdash
    \mval \kassoc \dval \kkcol \kSkel \Ty
\\[.3em]
  \frac{
    \Rctx \kvdash \mterm \assoc \dterm \kccol \all{\anyvar} \kSkel \Ty
  }{
    \Rctx \kvdash \mterm \assoc \dterm \kccol \kSkel \paren{\all{\anyvar} \Ty}
  }
\hspace{3em}
  \frac{
    \all{\idx}\ke \Rctx_\idx \kvdash \var_\idx \assoc \var_\idx
      \kccol \Skel\kc \paren{\Ty_{\idxB, \idx} \sqbr{\Seq{\TyB \ky/ \Tvar}}}
  \hspace{2em}
    \datatemplate
  }{
    \sum_\idx \mult_{\idxB, \idx}\ka \Rctx_\idx \kvdash
      \Con_\idxB\kc \seq{\var} \kassoc \Con_\idxB\kc \seq{\var} \kkcol \Skel\kc \paren{\Tcon\kd \seq{\TyB}}
  }
\\[.3em]
  \frac{
    \text{\(\Ty\) is \(\TyB \ky\to_\mult\ky \TyB'\), \(\Linearly\),
      \(\kRef \TyB\), \(\Now^\lft\), \(\Lend^\lft\kc \TyB\), or \(\BO^\lft\kc \TyB\)}
  }{
    \kvdash \mval \assoc \dval \kccol \kSkel \Ty
  }
\\[.3em]
  \frac{
    \Rctx \kvdash \mval \assoc \dval \kccol \kSkel \Ty
  }{
    \Rctx \kkvdash \mval \kassoc \Mut^{\seq{\bpath}}\kc \dval
      \kkcol \Skel\kc \paren{\Mut^\lft\kc \Ty}
  }
\hspace{2.5em}
  \frac{
    \Rctx \kvdash \mval \assoc \dval \kccol \kSkel \Ty
  }{
    \Rctx \kkvdash \mval \assoc \Share\kd \dval
      \kkcol \Skel\kc \paren{\Share^\lft\kc \Ty}
  }
\\[.3em]
  \frac{
    \seq{\bpathB} \ne \empseq
  \hspace{1.5em}
    \Rctx \kvdashul{\seq{\bpathB}, \seq{\bpath}}{\Hist} \var \kccol \Ty
  \hspace{1.5em}
    \RctxB \kvdash \mterm \assoc \dterm \kkcol \RTyB
  }{
    \gdepo^{\seq{\bpath}},\kb
    \gborrow^{\seq{\bpath}}_{\var;\ka \Ty;\ka \seq{\bpathB}} \ky+\ka
    \grelend^{\seq{\bpathB}, \seq{\bpath}}_\Hist \ka+\ka
    \Rctx \ka+\ka \RctxB \kkvdash
      \mterm \assoc \dterm \kkcol \RTyB
  }
\\[.3em]
  \Seq{\var \teq \dval} \kkvdashul{\seq{\bpath}}{\empseq}
    \dterm \kleadsto \dterm
\hspace{3em}
  \frac{
    \Rctx \kvdashul{\seq{\bpath}}{\Hist} \dval \kleadsto \dvalB
  }{
    \var \teq \dval \ka+\ka \varB \teq \dvalB \ka+\ka \Rctx \kkvdashul{\seq{\bpath}}{\Hist}
      \var \kleadsto \varB
  }
\hspace{3em}
  \frac{
    \Rctx \kvdashul{\Seq{\bpath.0}}{\Hist} \varB \leadsto \varC
  }{
    \Rctx \kkvdashul{\seq{\bpath}}{\Seq{\bpath \hstore \varB},\ka \Hist}
      \DBctx\kc \paren{\kRef \var} \kleadsto \DBctx\kc \paren{\kRef \varC}
  }
\\[.3em]
  \frac{
    \all{\idx}\kc \bpath_\idx \ka\notin\ka \dom \Hist
  \hspace{1.5em}
    \Hist \ne \emphist
  \hspace{1.5em}
    \Rctx \kvdashul{\Seq{\bpath.0}}{\Hist} \var \leadsto \varB
  }{
    \Rctx \kkvdashul{\seq{\bpath}}{\Hist}
      \DBctx\kc \paren{\kRef \var} \kleadsto \DBctx\kc \paren{\kRef \varB}
  }
\hspace{3em}
  \frac{
    \Hist \ne \emphist
  \hspace{1.5em}
    \all{\idx}\kc \Rctx \kvdashul{\Seq{\bpath.\idx}}{\Hist} \var_\idx \leadsto \varB_\idx
  }{
    \Rctx \kvdashul{\seq{\bpath}}{\Hist}
      \DBctx\kc \paren{\Con\kd \seq{\var}} \leadsto \DBctx\kc \paren{\Con\kd \seq{\varB}}
  }
\\[.3em]
  \frac{
    \RctxB \ka\subty\ka \Rctx
  \hspace{1.5em}
    \Rctx \kvdashul{\seq{\bpath}}{\Hist}
      \mterm \assoc \dterm \kccol \RTy
  }{
    \RctxB \kvdashul{\seq{\bpath}}{\Hist}
      \mterm \assoc \dterm \kccol \RTy
  }
\hspace{3em}
  \frac{
    \text{\(\anyvar\) is fresh in \(\Rctx\)}
  \hspace{1.5em}
    \Rctx \kvdashul{\seq{\bpath}}{\Hist}
      \mterm \assoc \dterm \kccol \Ty
  }{
    \Rctx \kkvdashul{\seq{\bpath}}{\Hist}
      \mterm \kassoc \dterm \kkcol \all{\anyvar} \Ty
  }
\\[.3em]
  \frac{
    \Hist \ne \emphist
  \hspace{1.5em}
    \all{\idx}\kc \Rctx_\idx \kvdashul{\Seq{\bpath.\idx}}{\Hist_\idx}
      \var_\idx \assoc \var_\idx
      \kccol \Ty_{\idxB, \idx}\sqbr{\Seq{\TyB \ky/ \Tvar}}
  \hspace{1.5em}
    \datatemplate
  }{
    \sum_\idx \Rctx_\idx \kkvdashul{\seq{\bpath}}{\Hist}
      \Con_\idxB\kc \seq{\var} \kassoc \Con_\idxB\kc \seq{\var}
      \kkcol \Tcon\kd \seq{\TyB}
  }
\\[.3em]
  \frac{
    \lftB \ka\le\ka \bigwedge_\idx \lft_\idx
  \hspace{1.5em}
    \Ty \subty \TyB
  \hspace{1.5em}
    \TyB \subty \Ty
  \hspace{1.5em}
    \Rctx \kvdashul{\seq{\bpath}}{\Hist} \dval \leadsto \dvalB
  }{
    \gmut^{\seq{\bpathB}},\kb \Seq{\glifetime^{\floor{\bpathB}}_\lft},\kb
    \gborrow^{\seq{\bpathB}}_{\var;\ka \Ty;\ka \seq{\bpath}},\kb
    \var \teq \mval \assoc \dvalB,\kb \Rctx \kkvdashul{\seq{\bpath}}{\Hist}
      \mval \kassoc \Mut^{\seq{\bpathB}}\kc \dval
      \kkcol \Mut^\lftB\kc \TyB
  }
\\[.3em]
  \frac{
    \lft \le \al \lftid
  \hspace{2em}
    \Rctx \kvdashul{\seq{\bpath}}{\Hist}
      \mterm \assoc \dterm \kccol \kSkel \Ty
  }{
    \gend^\lftid_\Hist,\ka \Rctx \kkvdashul{\seq{\bpath}}{\Hist}
      \mterm \kassoc \Mut^{\seq{\bpathB}}\kc \dterm
      \kkcol \Mut^\lft\kc \Ty
  }
\hspace{3em}
  \frac{
    \Rctx \kvdashul{\Seq{\bpath.0}}{\Hist} \var \leadsto \varB
  }{
    \Rctx \kkvdashul{\seq{\bpath}}{\Hist}
      \kRef \loc \kassoc \kRef \var
      \kkcol \Skel\kc \paren{\kRef \Ty}
  }
\\[.3em]
  \frac{
    \Hist \ne \emphist
  \hspace{1.5em}
    \all{\idx}\kc \Rctx_\idx \kvdashul{\Seq{\bpath.\idx}}{\Hist_\idx}
      \var_\idx \kccol \Skel\kc \paren{\Ty_{\idxB, \idx}\sqbr{\Seq{\TyB \ky/ \Tvar}}}
  \hspace{1.5em}
    \datatemplate
  }{
    \sum_\idx \Rctx_\idx \kkvdashul{\seq{\bpath}}{\Hist}
      \Con_\idxB\kc \seq{\var} \kassoc \Con_\idxB\kc \seq{\var}
      \kkcol \Skel\kc \paren{\Tcon\kd \seq{\TyB}}
  }
\end{gather*}}%

\section{Metatheory}
\label{appx:metatheory}

We present the omitted details of our metatheory introduced in \cref{sect:metatheory-metatheory}.

\subsection{Strong Confluence of Denotational Operational Semantics}
\label{appx:metatheory-confl}

\begin{proof}[Proof of \cref{thm:dos-confl}]
  We say that the target variable of a (derivation of a) reduction \(\DEnv; \var \to \DEnv'; \var\) is \(\varB\) if, other than \ref{rule:dred-ktx}, the reduction applies a rule whose conclusion is of the form \(\DEnv; \varB \to \DEnv'; \varB\).

  Assume that the target variables of \(\DCfg \to \DCfg_1\) and \(\DCfg \to \DCfg_2\) are \(\var\) and \(\varB\), respectively.
  If \(\var = \varB\), then \(\DCfg_1 = \DCfg_2\) holds, which contradicts the assumption.
  So we have \(\var \ne \varB\).
  By straightforward case analysis,
  we can derive the reductions \(\DCfg_1 \to \DCfg_+\) and \(\DCfg_2 \to \DCfg_+\) for some \(\DCfg_+\),
  choosing \(\varB\) as the target variable of the former reduction and \(\var\) for the latter.
  A notable case is where one reduction executes \(\reclaim\) \ref{rule:dred-reclaim};
  the confluence holds in this case, because the structure of restoration by history \cref{fig:red-borrow-restore} does not get affected by the reduction of any variable.
\end{proof}

\subsection{Basics of the Association System}
\label{appx:metatheory-assoc}

Technically, the association judgment does not have a rule for weakening the resulting type.
Instead, it is carefully designed so that it is closed under subtyping:
\begin{lemma}[\(\assoc\) is closed under subtyping]
\label{lem:assoc-subty}
  If \(\kc\RTy \subty \RTyB\), then \(\Rctx \kvdash \mterm \assoc \dterm \kccol \RTy\) implies \(\Rctx \kvdash \mterm \assoc \dterm \kccol \RTyB\).
\end{lemma}
\begin{proof}
  By straightforward induction over the derivation of subtyping and association.
\end{proof}

For \cref{thm:typed-assoc}, we introduce the following lemma:
\begin{lemma}\label{lem:typed-assoc-base}
  If a term \(\term\) satisfies \(\kb\Tctx \vdash \term \ccol \Ty\), then we have \(\Tctx \vdash \term \assoc \term \ccol \Ty\).
\end{lemma}
\begin{proof}
  By straightforward induction over the typing derivation.
\end{proof}

\begin{proof}[Proof of \cref{thm:typed-assoc}]
  Immediately from \cref{lem:typed-assoc-base}.
\end{proof}

\begin{proof}[Proof of \cref{thm:safe-leakfree}]
  Straightforward by definition.
\end{proof}

\subsection{Progress and Bisimulation}
\label{appx:metatheory-thms}

\begin{proof}[Proof sketch for \cref{thm:progress}]
  By careful case analysis.
  In particular, we maintain an invariant on histories to ensure progress for lender reclamation \ref{rule:dred-reclaim} in denotational operational semantics.
\end{proof}

\begin{proof}[Proof sketch for \cref{thm:bisim}]
  By careful case analysis.
  There are a few tricky factors.
  First, the variable that owns each local or global deposit for a borrow can drastically change, especially when a lifetime ends by \(\tendLifetime\);
  still, the variable always moves to either an ancestor or new children (possibly splitting the deposit), which is crucial for maintaining the disjointness invariant for \(\parBO\).
  Also, when copying a value by unrestrictedness or as a skeleton from a variable polymorphically typed by \ref{rule:assoc-all}, we re-type the variable to make it free of \ref{rule:assoc-all}, polymorphically re-typing the variables that the variable depends on.
\end{proof}

\fi

\end{document}